\pdfoutput=1

\documentclass[english,11pt]{article}
\include{macros}

\begin{document}

\title{How Many Human Survey Respondents is a Large Language Model Worth? An Uncertainty Quantification Perspective}

\author{Chengpiao Huang\thanks{Department of IEOR, Columbia University. Email: \texttt{chengpiao.huang@columbia.edu}.}
	\and Yuhang Wu\thanks{Decision, Risk, and Operations Division, Columbia Business School. Email: \texttt{yuhang.wu@columbia.edu}}
	\and Kaizheng Wang\thanks{Department of IEOR and Data Science Institute, Columbia University. Email: \texttt{kaizheng.wang@columbia.edu}.}
}

\date{This version: \today}

\maketitle

\begin{abstract}
Large language models (LLMs) are increasingly used to simulate survey responses, but synthetic data can be misaligned with the human population, leading to unreliable inference. We develop a general framework that converts LLM-simulated responses into reliable confidence sets for population parameters of human responses, quantifying the uncertainty induced by the human-LLM misalignment. The key design choice is the number of simulated responses: too many produce overly narrow sets with poor coverage, while too few yield overly wide and uninformative sets dominated by stochastic noise. We propose a data-driven approach that adaptively selects the simulation sample size to achieve nominal average-case coverage, regardless of the LLM's simulation fidelity or the confidence set construction procedure. The selected sample size is further shown to reflect the effective human population size that the LLM can represent, providing a quantitative measure of its simulation fidelity. Experiments on real survey datasets reveal heterogeneous simulation fidelity across different LLMs and domains.

\end{abstract}

\section{Introduction}

Large language models (LLMs) have demonstrated remarkable capabilities in mimicking human behaviors. Recent studies have leveraged LLMs to simulate human responses in various domains, including economic and social science experiments \citep{AAK23, Hor23, CLS23, BCD24, HYZ24, YLW24, ZHS24}, market research \citep{BIN23, GTo23, GSi24, WZZ24}, education \citep{MOL23, ZMT23, LWa24}, and so on. The typical simulation procedure consists in prompting an LLM with a real or fictional persona as well as a survey question, and collecting the LLM's responses. Compared to traditional survey methods that recruit and query real people, LLM simulations offer significant advantages in terms of time and cost efficiency, enabling the generation of large-scale synthetic responses with minimal effort.

However, a growing body of evidence suggests that LLMs are not perfectly aligned with the human population and, in some cases, the misalignment can be substantial \citep{AAK23, SDL23, GLB25}. This raises critical concerns about the reliability of insights derived from LLM-generated data. It remains a challenge how to properly simulate human responses using LLMs and how to account for their imperfections when using the simulated samples to make inference about the true human population.

To make the problem concrete, consider a surveyor who wishes to estimate the average human response to a Likert-scale survey question. The surveyor may prompt an LLM to simulate human respondents and use the synthetic responses in place of human data. However, due to LLM misalignment, na\"{i}vely generating a large number of synthetic responses can lead to an estimate that appears precise but is systematically biased. This naturally motivates the following two questions: 
\begin{center}
\emph{How many synthetic samples should be generated?}
\end{center}
\begin{center}
\emph{How can we quantify the error of estimates based on synthetic data?}
\end{center}

We propose to address these questions through the lens of \emph{uncertainty quantification}. Specifically, we seek to use LLM-generated data to construct a confidence interval (or confidence set) for a population parameter (e.g., the mean) of human responses. The resulting confidence interval provides a statistically valid range of plausible values for the true human population parameter, and quantifies the uncertainty of the synthetic estimate induced by the real-synthetic gap. The key to constructing such a valid confidence interval is choosing how many synthetic samples to use, so that the interval accounts for both simulation noise and the real-synthetic gap. On one hand, generating too many samples can make the confidence interval concentrate too tightly around the synthetic population parameter, failing to cover the true human population parameter. On the other hand, generating too few samples yields overly wide confidence intervals dominated by stochastic noise. The optimal sample size depends on the discrepancy between the real and synthetic response distributions, which is unknown in practice. This necessitates a data-driven approach to determining an appropriate number of simulated responses.

\paragraph{Contributions.} In this paper, we propose a general framework that selects a simulation sample size $\widehat{k}$ for reliable uncertainty quantification, and give this sample size interpretable operational meanings. In particular, the sample size $\widehat{k}$ reveals three operational insights:
\begin{enumerate}
\item \textbf{Practical simulation guidance.} It specifies how many synthetic samples should be generated for future simulation tasks. This yields statistically reliable uncertainty quantification while avoiding the false precision that can arise from generating too many synthetic samples.
\item \textbf{Intuitive alignment metric.} The width of the resulting confidence interval, which is determined by $\widehat{k}$, serves as an intuitive measure of the misalignment between the LLM and the human population. A wide interval indicates high uncertainty and a large human-LLM discrepancy.
\item \textbf{Effective human sample size.} For survey questions with binary responses, we further establish that $\widehat{k}$ serves as a novel measure of the LLM's simulation fidelity, providing a direct answer to the question: \emph{``How many human samples is the LLM worth?''}
It quantifies the size of the human population that the LLM effectively represents, as if the model were a \emph{Mechanical Turk} made up of $\widehat{k}$ real human agents that respond to user queries.
A larger $\widehat{k}$ indicates that the LLM captures richer information of the human population and thus has a higher simulation fidelity. We visualize this interpretation in \myCref{fig-LLMTurk}.
\end{enumerate}

\begin{figure}[h]
\FIGURE
{\includegraphics[scale=0.15]{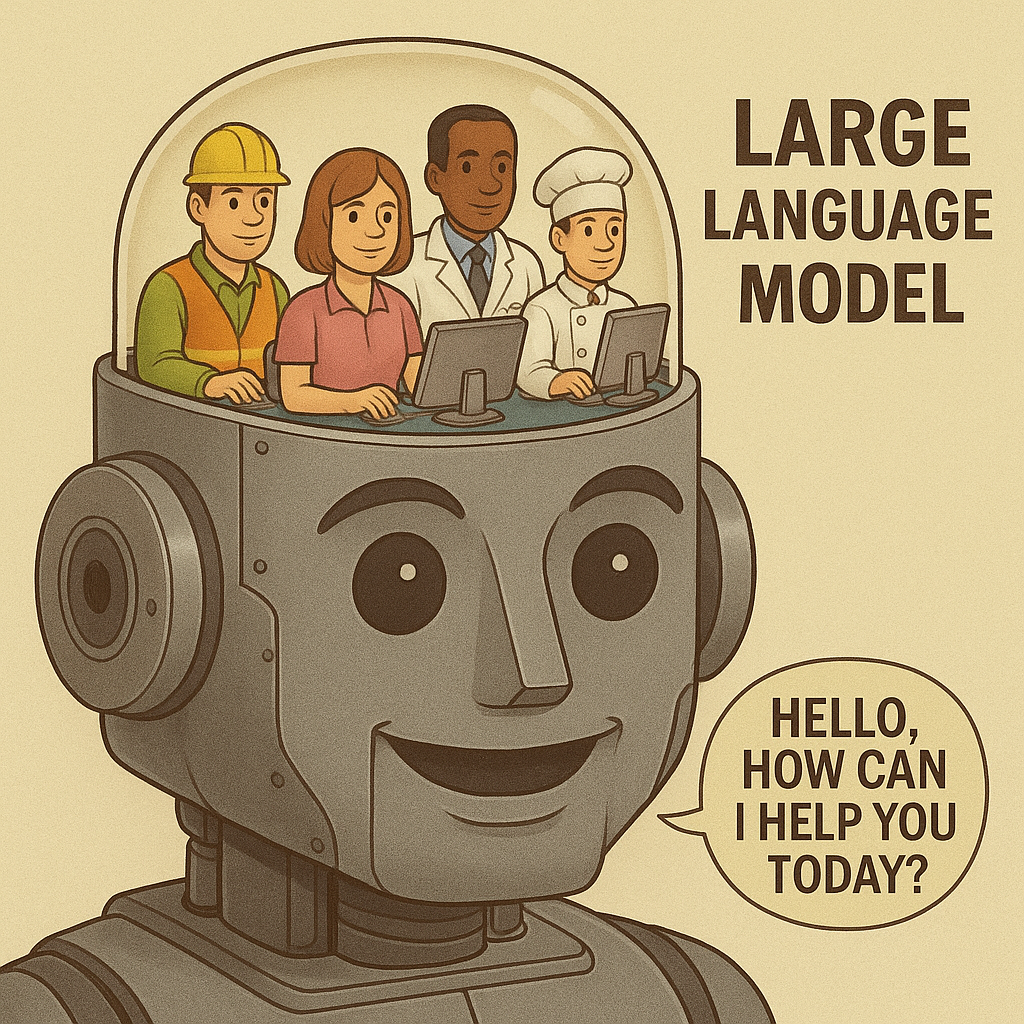}}
{An Interpretation of an LLM as Being Made Up of $\widehat{k}$ Real Human Agents. \label{fig-LLMTurk}}
{Generating an output from the LLM can be thought of as sampling a response from a human agent inside the LLM. The figure is generated by ChatGPT 5 \citep{GPT5}, and borrows ideas from the \emph{Mechanical Turk}, a chess-playing machine from the 18th century with a human player hidden inside.}
\end{figure}

To realize these operational insights, we make the following technical contributions:

\begin{itemize}
\item \textbf{Formulation.} We provide a rigorous mathematical framework for uncertainty quantification in LLM-based survey simulations.

\item \textbf{Methodology.} We propose a flexible methodology that adaptively selects a simulation sample size for valid uncertainty quantification. Our method applies to any LLM and any confidence set construction procedure.

\item \textbf{Theory.} We show that the selected sample size $\widehat{k}$ yields confidence sets with valid average-case coverage. In the one-dimensional case, $\widehat{k}$ is shown to quantify the human-LLM discrepancy through information-theoretic divergences. Moreover, under a ``Mechanical Turk'' model where the synthetic distribution is made up of a hidden population of $\kappa$ real humans, $\widehat{k}$ recovers the hidden population size $\kappa$.

\item \textbf{Experiments.} We test our method on multiple LLMs and two real survey datasets, one on social opinions and the other on educational test. The results verify the coverage validity and reveal substantial heterogeneity in LLM simulation fidelity across domains.
\end{itemize}

It is important to emphasize that our goal is to quantify the simulation fidelity of any given LLM-based procedure (that is, a fixed LLM paired with a fixed prompting rule), rather than improve or optimize it. We treat the human-LLM discrepancy as inherent: if a simulation procedure has poor fidelity, our framework faithfully reveals this through larger confidence sets, and the resulting uncertainty is a property of the procedure itself, not an artifact of our method.

\paragraph{Related works.} 

Our work relates to research on assessing the fidelity of LLM simulations and measuring their alignment with real human populations. Prior studies have explored similarity metrics between synthetic and human distributions \citep{SDL23, HMG24, DHM24, DNL24, CRD25} and Turing-type tests \citep{ABFG23, MXY24, JBe25} to evaluate LLM reliability. While these approaches provide valuable insights into LLM misalignment, their results can be hard to interpret and do not provide an operational guide for using an imperfect LLM in downstream tasks. In contrast, our work introduces a framework that measures simulation fidelity in a way that is immediately operational. By answering the question \emph{``How many humans is the LLM worth?''}, we provide a metric that is both an intuitive measure of the model's fidelity and a direct input for generating statistically sound inferences from its simulations. As a result, we obtain a data-driven selection rule for the synthetic sample size, which has been used as a knob for optimizing the bias-variance tradeoff in data augmentation \citep{SLS23,ORo25}.

Our work also connects to the broad literature on uncertainty quantification for stochastic simulation \citep{NPe21}. The uncertainty in the output of a simulation system is often decomposed into two primary sources: Monte Carlo error from a finite number of simulation runs, and input uncertainty that arises from using finite data to fit the input model \citep{CHo04,Lam16,BLS22}. While LLM's misalignment bears a resemblance to input uncertainty, it presents new and unique challenges due to the black-box nature. Classical approaches to input uncertainty, including the delta method, the bootstrap and Bayesian methods, often require access to model information such as gradients, or extensive model re-evaluation or retraining on i.i.d.~data. These conditions are violated for modern generative models, which typically have complex architectures, are accessed only through black-box interfaces, are trained on massive, non-i.i.d.~mixtures of real and synthetic data, and are prohibitively costly to retrain. In contrast, our method treats the LLM as a black box, acts as a post-processing step on the LLM's outputs, and does not require any model retraining. It jointly addresses the Monte Carlo error and the misalignment by selecting a simulation sample size that balances these two sources of uncertainty, and yields reliable confidence intervals.

The challenges of modern generative models have motivated a recent line of work on model-free statistical inference, including conformal inference \citep{VGF05, SVo08, BAL21, ABF24,KOR24} and prediction-powered inference \citep{ABF23}. At a high level, these methods use labeled data from the true distribution to calibrate imperfect point predictions from an arbitrary black-box model and then construct valid set estimates.
Our approach follows a similar spirit. The ``features'' and ``labels'' in our setting correspond to the survey questions and the population parameters of human responses, respectively.  However, our labels are not directly observable. As a result, the labeled calibration data needed for these methods is not available. Moreover, for every simulation sample size $k$, one can produce a point prediction of the label using $k$ synthetic responses generated by the LLM. As the optimal sample size is not known a priori, there are infinitely many candidate point predictions to choose from. This makes it difficult to apply existing statistical inference methods.

A preliminary version of this work \citep{Short25} appeared in the Forty-Second International Conference on Machine Learning (ICML 2025). The current version contains significantly expanded theoretical analysis and more in-depth discussions of $\widehat{k}$ as an effective human sample size (\myCref{sec-interpretations}), an additional case study on a concentration-based confidence interval (\myCref{sec-KL}), and more comprehensive numerical experiments (\myCref{sec-experiments}).

\paragraph{Outline.} The rest of the paper is organized as follows. \myCref{sec-warmup} studies binary response simulation as a motivating example. \myCref{sec-general} presents the general setup and methodology. \myCref{sec-interpretations} gives interpretations of the selected simulation sample size. \myCref{sec-experiments} illustrates our proposed method on real datasets. \myCref{sec-discussions} concludes the paper.

\paragraph{Notation.} We use $\ZZ_+$ to denote the set of positive integers. For $n\in\ZZ_+$, define $[n]=\{1,2,...,n\}$. For non-negative sequences $\{a_n\}_{n=1}^{\infty}$ and $\{b_n\}_{n=1}^{\infty}$, we write $a_n=O(b_n)$ if there exists $C>0$ such that for all $n$, it holds that $a_n \le C b_n$. We write $a_n = \Omega(b_n)$ if $b_n = O(a_n)$. We write $a_n = \Theta(b_n)$ if $a_n=O(b_n)$ and $a_n=\Omega(b_n)$. The notation $\Bernoulli(p)$ denotes the Bernoulli distribution with mean $p$. The notation $N(\mu,\sigma^2)$ denotes the normal distribution with mean $\mu$ and variance $\sigma^2$.
\section{Warm-up: Simulating a Binary Survey}\label{sec-warmup}

To motivate our problem and methodology, we will start with a simple setting where an LLM simulates binary responses to a survey question. In \myCref{sec-general}, we will present the general setup and the methodology.

\subsection{Motivating Example: Educational Test}\label{sec-example-education}

Consider a school aiming to estimate the proportion $\mu \in [0,1]$ of students that can correctly answer a new test question. This estimate is critical for evaluating both student progress and the question's effectiveness at differentiating between students with varying levels of understanding. These insights can in turn inform how the school can tailor its teaching strategies to better meet student needs.

The most direct approach is to give the test to $n$ students and collect their results $y_1,...,y_n\in\{0,1\}$, where $y_i$ indicates whether student $i$ answers the question correctly. A point estimate for $\mean$ is the sample mean $\responsebar = \frac{1}{n} \sum_{i=1}^n \response_i$. Given $\alpha\in(0,1)$, we can construct a confidence interval for $\mean$:
\begin{equation}\label{eqn-CI-intro-standard}
\left[ \responsebar - \frac{z_{\alpha/2}\cdot \samplesd}{\sqrt{n}}, ~  \responsebar + \frac{z_{\alpha/2}\cdot \samplesd}{\sqrt{n}} \right],
\end{equation}
where $\samplesd = \sqrt{\responsebar (1-\responsebar) }$ is the sample standard deviation, and $z_{\alpha/2}$ is the $(1-\alpha/2)$-quantile of the standard normal distribution $\Normal(0,1)$. By the Central Limit Theorem (CLT), this interval has asymptotic coverage probability $1-\alpha$ as $n\to\infty$. 

Alternatively, the school can use an LLM to simulate students' responses to the question. Compared with testing real students, this approach is more time- and cost-efficient. If we prompt the LLM $k$ times with random student profiles, then it generates $k$ synthetic responses, which leads to synthetic outcomes $\simresponse_1,...,\simresponse_k\in\{0,1\}$. We may then compute the synthetic sample mean $\simresponsebar_k = \frac{1}{k} \sum_{i=1}^k \simresponse_i$. However, as the LLM is generally misaligned with the true student population, $\simresponsebar_k$ can be a poor and unreliable estimate of $\mean$. To quantify its error and make reliable inference, one may form a CLT-based confidence interval for $\mean$:
\begin{equation}\label{eqn-CI-intro-sim}
\simCI(k) =  \left[ \, \simresponsebar_k - z_{\alpha/2}\cdot\simsamplesd_k\sqrt{\frac{C}{k}}, ~~
\simresponsebar_k + z_{\alpha/2}\cdot\simsamplesd_k\sqrt{\frac{C}{k}} \, \right],
\end{equation}
where $\simsamplesd_k = \sqrt{\simresponsebar_k (1-\simresponsebar_k)}$, and $C>1$ is a scaling parameter that dilates the width of the confidence interval. This dilation accounts for the fact that, when the LLM's response distribution deviates from the true population, an unscaled interval may fail to achieve the target $(1-\alpha)$ coverage probability, regardless of the sample size $k$; we give an example in \myCref{sec-impossibility-exact-CLT}. As we will show in \myCref{sec-interpretations}, the dilation factor $C$ also plays a crucial role in revealing the LLM's simulation fidelity.

The statistical validity of $\simCI(k)$ depends crucially on the simulation sample size $k$. As $k\to\infty$, the interval concentrates tightly around the synthetic mean $\simmean = \EE [\simresponse_1]$ and fails to cover the true mean $\mean$. On the other hand, when $k$ is small, the interval becomes too wide and uninformative, even though it may cover $\mean$ with high probability. We provide an illustration in \myCref{fig-tradeoff}.

\begin{figure}[h]
\FIGURE
{\includegraphics[scale=0.55]{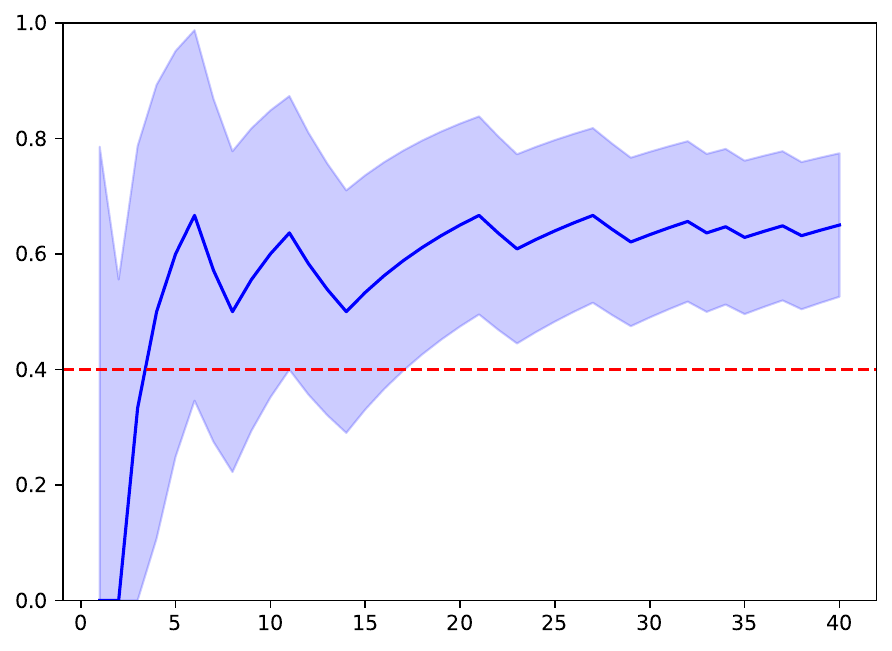}}
{The Coverage-Width Trade-off for the Simulation Sample Size $k$. \label{fig-tradeoff}}
{The true mean is $\mean=0.4$ (red dashed line), and the synthetic distribution has a mean of $\simmean=0.6$. The horizontal axis is the simulation sample size $k$. The blue curve plots the sample mean $\simresponsebar_k$ of the synthetic data, and the blue shaded region visualizes the confidence interval $\simCI(k)$, for $k\in[40]$. For a small sample size ($k\le 6$), the interval is too wide. For a large sample size (say $k\ge 18$), the interval becomes too narrow and fails to cover $\mean$.}
\end{figure}

To tackle the coverage-width trade-off illustrated in \myCref{fig-tradeoff}, we will develop a principled approach to selecting a simulation sample size $\widehat{k}$. Our approach will ensure that the resulting confidence interval $\simCI(\widehat{k})$ is both statistically valid (achieving target coverage $\PP\big(\mean\in\simCI(\widehat{k})\big)\approx 1-\alpha$) and practically informative (not excessively wide). The width of this interval quantifies the LLM's misalignment with the human population: a wide interval signifies a large gap $|\mean - \simmean|$ between the real and synthetic response distributions.

\subsection{Methodology for Selecting Simulation Sample Size}\label{sec-method-1D}

We now introduce our method for choosing a good simulation sample size $\widehat{k}$. The key idea is to leverage a set of test questions for which real students' results are available. By observing how well the LLM's synthetic confidence interval performs on this set of questions, we can guide the choice of $k$ for a new question that bears similarity to the existing ones.

Specifically, we assume access to a calibration set of $m$ similar test questions, which may come from previous tests or a question bank. For each calibration question $j=1,...,m$, we have access to $n_j$ real students' results $\dataset_j = \{ \response_{j,i} \}_{i=1}^{n_j}$. We also simulate LLM responses $\simdataset_j = \{ \simresponse_{j,i} \}_{i=1}^K$, where $K\in\ZZ_+$ is the simulation budget.

Similar to \eqref{eqn-CI-intro-sim}, we use the simulated responses $\simdataset_j$ to form confidence intervals, aiming to cover the true proportion $\mean_j$ of students that can answer the $j$-th question correctly. The first $k$ samples $\{\simresponse_{j,i}\}_{i=1}^k$ of $\simdataset_j$ yield a confidence interval
\begin{equation}\label{eqn-CI-intro-sim-calibrate}
\simCI_j(k) = \left[ \, \simresponsebar_{j,k} - z_{\alpha/2} \cdot \simsamplesd_{j,k} \sqrt{\frac{C}{k}}  ,  ~~
\simresponsebar_{j,k} + z_{\alpha/2} \cdot \simsamplesd_{j,k} \sqrt{\frac{C}{k}}  \, \right], 
\end{equation}
where $\simresponsebar_{j,k} = \frac{1}{k} \sum_{i=1}^k \simresponse_{j,i}$ is the sample mean, and $\simsamplesd_{j,k} = \sqrt{\simresponsebar_{j,k} ( 1 - \simresponsebar_{j,k} )}$ is the estimated standard deviation. We also set the convention $\simCI(0) = \simCI_j(0) = [0,1]$, as nothing can be said about the true parameter without data. We will pick $\widehat{k}\in\{0,1...,K\}$ such that $\simCI_j(\widehat{k})$ covers $\mean_j$ on most calibration questions $j\in[m]$. We expect this choice of $\widehat{k}$ to also work for $\simCI(k)$, as the test questions are similar.

Ideally, we would select $k$ such that the average miscoverage across the calibration set does not exceed the target level $\alpha$:
\begin{equation}\label{eqn-oracle-criterion-1D}
\frac{1}{m} \sum_{j=1}^m \ind \{ \mean_j \not\in \simCI_j(k) \} \le \alpha.
\end{equation}
As the true mean $\mean_j$ is not available, we use the real data $ \dataset_j$ to compute the sample average $\responsebar_j = \frac{1}{n_j} \sum_{i=1}^{n_j} \response_{j,i}$ as a proxy. This leads to the following miscoverage metric:
\begin{equation}\label{eqn-proxy-1D}
\coverage(k) = \frac{1}{m} \sum_{j=1}^m \ind \{ \responsebar_j \not\in \simCI_j(k) \}.
\end{equation}
Our selection criterion is to choose the largest sample size $k$ that maintains a low miscoverage rate across all smaller sample sizes:
\begin{equation}\label{eqn-empirical-criterion-1D}
\widehat{k} = \max \left\{ 0\le k \le K : \coverage(i) \le \alpha/2 ~~\forall i\le k \right\}.
\end{equation}
Note that $\widehat{k}$ is well-defined because $\coverage(0) = 0$ by convention.

The threshold $\alpha/2$ in \eqref{eqn-empirical-criterion-1D} is chosen to account for the additional sampling error introduced by using $\responsebar_j$ as a proxy of $\mean_j$. By CLT, when $n_j$ is large, the sample mean $\responsebar_j$ falls on the left and right sides of $\mean_j$ almost equally likely. When the synthetic confidence interval $\simCI_j(k)$ misses the true mean $\mean_j$, the sample mean $\responsebar_j$ has a $50\%$ chance of falling on the side of $\mean_j$ that is missed by the confidence interval as well. Thus, roughly speaking, the frequency of having $\responsebar_j \not\in \simCI_j(k)$ is at least half of the frequency of having $\mean_{j} \not\in \simCI_j(k)$. In other words, it approximately holds that
\begin{equation}\label{eqn-proxy-to-oracle-1D}
\coverage(k) \geq \frac{1}{2} \cdot \frac{1}{m} \sum_{j=1}^m \ind \{ \mean_{j} \not\in \simCI_j(k) \}.
\end{equation}
Substituting \eqref{eqn-proxy-to-oracle-1D} into \eqref{eqn-oracle-criterion-1D} yields the threshold $\alpha/2$ for choosing $\widehat{k}$.

\subsection{Theoretical Analysis}\label{sec-theory-1D}

In this section, we present a theoretical analysis of our proposed method. We begin by formalizing the setup in \myCref{sec-example-education} and \myCref{sec-method-1D} in mathematical terms.

We model the student population by a distribution $\distribution$ over a space $\profilespace$ of \emph{student profiles} (e.g., vectors of background information, classes taken, grades, etc.). To simulate student responses from the LLM, synthetic student profiles are drawn from a synthetic student population $\simdistribution$ over $\profilespace$, and then fed to the LLM.

We use $\testfunction$ and $\{ \testfunction_j \}_{j=1}^m$ to denote the new test question of interest and the $m$ calibration questions, respectively. 
Students' performance on test questions is characterized by a \emph{performance function} $\performancefunction$: a student with profile $\profile\in\profilespace$ answers a question $\testfunction$ correctly with probability $\performancefunction ( \profile , \testfunction ) \in [0, 1]$. The average student performance on the questions $\testfunction$ and $\{ \testfunction_j \}_{j=1}^m$ are then $\mean = \EE_{\profile \sim \distribution } \performancefunction ( \profile ,  \testfunction ) $ and $\mean_j = \EE_{\profile \sim \distribution } \performancefunction (  \profile, \testfunction_j ) $, respectively.
Similarly, the LLM generates synthetic student performance from a \emph{synthetic performance function} $\simperformancefunction$: when prompted with a synthetic profile $\simprofile\in\profilespace$, the LLM answers a question $\testfunction$ correctly with probability $\simperformancefunction ( \simprofile , \testfunction ) \in [0, 1]$. 

The collection of the real dataset $\dataset_j = \{ \response_{j,i} \}_{i=1}^{n_j}$ can be thought of as drawing $n_j$ i.i.d.~student profiles $\{ \profile_{j,i} \}_{i=1}^{n_j} \sim \distribution$ and then sampling $\response_{j,i} \sim \Bernoulli (  \performancefunction ( \profile_{j,i} , \testfunction_j  ) )$ for each $i\in[n_j]$. Similarly, the generation of the synthetic dataset $\simdataset_j = \{ \simresponse_{j,i} \}_{i=1}^{K}$ can be thought of as drawing i.i.d.~synthetic profiles $\{ \simprofile_{j,i} \}_{i=1}^{K} \sim \simdistribution$ and then sampling $\simresponse_{j,i} \sim \Bernoulli ( 
\simperformancefunction
( \simprofile_{j,i} , \testfunction_j  ) )$ for each $i\in [K]$. We denote the synthetic responses to the new question $\testfunction$ by $\simdataset = \{ \simresponse_i \}_{i=1}^K$. We note that when collecting real or synthetic samples, the performance functions $\performancefunction$ and $\simperformancefunction$ never appear explicitly. They are conceptual tools for the theoretical analysis only.

Finally, we assume that the test questions are drawn randomly from a question bank, and that the datasets are independent.

\begin{assumption}[Randomly sampled questions]\label{assumption-iid-test-1D}
The questions $\testfunction,\testfunction_1,...,\testfunction_m$ are independently sampled from a distribution $\Pi$ over a space $\testfunctionfamily$.
\end{assumption}

\begin{assumption}[Independent data]\label{assumption-indep-data-1D}
For each $j\in[m]$, conditioned on $\testfunction_j$, the datasets $\dataset_j$ and $\simdataset_j$ are independent. Conditioned on $\testfunction_1,...,\testfunction_m$, the dataset tuples $(\dataset_1,\simdataset_1),...,(\dataset_m,\simdataset_m)$ are independent. Finally, $(\testfunction,\simdataset)$ is independent of $\big\{ (\testfunction_j,\dataset_j,\simdataset_j) \big\}_{j=1}^m$.
\end{assumption}

We are now ready to state the theoretical guarantee of our approach. The assumption $\PP( \responsebar_j \le \mean_{j} \mid \testfunction_j ) \in [ \frac{1}{2} - \eta, \frac{1}{2} + \eta ]$ quantifies the CLT approximation that the sample mean $\responsebar_j$ falls on the left and right sides of $\mean_j$ almost equally likely.

\begin{theorem}[Coverage guarantee]\label{thm-coverage-1D}
Let Assumptions \myref{assumption-iid-test-1D} and \myref{assumption-indep-data-1D} hold. Assume that $\PP( \responsebar_j \le \mean_{j} \mid \testfunction_j ) \in [ \frac{1}{2} - \eta, \frac{1}{2} + \eta ]$ for each $j\in[m]$, where $\eta\in[0,1/2)$. Fix $\alpha\in(0,1)$. Then the simulation sample size $\widehat{k}$ defined by \eqref{eqn-empirical-criterion-1D} satisfies
\[
\PP\Big( \mean \in \simCI(\widehat{k}) \Big) \ge 1-(1-2\eta)^{-1} \bigg( \alpha + \sqrt{\frac{2}{m}} \bigg).
\]
The probability is taken with respect to the randomness of $\big\{ ( \testfunction_j , \dataset_j, \simdataset_j ) \big\}_{j=1}^m$, $\testfunction$ and $\simdataset$.
\end{theorem}

\begin{proof}[Proof of \myCref{thm-coverage-1D}]
See \myCref{sec-thm-coverage-1D-proof}.
\end{proof}

\myCref{thm-coverage-1D} provides a coverage guarantee of our selected confidence interval $\simCI(\widehat{k})$. If the CLT approximation is accurate ($\eta\approx0$), then $\simCI(\widehat{k})$ that covers the true mean $\mean$ with probability at least $1-\alpha-O(\sqrt{1/m})$, which converges to the desired $1-\alpha$ coverage as the number of calibration questions $m$ grows. This result establishes the statistical validity of our approach. We later complement this in \myCref{cor-sharpness} by showing that the selected confidence interval is not overly conservative and has near-optimal width.

\section{General Framework and Methodology}\label{sec-general}

In this section, we extend our methodology in \myCref{sec-warmup} to a general framework that applies to multi-dimensional survey responses and any method for constructing confidence sets.

\subsection{Problem Formulation}

Let $\profilespace$ be a space of individual profiles. The true human population is described by a probability distribution $\distribution$ over $\profilespace$. When a person with profile $\profile\in\profilespace$ is given a survey question $\testfunction$, they provide a response $\response$ from a response distribution $\responsedist( ~ \cdot \mid \profile, \testfunction)$ over a response space $\responsespace$. 

Our goal is to construct a confidence set for a parameter $\statistic(\testfunction)$ of the population's response distribution $\responsedistalt(~ \cdot \mid \testfunction )$, which is the average of the individual response distribution over the entire population: $
\responsedistalt(~ \cdot \mid \testfunction ) = \int_{\profilespace} \responsedist( ~ \cdot \mid \profile, \testfunction) \, \distribution(d\profile)$.
For instance, in the educational test example in \myCref{sec-warmup}, the response space is $\responsespace=\{0,1\}$, and $\statistic(\testfunction)$ is the proportion of students that can correctly answer test question $\testfunction$. More generally, the parameter $\statistic(\testfunction)$ can be multi-dimensional and/or continuous-valued. We give several examples below.

\begin{example}[Public opinion survey]\label{example-public-survey}
Consider a public opinion survey where a survey question $\testfunction$ has $5$ options, such as ``How often do you talk to your neighbors?''. An individual profile $\profile \in \profilespace$ may consist of age, gender, occupation, etc. The response space $\responsespace$ can be represented by a $5$-element set $\{e_i\}_{i=1}^5$, where a response $\response = e_i$ indicates that a person chooses the $i$-th option. A parameter of interest $\statistic ( \testfunction )$ is a $5$-dimensional vector that summarizes the proportion of the population choosing each of the $5$ options in the question $\testfunction$.
\end{example}

\begin{example}[Sentiment in opinion survey]\label{example-public-survey-1D}
Consider the setup in \myCref{example-public-survey}. When the $5$ choices in a survey question correspond to ordered sentiments (e.g., on a Likert scale), we can map them to numeric scores, say, $v = (-1,-\frac{1}{3},0,\frac{1}{3},1)^\top$. Then the parameter $\widetilde{\statistic} (\testfunction) = \langle v,\statistic (\testfunction)  \rangle$ measures the population's average sentiment for the question $\testfunction$.
\end{example}

\begin{example}[Market research]
Suppose a company wants to estimate its customers' willingness-to-pay (WTP) for a new product. Here, each survey question $\testfunction$ is associated with a product, and a response $\response$ is a customer's stated WTP. An important parameter $\statistic(\testfunction)$ can be the median WTP or another quantile of the WTP distribution across the customers. For a given quantile level $\tau\in(0,1)$, the $\tau$-quantile is defined by $\statistic (\testfunction) = \inf \left\{ q\in[0,\infty) : \PP_{\response \sim \responsedistalt(\cdot\mid \testfunction )} (\response\le q) \ge \tau \right\}$.
\end{example}

We consider constructing a confidence set for the population parameter $\statistic (\testfunction)$ by using simulated responses from an LLM. This confidence set is intended to provide valid inference for the true human population parameter, while quantifying the uncertainty induced by human-LLM misalignment. Given a profile $\profile$, a survey question $\testfunction$ and a prompt $\prompt$, the LLM simulates a response $\simresponse$ from a distribution $\simresponsedist(~\cdot\mid \profile , \testfunction, \prompt)$ which aims to mimic $\responsedist(~\cdot\mid \profile , \testfunction)$. We generate i.i.d.~synthetic profiles $\{ \simprofile_i \}_{i=1}^K$ from a distribution $\simdistribution$, then feed them into the LLM along with $\testfunction$ and $\prompt$. The LLM then generates synthetic responses $\{ \simresponse_i \}_{i=1}^K$, where $\simresponse_i \sim \simresponsedist(~\cdot\mid \simprofile_i , \testfunction, \prompt)$. Here $K\in\ZZ_+$ is the simulation budget.

Using the simulated samples $\simdataset = \{ \simresponse_i \}_{i=1}^K$, we can construct candidate confidence sets. While our motivating example in \myCref{sec-warmup} uses the one-dimensional CLT-based confidence interval \eqref{eqn-CI-intro-sim}, our framework is general and accommodates standard confidence set construction techniques, such as inverting hypothesis tests \citep{CBe02}, the bootstrap \citep{Efr79}, and the empirical likelihood ratio function \citep{OWe90}. We will assume access to a black-box procedure $\setmap$ that transforms a dataset $\dataset$ into a confidence set $\setmap(\dataset)\subseteq\RR^d$, where $d$ is the dimension of $\statistic(\testfunction)$. Applying this procedure $\setmap$ to the synthetic data $\simdataset$ yields a family of candidate confidence sets $\{\simCIalt(k)\}_{k=1}^{K}$ by
\begin{equation}\label{eqn-CI-sim}
\simCIalt(k) = \setmap \big( \{ \simresponse_i \}_{i=1}^k \big).
\end{equation}
We also set $\simCIalt(0) = \RR^d$, so $\statistic ( \testfunction ) \in \simCIalt(0)$ always.

As the LLM may not be a faithful reflection of the true human population, the core challenge is to select a sample size $k$ such that $\simCIalt(k)$ covers $\statistic(\testfunction)$ with high probability. As in \myCref{sec-method-1D}, we calibrate the choice of $k$ using a calibration set of $m$ similar survey questions $\{\testfunction_j\}_{j=1}^m$, where question $j$ contains $n_j$ real human responses $\dataset_j = \{ \response_{j,i} \}_{i=1}^{n_j}$ from surveyees $\{ \profile_{j,i} \}_{i=1}^{n_j} \sim \distribution$ . We also feed synthetic profiles $\{ \simprofile_{j,i} \}_{i=1}^K\sim\simdistribution$, the question $\testfunction_j$ and the prompt $\prompt$ into the LLM, which then simulates responses $\simdataset_j = \{ \simresponse_{j,i} \}_{i=1}^{K}$ with $\simresponse_{j,i} \sim \simresponsedist (~\cdot \mid \simprofile_{j,i}, \testfunction_j, \prompt )$. We now formally state our central problem of simulation sample size selection.

\begin{problem}[Uncertainty quantification]\label{problem-general}
Given $\alpha \in (0,1)$, how to use $\{ \dataset_j \}_{j=1}^m$ and $\{ \simdataset_j \}_{j=1}^m$ to choose $\widehat{k}\in\{0,1,...,K\}$ such that
\[
\PP \Big( \statistic ( \testfunction ) \in \simCIalt(\widehat{k}) \Big) \approx 1-\alpha?
\]
\end{problem}

\subsection{General Methodology for Sample Size Selection}

We will now extend the method in \myCref{sec-method-1D} to tackle Problem \myref{problem-general}, where the population parameter $\statistic(\testfunction)$ may be multi-dimensional and confidence sets are produced by an arbitrary procedure $\setmap$. 

For each calibration question $j\in[m]$, applying $\setmap$ to the first $k$ samples of the synthetic responses $\simdataset_j$ yields a confidence set
\begin{equation}\label{eqn-CI-sim-calibrate}
\simCIalt_j(k) = \setmap \left( \{ \simresponse_{j,i} \}_{i=1}^k \right),\quad \forall k\in[K].
\end{equation}
We also set $\simCIalt_j(0) = \RR^d$. Our goal is to choose a sample size $\widehat{k}\in\{0,1,...,K\}$ such that $\simCIalt_j(\widehat{k})$ covers $\statistic (\testfunction_j)$ for most calibration questions $j\in[m]$.

Analogously to \eqref{eqn-oracle-criterion-1D}, an ideal choice of $k$ would give a miscoverage rate no more than the target level $\alpha$ over the $m$ calibration questions:
\begin{equation}\label{eqn-oracle-criterion}
\frac{1}{m} \sum_{j=1}^m \ind \{  \statistic( \testfunction_j )  \not\in \simCIalt_j(k) \} \le \alpha.
\end{equation}
Of course, the true parameters $ \statistic( \testfunction_j )$ are unknown, so \eqref{eqn-oracle-criterion} cannot be evaluated directly. In \myCref{sec-method-1D}, we approximated the unknown means by unbiased point estimates, but it does not easily extend to more general, multi-dimensional parameters $\statistic(\testfunction_j)$.

Instead of using a point estimate, we will construct a \emph{confidence set estimate} for $\statistic( \testfunction_j )$. Fix a confidence level $\gamma\in(0,1)$. We use the real data $\dataset_j$ to construct a confidence set $\CIalt_j$ that satisfies
\begin{equation}\label{eqn-CI-calibrate}
\PP\Big( \statistic ( \testfunction_j)  \in \CIalt_j \Bigm| \testfunction_j \Big) \ge \gamma.
\end{equation}
These confidence sets are easy to construct as the samples in $\dataset_j$ follow the true response distribution. 

We then approximate the coverage criterion $\statistic( \testfunction_j )  \in \simCIalt_j(k)$ by a set-containment test $\CIalt_j \subseteq \simCIalt_j(k)$. When $\statistic( \testfunction_j ) \in \CIalt_j$, the condition $\CIalt_j \subseteq \simCIalt_j(k)$ is sufficient for $\statistic( \testfunction_j )  \in \simCIalt_j(k)$. This motivates the following empirical miscoverage metric, which is defined as the frequency that the real-data confidence set $\CIalt_j$ fails to be contained by the synthetic confidence set $\simCIalt_j(k)$:
\begin{equation}\label{eqn-proxy}
\coveragealt(k) = \frac{1}{m} \sum_{j=1}^m \ind \{ \CIalt_j\not\subseteq \simCIalt_j(k) \}.
\end{equation}
Since $\statistic( \testfunction_j ) \in \CIalt_j$ with probability at least $\gamma$, then the frequency of having $\CIalt_j \not\subseteq \simCIalt_j(k)$ is at least $\gamma$ times the frequency of $\statistic( \testfunction_j ) \not\in \simCIalt_j(k)$. Thus, roughly speaking, 
\begin{equation}\label{eqn-proxy-to-oracle}
\coveragealt(k) \ge \gamma \cdot\left( \frac{1}{m} \sum_{j=1}^m \ind \{  \statistic( \testfunction_j )  \not\in \simCIalt_j(k) \} \right).
\end{equation}

Combining \eqref{eqn-oracle-criterion} and \eqref{eqn-proxy-to-oracle} leads to the sample size selection criterion
\begin{equation}\label{eqn-empirical-criterion}
\widehat{k} = \max \left\{ 0\le k \le K : ~ \coveragealt(i) \le \gamma\alpha, ~\forall i\le k \right\}.
\end{equation}
Note that $\widehat{k}$ is well-defined because $\coveragealt(0) = 0$. The full procedure is summarized in \myCref{alg-general}. 

\begin{algorithm}[t]
	{\bf Input:} Survey questions with real and simulated responses $\big\{(\testfunction_j, \dataset_j, \simdataset_j)\big\}_{j=1}^m$, target miscoverage probability $\alpha$, confidence set construction procedure $\setmap$, confidence level $\gamma$, simulation budget $K$. \\
	{\bf For $j=1,...,m$:} \\
		\hspace*{.6cm} Use $\setmap$ and $\simdataset_j$ to construct synthetic confidence sets $\{\simCIalt_j(k)\}_{k=0}^K$ according to \eqref{eqn-CI-sim-calibrate}. \\
		\hspace*{.6cm} Use $\dataset_j$ to construct a confidence set $\CIalt_j$ satisfying \eqref{eqn-CI-calibrate}.\\
	Define the empirical miscoverage metric
	\[
	\coveragealt(k) = \frac{1}{m} \sum_{j=1}^m \ind \{ \CIalt_j\not\subseteq \simCIalt_j(k) \}.
	\]
	Choose $\widehat{k} = \max \left\{ 0\le k \le K :~ \coveragealt(i) \le \gamma\alpha ,~\forall i\le k \right\}$. \\
	{\bf Output:} Sample size $\widehat{k}$.\caption{Simulation Sample Size Selection}
	\label{alg-general}
\end{algorithm}

We now present the coverage guarantee for our method, which shows that the chosen confidence set $\simCIalt(\widehat{k})$ has coverage probability at least $1-\alpha-O(1/\sqrt{m})$.

\begin{theorem}[Coverage guarantee]\label{thm-coverage}
Let Assumptions \myref{assumption-iid-test-1D} and \myref{assumption-indep-data-1D} hold. Fix $\alpha\in(0,1)$. The sample size $\widehat{k}$ selected by \myCref{alg-general} satisfies
\[
\PP\Big( \statistic ( \testfunction ) \in \simCIalt(\widehat{k}) \Big) \ge 1-\alpha - \gamma^{-1} \sqrt{\frac{1}{2m}}.
\]
\end{theorem}

\begin{proof}[Proof of \myCref{thm-coverage}]
See \myCref{sec-thm-coverage-proof}.
\end{proof}

Remarkably, this coverage guarantee holds for any LLM and any confidence set construction procedure $\setmap$, regardless of their qualities. Nevertheless, the size of the chosen confidence set $\simCIalt(\widehat{k})$ depends on these factors. For example, if the LLM is poorly aligned, then the resulting confidence set $\simCIalt(\widehat{k})$ will inevitably be wide in order to account for the large misalignment gap.

\section{How Many Human Samples is an LLM Worth?}\label{sec-interpretations}

In the previous sections, we proposed a data-driven method for selecting a simulation sample size $\widehat{k}$ that ensures a statistically valid confidence interval. In this section, we advance the interpretation of $\widehat{k}$ further, showing that beyond its role for uncertainty quantification, it quantifies the \emph{effective number of the human samples} that the LLM can represent. For clarity and to distill the core insights, our analysis will focus on the one-dimensional binary response setting in \myCref{sec-warmup}.

We develop our argument in three steps. In \myCref{sec-MTurk}, we introduce a conceptual analogy that compares an LLM to a ``Mechanical Turk'' composed of $\kappa$ hidden human agents, where a larger $\kappa$ indicates higher simulation fidelity. In \myCref{sec-MTurk-IT}, we formalize this intuition, by showing that the ``hidden population size'' $\kappa$ is inversely proportional to an information-theoretic measure (the $\chi^2$-divergence) between the real and synthetic response distributions. Consequently, a larger $\kappa$ implies a smaller distribution discrepancy. Finally, in \myCref{sec-effective-sample-size}, we prove that the sample size $\widehat{k}$ selected by our method is also inversely proportional to the $\chi^2$-divergence, and thus $\widehat{k}\approx \kappa$. This establishes our central claim that $\widehat{k}$ is an effective human sample size and a fidelity measure of the LLM.

\subsection{LLM as a Mechanical Turk}\label{sec-MTurk}

\begin{figure}[h]
\FIGURE
{\includegraphics[scale=0.35]{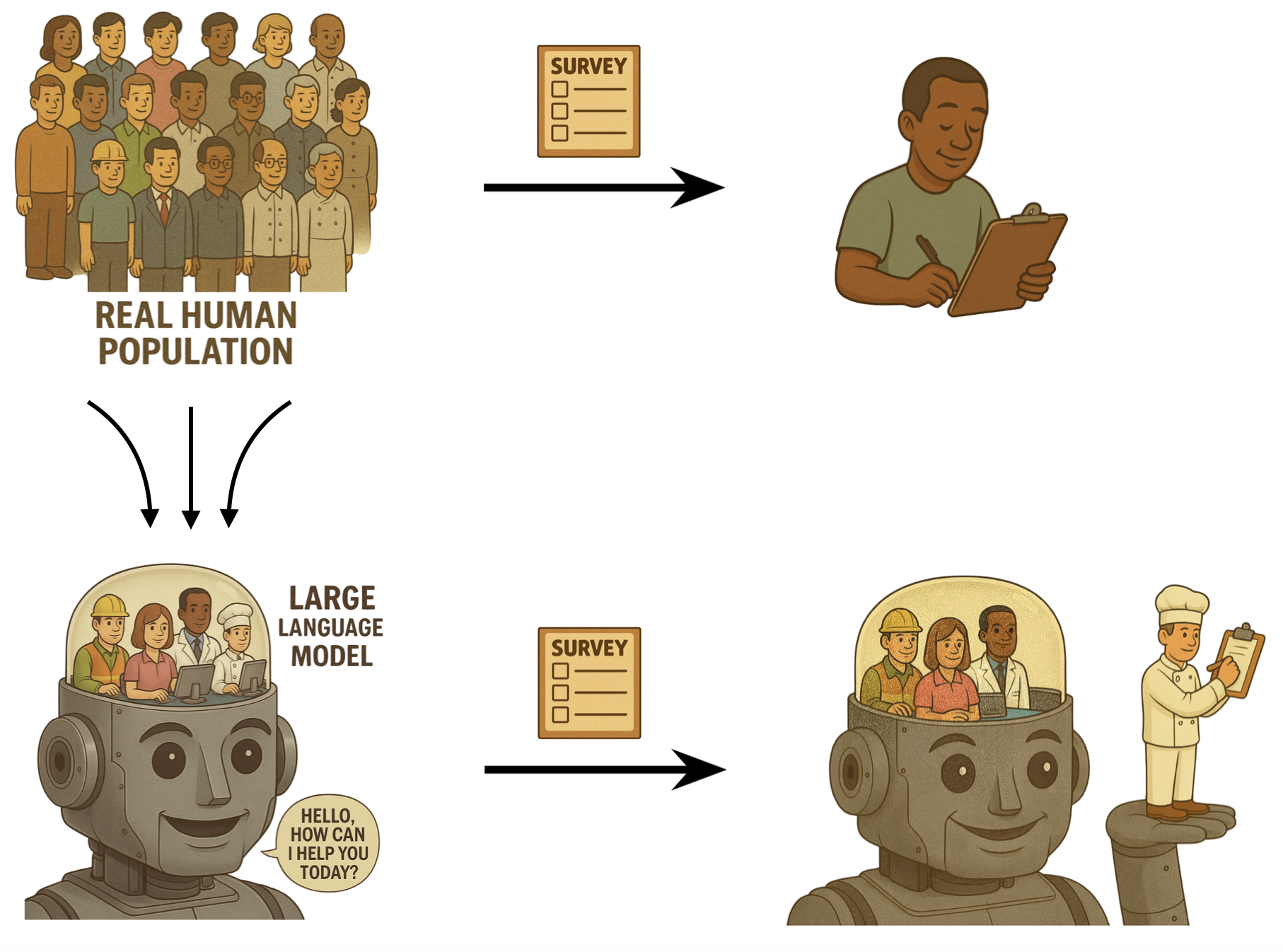}}
{The LLM as a Mechanical Turk. \label{fig-human-LLMTurk}}
{The top panel illustrates the traditional survey process, where an individual is sampled from the real human population to provide a response. The bottom panel depicts our conceptual model, in which the LLM is modeled as containing a hidden pool of agents drawn from that same population. When the LLM is queried, its response is generated by one of these internal agents. This figure is created based on images generated by ChatGPT 5 \citep{GPT5}.}
\end{figure}

We begin by introducing a conceptual analogy: LLM as a ``Mechanical Turk.'' Imagine that the LLM is a call center containing a hidden pool of $\kappa$ real human agents randomly sampled from the true population. When the LLM is queried, it randomly selects one of the agents and returns their response. Thus, the LLM effectively simulates a subpopulation formed by these $\kappa$ agents. \myCref{fig-human-LLMTurk} provides a visualization for this analogy.

Intuitively, the size of the hidden pool $\kappa$ reflects the LLM's simulation fidelity. A larger $\kappa$ implies that the LLM captures richer information about the true population, whereas a small $\kappa$ suggests that the LLM only reflects a non-representative subgroup.

We can formalize this analogy by defining the LLM's response distribution under this model. Recall that for a given survey question $\testfunction$, the true human response distribution $\responsedistalt(~ \cdot \mid \testfunction)$ is a Bernoulli distribution with mean $\mean(\testfunction) = \EE_{\profile\sim\distribution} \performancefunction ( \profile , \testfunction )$:
\[
\responsedistalt(~ \cdot \mid \testfunction) = \Bernoulli\left( \mean(\testfunction) \right).
\] 
The LLM's synthetic response distribution, which we denote by $\simresponsedistalt(~ \cdot \mid \testfunction)$, is a Bernoulli distribution with mean $\simmean (\testfunction) = \EE_{\simprofile\sim\simdistribution} \simperformancefunction ( \simprofile, \testfunction)$:
\[
\simresponsedistalt(~ \cdot \mid \testfunction)
= 
\Bernoulli \left( \simmean (\testfunction) \right).
\]
The Mechanical Turk model posits that the LLM is made up of $\kappa$ hidden human agents, with i.i.d.~profiles $\profile_1,...,\profile_{\kappa} \sim \distribution$. In this case, the synthetic response distribution is the average response distribution of these $\kappa$ agents, which has mean $\widehat{\mean}_{\kappa}(\testfunction) = \frac{1}{\kappa} \sum_{i=1}^{\kappa} \performancefunction(\profile_i, \testfunction)$. We formally state the model as follows.

\begin{assumption}[Mechanical Turk model]\label{assumption-MTurk}
It holds that $\simmean (\testfunction) = \widehat{\mean}_{\kappa}(\testfunction)$.
\end{assumption}

We refer to $\kappa$ as the \emph{hidden population size}. It reflects the amount of information about the human population captured by the LLM. If we were to query this model an infinite number of times, then we would keep resampling each of the $\kappa$ human agents, and precisely learn the mean response $\widehat{\mean}_{\kappa}(\testfunction)$ of this small, unobserved sample of $\kappa$ humans. This is similar to the bootstrap in statistics; we will discuss their relation in \myCref{sec-bootstrap}.

Our discussion so far has remained a thought experiment. In what follows, we will mathematically formalize the relation between $\kappa$ and the LLM's simulation fidelity, and show that the simulation sample size $\widehat{k}$ can recover $\kappa$.

\subsection{Information-Theoretic Characterization of Hidden Population Size}\label{sec-MTurk-IT}

In this section, we formalize the intuition that a larger hidden population size $\kappa$ indicates higher simulation fidelity, by showing that $\kappa$ is inversely proportional to an information-theoretic discrepancy measure between the real and synthetic response distributions. The discrepancy measure we use is the \emph{$\chi^2$-divergence} between the real response distribution $\Bernoulli\left( \mean(\testfunction) \right)$ and the synthetic response distribution $\Bernoulli\left( \simmean(\testfunction) \right)$:
\[
\discchi(\testfunction)
=
\frac{\left| \mean(\testfunction) - \simmean(\testfunction) \right|^2}{\simmean(\testfunction)\big(1-\simmean(\testfunction) \big)}
\]
over a survey question $\testfunction$. This metric measures the squared difference between the true and synthetic mean responses, normalized by the variance of the synthetic response distribution. It penalizes misalignment more heavily when the synthetic distribution is nearly deterministic (that is, when $\simmean(\testfunction)$ is near $0$ or $1$); in that case, the LLM almost always generates the dominant response and fails to simulate the heterogeneity of the human population.

The metric $\discchi(\testfunction)$ measures the LLM's misalignment on the survey question $\testfunction$. As $\testfunction\sim\Pi$ is drawn randomly, we are interested in the typical level of $\discchi(\testfunction)$ across questions, or mathematically, its \emph{quantile}.

\begin{definition}[Quantile]
Let $\alpha\in(0,1)$. The \textbf{$(1-\alpha)$-quantile} of a random variable $X$ is defined by
\[
\quantile_{1-\alpha}(X) = \inf\left\{ x\in\RR : \PP(X\le x) \ge 1-\alpha \right\}.
\]
\end{definition}

The quantile $\quantile_{1-\alpha}(\discchi(\testfunction))$ measures the discrepancy of the real and synthetic response distributions over \emph{the most common $(1-\alpha)$ proportion} of the questions. Under the Mechanical Turk model (Assumption \myref{assumption-MTurk}) where the LLM consists of $\kappa$ real humans $\profiles_{\kappa} = \{\profile_i\}_{i=1}^{\kappa}$, we use $\quantile_{1-\alpha}\left( \discchi(\testfunction) \mid \profiles_{\kappa} \right)$ to denote the conditional $(1-\alpha)$-quantile of $\discchi(\testfunction)$ given $\profiles_{\kappa}$.

\myCref{thm-kappa-IT-chi} below shows that with high probability, $\quantile_{1-\alpha}\left( \discchi(\testfunction) \mid \profiles_{\kappa} \right) \lesssim 1/\kappa$. Therefore, a larger hidden population size $\kappa$ implies a smaller distribution discrepancy, which is consistent with the intuition in \myCref{sec-MTurk}.

\begin{theorem}[Information-theoretic characterization of $\kappa$]\label{thm-kappa-IT-chi}
Let Assumption \myref{assumption-MTurk} hold. Suppose there exists $r\in(0,1/2]$ such that $\mean(\testfunction)\in[r,1-r]$ for all $\testfunction\in\testfunctionfamily$. Let $\delta\in(0,1)$. There exist $\kappa_0(\alpha,\delta,r)$ and $c(\alpha,\delta,r)$ depending on $\alpha, \delta, r$ such that when $\kappa \ge \kappa_0(\alpha,\delta,r)$, the following holds with probability at least $1-\delta$:
\[
\quantile_{1-\alpha} \big( \discchi(\testfunction) \mid \profiles_{\kappa} \big) \le \frac{c(\alpha,\delta,r)}{\kappa}.
\]
Here the probability is taken over the randomness of $\profiles_{\kappa} = \{\profile_i\}_{i=1}^{\kappa}$.
\end{theorem}

\begin{proof}[Proof of \myCref{thm-kappa-IT-chi}]
See \myCref{sec-thm-kappa-IT-chi-proof}.
\end{proof}

\begin{remark}[Discrepancy measure]
As an alternative to our quantile-based discrepancy measure $\quantile_{1-\alpha}(\discchi(\testfunction))$, another popular and intuitive discrepancy measure is the \emph{integral probability metric}
\[
\sup_{\testfunction}\left| \mean(\testfunction) - \simmean(\testfunction) \right|.
\]
The crucial difference is that the latter is the \emph{worst-case} total variation distance between the real and synthetic distributions over \emph{all} possible questions $\testfunction$. In contrast, our discrepancy measure is a quantile of the $\chi^2$-divergence, and captures the distribution discrepancy only for the \emph{most common} questions. After all, we are interested in a high-probability coverage guarantee. In this regard, our quantile-based measure is more refined and optimistic than the integral probability metric.
\end{remark}

\subsection{Simulation Sample Size as Effective Human Sample Size}\label{sec-effective-sample-size}

In this section, we prove our central result that the simulation sample size $\widehat{k}$ selected by our method can approximately recover the hidden population size $\kappa$ in the Mechanical Turk model. To formalize this connection, we need to explicitly account for the dilation factor $C>1$ in the confidence interval \eqref{eqn-CI-intro-sim}; recall that it was originally introduced to guarantee that valid coverage is achievable. We will write the confidence interval \eqref{eqn-CI-intro-sim} as $\simCI(k;C)$ to emphasize its dependence on $C$:
\[
\simCI(k;C) =  \left[ \, \simresponsebar_k - z_{\alpha/2}\cdot\simsamplesd_k\sqrt{\frac{C}{k}}, ~~
\simresponsebar_k + z_{\alpha/2}\cdot\simsamplesd_k\sqrt{\frac{C}{k}} \, \right].
\]
Consequently, the sample size $\widehat{k}$ selected by our method \eqref{eqn-empirical-criterion-1D} is also dependent on $C$:
\begin{equation}\label{eqn-selected-sample-size}
\widehat{k}(C) = \max\bigg\{ 0\le k\le K : ~ \frac{1}{m} \sum_{j=1}^m \ind \{ \responsebar_j \not\in \simCI_j(i;C) \} \le \frac{\alpha}{2},~~\forall i\in[k] \bigg\}.
\end{equation}
Our analysis will show that for large $C$, the normalized sample size $\widehat{k}(C)/C$ is approximately inversely proportional to the distribution discrepancy $\quantile_{1-\alpha}(\discchi(\testfunction))$. Combined with the information-theoretic characterization of $\kappa$ in \myCref{thm-kappa-IT-chi}, this implies that $\widehat{k}(C) / C \gtrsim \kappa$.

To facilitate our theoretical analysis, we make the following regularity assumption, which states that for most test questions, the synthetic response distributions are non-degenerate. This assumption ensures that on most questions, the LLM can generate both ``yes'' and ``no'' responses, so that the $\chi^2$-divergence is well-defined.
 
\begin{assumption}[Non-degenerate synthetic response distribution]\label{assumption-nondegeneracy-synthetic}
There exists a constant $\eta\in(0,1/2]$ such that
\[
\PP\Big( \simmean(\testfunction) \in [\eta, \, 1-\eta] \Big) \ge 1-\alpha/8.
\]
Here the probability is taken over the randomness of the question $\testfunction$.
\end{assumption}

To keep the core insights clear, we will study an \emph{oracle} version of $\widehat{k}(C)$, which is the largest sample size that would achieve the target $(1-\alpha)$ coverage if we had perfect knowledge of the response distributions: 
\begin{equation}\label{eqn-oracle-sample-size-chi}
k^*(C)= \sup \big\{ k \in \ZZ_+ : \PP\big( \mean(\testfunction) \in \simCI(k;C) \big) \ge 1-\alpha \big\}.
\end{equation}
We are now ready to state our result that characterizes the normalized sample size $k^*(C)/C$.

\begin{theorem}[Oracle simulation sample size]\label{thm-oracle-effective-sample-size-chi}
Let Assumption \myref{assumption-nondegeneracy-synthetic} hold. Let the sample size $k^*(C)$ be defined by \eqref{eqn-oracle-sample-size-chi}. There exists $\varepsilon(\alpha,C)$ such that for all $C>1$,
\[
\frac{z_{\alpha/2}^2}{\quantile_{1-\alpha/2}\left( \discchireal(\testfunction) \right)} - \varepsilon(\alpha,C)
\le 
\frac{k^*(C)}{C}
\le 
\frac{z_{\alpha/2}^2}{\quantile_{1-2\alpha}\left( \discchireal(\testfunction) \right)} + \varepsilon(\alpha,C),
\]
and $\lim\limits_{C\to\infty}\varepsilon(\alpha,C) = 0$.
\end{theorem}

\begin{proof}[Proof of \myCref{thm-oracle-effective-sample-size-chi}]
See \myCref{sec-thm-oracle-effective-sample-size-chi-proof}.
\end{proof}

\myCref{thm-oracle-effective-sample-size-chi} reveals the important relation between the simulation sample size and the LLM's simulation fidelity: the normalized simulation sample size $k^*(C)/C$ is inversely proportional to the human-LLM discrepancy $\quantile_{1-\alpha}(\discchi(\testfunction))$. This relation is intuitive: if the LLM has high fidelity (small distribution discrepancy), then we are justified in drawing a large number of samples to achieve a precise estimate. A similar result holds for $\widehat{k}(C)$, which we defer to \myCref{thm-effective-sample-size-chi} in \myCref{sec-effective-sample-size-full}. 

This insight, combined with the result that the hidden population size $\kappa$ is also inversely proportional to the human-LLM discrepancy (\myCref{thm-kappa-IT-chi}), shows that under the Mechanical Turk model, $\widehat{k}(C) / C \gtrsim \kappa$, serving as an empirical estimate of the hidden population size $\kappa$. This leads to our main insight:
\begin{center}
\emph{$\widehat{k}$ is an effective human sample size and a fidelity measure.}
\end{center}
In particular, a large value of $\widehat{k}(C)/C$ indicates a small discrepancy between the real and synthetic response distributions, and that the LLM captures rich information about the human population. We remark that \myCref{thm-oracle-effective-sample-size-chi} itself does not rely on the Mechanical Turk model.

Based on the counterpart of \myCref{thm-oracle-effective-sample-size-chi} for $\widehat{k}(C)$, we prove the following corollary that the resulting confidence interval $\simCI(\widehat{k}(C);C)$ is not overly conservative.

\begin{corollary}[Near-optimal interval width]\label{cor-sharpness}
Let Assumption \myref{assumption-nondegeneracy-synthetic} hold. Suppose the $m$ calibration questions have the same number of human responses: $n_j=n$ for each $j\in[m]$. Choose $\delta\in(0,1)$. For $C$ sufficiently large, with probability at least $1-\delta$, the selected confidence interval $\simCI(\widehat{k}(C);C)$ has width 
\[
O\left( \max\Big\{ C^{1/2}K^{-1/2}, \, n^{-1/2} , \, \quantile_{1-\alpha/8}^{1/2}\big(\discchireal(\testfunction)\big) \Big\} \right),
\]
where $O(\cdot)$ hides a constant depending on $\alpha$ and $\eta$.
\end{corollary}

\begin{proof}[Proof of \myCref{cor-sharpness}]
See \myCref{sec-cor-sharpness-proof}.
\end{proof}

\myCref{cor-sharpness} shows that the width of our confidence interval is affected by three factors: the simulation budget $K$, the amount of real human data $n$, and the LLM's simulation error $\discchireal(\testfunction)$. In particular, it implies that the selected confidence interval has a near-optimal width. To see this, suppose that the simulation budget $K$ is large, so with high probability, $\simCI(\widehat{k})$ has a width of $O \big( \max \big\{ n^{-1/2}, \quantile_{1-\alpha/8}^{1/2}\big(\discchireal(\testfunction)\big) \big\}  \big)$. First, as $n$ real human responses can identify the true mean only up to an error of $O(n^{-1/2})$, then any valid confidence interval $\simCI(k)$ must have a width of $\Omega(n^{-1/2})$. Second, any valid $\simCI(k)$ needs to cover the discrepancy between the real and synthetic human response distributions, and must have a width of $\Omega\big( \sqrt{\discchireal(\testfunction)} \big)$ for most realizations of the random question $\testfunction\sim\Pi$. This shows the near-optimality of the interval width.

\subsection{Connection to Parametric Bootstrap}\label{sec-bootstrap}

The Mechanical Turk analogy has provided an intuitive model for understanding the simulation sample size $\widehat{k}$. In this section, we offer a more statistical view of the Mechanical Turk model by comparing our framework with the classical parametric bootstrap method (e.g., Section 10.4 of \cite{EHa21}), which is one of the most widely-used methods for uncertainty quantification in stochastic simulation. This comparison lends an additional statistical meaning of the hidden population size $\kappa$, and illustrates how our method addresses the unique challenges posed by pre-trained, black-box generative models such as LLMs.

In the parametric bootstrap (\myCref{fig-diagram-bootstrap}), one begins with a dataset of $\kappa$ i.i.d.~samples $y_1,...,y_{\kappa}$ from a true distribution, say $\mathscr{P}$. These samples are used to fit a parametric model $\mathscr{P}_{\widehat{\theta}}$ that serves as an estimate of the true model. To construct confidence intervals, one repeatedly resamples $y_1^*,...,y_{\kappa}^*$ from the fitted model. Crucially, each time we resample from $\mathscr{P}_{\widehat{\theta}}$, the sample size should also be $\kappa$, which can be viewed as the \emph{effective sample size} captured by the model.

\begin{figure}[h]
\FIGURE
{\includegraphics[scale=0.8]{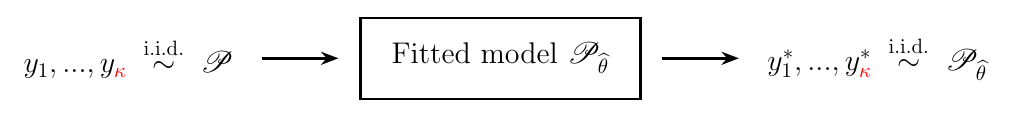}}
{Parametric Bootstrap. \label{fig-diagram-bootstrap}}
{}
\end{figure}

In contrast, our framework (\myCref{fig-diagram-LLM}) begins with a given black-box model (such as an LLM), say $\mathscr{P}^{\mathsf{syn}}$. The training procedure is often very complicated and unknown to us, e.g., it might have been trained on non-i.i.d.~human data, and even synthetic data. Because of this, it is unclear \emph{a priori} how many samples should be drawn from $\mathscr{P}^{\mathsf{syn}}$ to construct confidence intervals. Our framework takes a data-driven approach to select a sample size $\widehat{k}$, whose normalized version $\widehat{k}/C$ can be thought of as an estimate of the model's effective sample size $\kappa$. Then, for future simulation tasks, we know to generate $\widehat{k}$ samples from $\mathscr{P}^{\mathsf{syn}}$.

\begin{figure}[h]
\FIGURE
{\includegraphics[scale=0.8]{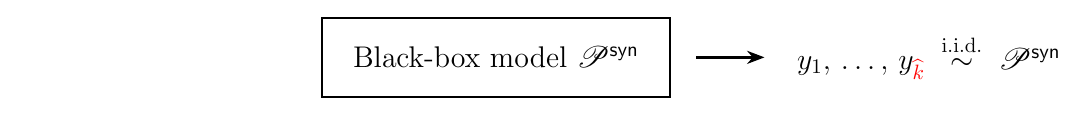}}
{Our framework. \label{fig-diagram-LLM}}
{}
\end{figure}

In \myCref{sec-dilation-bootstrap}, we further establish a more quantitative connection, where our parameters can be directly interpreted in bootstrap terms: the hidden population size $\kappa$ represents the number of samples in each bootstrap resample, the dilation factor $C$ corresponds to the number of bootstrap resamples, and our selected sample size $\widehat{k} = \widehat{k}(C)$ is the total number of samples drawn, so that $\widehat{k}(C) \approx C \kappa$.

\subsection{Second Case Study: Concentration-Based Confidence Interval}\label{sec-KL}

Our analysis of the simulation sample size $\widehat{k}$ so far has been based on the specific CLT-based confidence interval \eqref{eqn-CI-intro-sim}. To demonstrate the generality of our framework, we now show that the same principles hold for a distinct type of confidence intervals constructed based on a Chernoff concentration bound. These intervals provide non-asymptotic, finite-sample coverage guarantees, which are often desirable in practice. The presentation in this section parallels that of Sections \myref{sec-MTurk} to \myref{sec-effective-sample-size}.

The confidence interval we study is based on a Chernoff concentration bound for Bernoulli random variables. It uses the \emph{Kullback-Leibler (KL) divergence}.

\begin{definition}[Kullback-Leibler divergence]
Let $p,q\in(0,1)$. The \textbf{Kullback-Leibler divergence} between two Bernoulli distributions $\Bernoulli(q)$ and $\Bernoulli(p)$ is defined by
\[
\KLdiv{q}{p} = q\log \bigg(\frac{q}{p}\bigg) + (1-q)\log \bigg(\frac{1-q}{1-p}\bigg) .
\]
\end{definition}

We state the Chernoff concentration bound as follows.

\begin{lemma}[Chernoff bound for Bernoulli]\label{lem-Ber-concentration-KL}
Fix $q\in(0,1)$. Let $\{x_i\}_{i=1}^k$ be i.i.d.~samples of $\Bernoulli(q)$. Define $\bar{x}_k = \frac{1}{k}\sum_{i=1}^k x_i$. For every $\alpha\in(0,1)$, 
\[
\PP\left( \KLdiv{\bar{x}_k}{q} \le \frac{\log(2/\alpha)}{k} \right) \ge 1-\alpha.
\]
\end{lemma}

\begin{proof}[Proof of \myCref{lem-Ber-concentration-KL}]
See \myCref{sec-lem-Ber-concentration-KL-proof}.
\end{proof}

By inverting \myCref{lem-Ber-concentration-KL}, we obtain a confidence interval of the following form:
\[
\left\{ p\in(0,1): \KLdiv{\bar{x}_k}{p} \le \frac{\log(2/\alpha)}{k} \right\},
\]
which is guaranteed to cover the true mean $q$ with probability at least $1-\alpha$, for every sample size $k\in\ZZ_+$. One may wonder how tight this confidence interval is. In \myCref{sec-lem-KL-Bernstein-proof}, we show that it is at least as tight as an empirical Bernstein concentration bound \citep{MPo09}, given by
\[
\left[ \simresponsebar_k - R_k, \simresponsebar_k + R_k \right],
\quad\text{with}\quad
R_k = \simsamplesd_k \sqrt{\frac{2\log(2/\alpha)}{k}} + \frac{2\log(2/\alpha)}{k},
\]
which shares a similar variance-dependent structure as the CLT-based confidence interval.

In our setting, given synthetic data $\{\simresponse_i\}_{i=1}^k$, we can form a dilated concentration-based confidence interval in the same way as the dilated CLT-based confidence interval \eqref{eqn-CI-intro-sim}:
\begin{equation}\label{eqn-CI-KL}
\simCIKL(k;C) = \left\{ p \in(0,1) : \KLdiv{\simresponsebar_k}{p} \le \frac{C\log(2/\alpha)}{k} \right\},
\end{equation}
with a dilation factor $C > 1$. Applying the same sample size selection method \eqref{eqn-empirical-criterion-1D} on the confidence interval $\simCIKL(k;C)$ yields a simulation sample size $\widehat{k}_{\KL}(C)$.

Paralleling our analysis in \myCref{sec-MTurk-IT} and \myCref{sec-effective-sample-size}, we will connect the hidden population size $\kappa$ and the normalized simulation sample size $\widehat{k}_{\KL}(C)/C$ by showing that they are both inversely proportional to an information-theoretic discrepancy measure between the real and synthetic distributions. The distribution discrepancy measure we use is the KL divergence:
\[
\discKL(\testfunction) = \KLdiv[\big]{\simmean(\testfunction)}{\mean(\testfunction)}.
\] 

We first show that the hidden population size $\kappa$ is inversely proportional to a quantile of $\discKLreal(\testfunction)$.

\begin{theorem}[Information-theoretic characterization of $\kappa$]\label{thm-kappa-IT-KL}
Let Assumption \myref{assumption-MTurk} hold. Choose $\delta\in(0,1)$. With probability at least $1-\delta$, it holds that
\[
\quantile_{1-\alpha} \big( \discKL(\testfunction) \mid \profiles_{\kappa} \big) \le \frac{c'(\alpha,\delta)}{\kappa},
\]
where $c'(\alpha,\delta)$ is a constant depending on $\alpha, \delta$.
\end{theorem}

\begin{proof}[Proof of \myCref{thm-kappa-IT-KL}]
See \myCref{sec-thm-kappa-IT-KL-proof}.
\end{proof}

We next turn to the selected sample size $\widehat{k}_{\KL}(C)$. Similar to \myCref{sec-effective-sample-size}, we consider an oracle version that assumes knowledge of the population coverage probability:
\begin{equation}\label{eqn-oracle-sample-size-KL}
k_{\KL}^*(C) = \sup \big\{ k \in \ZZ_+ : \PP\big( \mean(\testfunction) \in \simCIKL(k;C) \big) \ge 1-\alpha \big\}.
\end{equation}
\myCref{thm-oracle-effective-sample-size-KL} below characterizes the behavior of $k^*_{\KL}(C)$. We assume that both the real and synthetic response distributions are non-degenerate.

\begin{assumption}[Non-degenerate response distributions]\label{assumption-nondegeneracy-both}
There exists a constant $\eta\in(0,1/2]$ such that
\[
\PP\Big( \mean(\testfunction), \, \simmean(\testfunction) \in [\eta, \, 1-\eta] \Big) \ge 1-\alpha/8.
\]
Here the probability is taken over the randomness of the question $\testfunction$.
\end{assumption}

\begin{theorem}[Oracle simulation sample size]\label{thm-oracle-effective-sample-size-KL}
Let Assumption \myref{assumption-nondegeneracy-both} hold. Let the sample size $k_{\KL}^*(C)$ be defined by \eqref{eqn-oracle-sample-size-KL}. Then there exists $\varepsilon'(\alpha,C)$ such that for all $C > 1$,
\[
\frac{\log(2/\alpha)}{\quantile_{1-\alpha/2}\left( \discKLreal(\testfunction) \right)} - \varepsilon'(\alpha,C)
\le 
\frac{k_{\KL}^*(C)}{C}
\le 
\frac{\log(2/\alpha)}{\quantile_{1-2\alpha}\left( \discKLreal(\testfunction) \right)} + \varepsilon'(\alpha,C),
\]
and $\lim\limits_{C\to\infty}\varepsilon'(\alpha,C) = 0$.
\end{theorem}

\begin{proof}
See \myCref{sec-thm-oracle-effective-sample-size-KL-proof}.
\end{proof}

\myCref{thm-oracle-effective-sample-size-KL} has a similar form as \myCref{thm-oracle-effective-sample-size-chi}. It shows that $k_{\KL}^*(C)/C$ is approximately reciprocal to the LLM-human discrepancy $\quantile_{1-\alpha}\left( \discKLreal(\testfunction) \right)$. Combined with \myCref{thm-kappa-IT-KL}, we see that $k_{\KL}^*(C)/C\gtrsim \kappa$ under the Mechanical Turk model. We remark that \myCref{thm-oracle-effective-sample-size-KL} itself does not rely on the Mechanical Turk interpretation.

\section{Numerical Experiments}\label{sec-experiments}

In this section, we apply both our methods in \myCref{sec-warmup} and \myCref{sec-general} to LLMs over real datasets, to verify the coverage validity of our uncertainty quantification framework and to quantify the human-LLM discrepancy across LLMs and domains. The Python code for reproducing the results is available at \url{https://github.com/yw3453/uq-llm-survey-simulation}.

\subsection{Experiment Setup}

Below we describe the datasets, LLMs, confidence intervals and hyperparameters used in the experiments.

\paragraph{Datasets.} We use two datasets for survey questions, each corresponding to one uncertainty quantification task. The first dataset is the OpinionQA dataset \citep{SDL23}. It was built from Pew Research's American Trends Panel \citep{Pew25}, and contains the general US population's responses to survey questions spanning topics such as science, politics, and health. After pre-processing we have 385 unique questions and 1,476,868 responses to these questions from at least 32,864 people. These questions have $5$ choices corresponding to ordered sentiments, which we map to sentiment scores $-1,-\frac{1}{3},0,\frac{1}{3},1$, respectively. Each question has at least $400$ responses. For each response, we have information on the surveyee's political profile, religious affiliation, educational background, socio-economic status, etc., which we use to generate synthetic profiles. See \myCref{sec-opinion} for more details. We consider the task of constructing a confidence interval on the U.S.~population's average sentiment score for a survey question. This is the setup in \myCref{example-public-survey-1D}.

The second dataset is the EEDI dataset created by \cite{HMG24}, which was built upon the NeurIPS 2020 Education Challenge dataset \citep{WLS21}. It consists of students' responses to mathematics multiple-choice questions on the Eedi online educational platform \citep{Eedi25}. The dataset contains 573 unique questions and 443,433 responses to these questions from 2,287 students. All questions have four choices (A, B, C, D). Out of these questions, we use questions that have at least $100$ student responses. Excluding questions with graphs or diagrams, we are left with a total of $412$ questions. For each student, we have information on their gender, age, and socioeconomic status, which we use to generate synthetic profiles. See \myCref{sec-eedi} for more details. We consider the task of constructing a confidence interval on the probability of a student answering a question correctly. This is similar to the setup in \myCref{sec-warmup}.

\paragraph{LLMs.} We consider $8$ LLMs: Claude 3.5 Haiku \citep{Ant24}, DeepSeek-V3 \citep{LFX24}, GPT-3.5 Turbo \citep{GPT3.5}, GPT-4o \citep{GPT4o}, GPT-4o mini \citep{GPT4omini}, GPT-5 mini \citep{GPT5}, Llama 3.3 70B \citep{DJP24}, and Mistral 7B \citep{JSM23}.

\paragraph{Synthetic confidence set construction.} Our method can be built upon any procedure $\setmap$ for constructing the synthetic confidence set. We will use confidence intervals derived from the empirical Bernstein concentration bound (see Theorem 11 in \cite{MPo09}):
\begin{equation}\label{eqn-CI-Bern-sim}
\simCIBern(k) =  \left[ \simresponsebar_k - r_k, \simresponsebar_k + r_k \right],
\quad\text{with}\quad
r_k = \simsamplesd_k \sqrt{\frac{2C\log(4/\alpha)}{k}} + \frac{7CM\log(4/\alpha)}{3(k-1)},
\end{equation}
where $C>1$ is the dilation factor and $M$ is the numerical range of the responses. It is a finite-sample analogue to the CLT-based interval \eqref{eqn-CI-intro-sim} analyzed in our theory in \myCref{sec-interpretations}. The primary advantage is that the empirical Bernstein bound offers a valid coverage guarantee for any sample size $k$, while the CLT-based confidence interval only has an asymptotic coverage guarantee. It achieves this finite-sample coverage by replacing the normal quantile $z_{\alpha/2}$ with the term $\sqrt{\log(2/\alpha)}$, and including an additional $O(1/k)$ term, which safeguards coverage even when the sample standard deviation $\simsamplesd_k$ is not an accurate approximation of the true standard deviation for small $k$. For completeness, we formally state the finite-sample coverage guarantee of $\simCIBern(k)$ in \myCref{lem-emp-Bernstein}. 

\paragraph{Hyperparameters.} We consider target miscoverage levels $\alpha\in\{0.05,0.1,0.15,0.2\}$, dilation factor $C=2$, and confidence level $\gamma=0.5$ for confidence intervals constructed from human responses in \eqref{eqn-CI-calibrate}. We take the range $M=2$ for the OpinionQA dataset since the responses range within $[-1,1]$, and $M=1$ for the EEDI dataset since the responses are binary. The simulation budget $K$ is set to be sufficiently large, so that the selected simulation sample size $\widehat{k}$ is not limited by the budget.

\subsection{Experiment Procedure}

We now describe our experiment procedure for applying the methods in \myCref{sec-warmup} and \myCref{sec-general} to each dataset. We refer the former as the \emph{simple method}, and the latter as the \emph{general method}. Denote the dataset by $\{ ( \testfunction_j , \dataset_j ) \}_{j=1}^{J}$, where $\testfunction_j$ is a survey question and $\dataset_j = \{ \response_{j,i} \}_{i=1}^{n_j}$ is a collection of human responses to the question $\testfunction_j$. For each $j\in[J]$, we simulate $K$ responses $\simdataset_j$ from an LLM. We then randomly split $ \datasetmeta = \{ ( \dataset_j, \simdataset_j ) \}_{j=1}^{J}$ into a training set $\datasetmeta^{\train} = \{ ( \dataset_j, \simdataset_j ) \}_{j\in\cJ_{\train}}$ and a testing set $\datasetmeta^{\test} = \{ ( \dataset_j, \simdataset_j ) \}_{j\in\cJ_{\test}}$, with $|\datasetmeta^{\train}| : |\datasetmeta^{\test}| = 3 : 2$. We run experiments over $100$ random train-test splits of the survey questions.

\paragraph{Selection of simulation sample size.} We apply both the simple method in \myCref{sec-warmup} and general method in \myCref{sec-general} with the training set $\datasetmeta^{\train}$ to select a simulation sample size $\widehat{k}$. For the general method, to construct the confidence set $\CIalt_j$ in \eqref{eqn-CI-calibrate}, we use the standard CLT-based confidence interval:
\begin{equation}\label{eqn-CI-calibrate-CLT}
\CIalt_j = \bigg[ \responsebar_j - \frac{z_{(1-\gamma)/2} \cdot \samplesd_j}{\sqrt{n_j}} , ~
\responsebar_j + \frac{z_{(1-\gamma)/2} \cdot \samplesd_j }{\sqrt{n_j}} \bigg],
\end{equation}
where $\responsebar_j = \frac{1}{n_j} \sum_{i=1}^{n_j} \response_{j,i}$ and $\samplesd_j = \sqrt{\responsebar_j (1-\responsebar_j) }$. Since $n_j$ is at least $100$, $\CIalt_j$ has approximately $\gamma$ coverage probability.

\paragraph{Evaluation of selected confidence interval.} We use $\datasetmeta^{\test}$ to evaluate the coverage validity of the chosen confidence interval $\simCIBern(\widehat{k})$. As the true population mean $\mean$ is unavailable, the true coverage probability $\PP\big( \mean \in \simCIBern(\widehat{k}) \big)$ cannot be computed. However, we can apply the same ideas as \eqref{eqn-proxy-1D} in \myCref{sec-method-1D} and \eqref{eqn-proxy} in \myCref{sec-general} to compute a proxy for the miscoverage level. 

To evaluate the simple method, for each survey question $j\in\cJ_{\test}$, the selected sample size $\widehat{k}$ leads to a synthetic confidence interval $\simCIBern_j(\widehat{k})$. We form the sample mean $\responsebar_j$ from real data $\dataset_j$ and define
\begin{equation}\label{eqn-emp-proxy}
\coverage_{\test}(k) =  \frac{2}{|\cJ_{\test}|} \sum_{j\in\cJ_{\test}} \ind \big\{ \responsebar_j\not\in \simCIBern_j(k) \big\}.
\end{equation}
The proof of \myCref{thm-coverage-1D} shows that under the CLT approximation $\PP(\responsebar_j\le \mean_j\mid\testfunction_j)=1/2$, for every $k\in[K]$ and survey question $j$,
\[
\frac{1}{2} \PP\big( \mean \not\in \simCIBern(k) \big) \le \PP\big( \responsebar_j\not\in \simCIBern_j(k) \big) = \frac{1}{2} \EE \big[ \coverage_{\test}(k) \big] .
\]
Thus, if $\EE \big[ \coverage_{\test}(\widehat{k}) \big] \le \alpha$, then $\PP\big( \mean \not\in \simCIBern(\widehat{k})\big) \le \alpha$ must hold.

Similarly, to evaluate the general method, for each survey question $j\in\cJ_{\test}$, we form the confidence set $\CIalt_j$ from real data $\dataset_j$ as in \eqref{eqn-CI-calibrate-CLT}. We then define
\begin{equation}\label{eqn-emp-proxy-general}
\coveragealt_{\test}(k) = \frac{1}{\gamma} \cdot \frac{1}{|\cJ_{\test}|} \sum_{j\in\cJ_{\test}} \ind \big\{ \CIalt_j\not\subseteq \simCIBern_j(k) \big\}.
\end{equation}
The proof of \myCref{thm-coverage} shows that, for every $k\in[K]$ and survey question $j$,
\[
\gamma \cdot \PP\big( \mean \not\in \simCIBern(k) \big) \le \PP\big( \CIalt_j\not\subseteq \simCIBern_j(k) \big) = \gamma \cdot \EE \big[ \coveragealt_{\test}(k) \big] .
\]
Thus, if $\EE \big[ \coveragealt_{\test}(\widehat{k}) \big] \le \alpha$, then $\PP\big( \mean \not\in \simCIBern(\widehat{k})\big) \le \alpha$ must hold.

\subsection{Experiment Results}\label{sec-experiments-results}

For different LLMs on both datasets, we evaluate the following metrics:
\begin{enumerate}
\item miscoverage probability proxies $\coverage_{\test}(\widehat{k})$ and $\coveragealt_{\test}(\widehat{k})$,
\item sharpness of $\widehat{k}$, through comparison with an oracle sample size $k_{\test}^* = \max\{k : \coverage_{\test}(k) \le \alpha \}$,
\item width of the selected confidence interval $\simCIBern(\widehat{k})$,
\item estimated hidden population size $\widehat{\kappa} = \widehat{k}/C$.
\end{enumerate}
The first two metrics are used to verify the validity and sharpness of our sample size selection methods, and the last two metrics provide insights into the LLMs' simulation fidelity on the datasets. The results for the general method are very similar to the ones for the simple method. Therefore, for space considerations, we will only present results for the simple method in the main text, and defer results for the general method to \myCref{sec-appendix-experiments-results}.

\paragraph{Coverage validity.} In \myCref{fig-miscoverage-simple}, we plot the miscoverage probability proxy $\coverage_{\test}(\widehat{k})$ for different LLMs on the two datasets, averaged over $100$ random train-test splits of the questions. The half-width of each error bar is 1.96 times the standard error. We observe that $\coverage_{\test}(\widehat{k})$ is consistently close to or below the target miscoverage level $\alpha$, which shows that the selected confidence intervals $\simCIBern(\widehat{k})$ have approximately $(1-\alpha)$ coverage probabilities. This verifies the coverage guarantee of our method in \myCref{thm-coverage-1D}.

\begin{figure}[h]
	\FIGURE{
    \subcaptionbox{OpinionQA Dataset}
    {\includegraphics[width=0.46\linewidth]{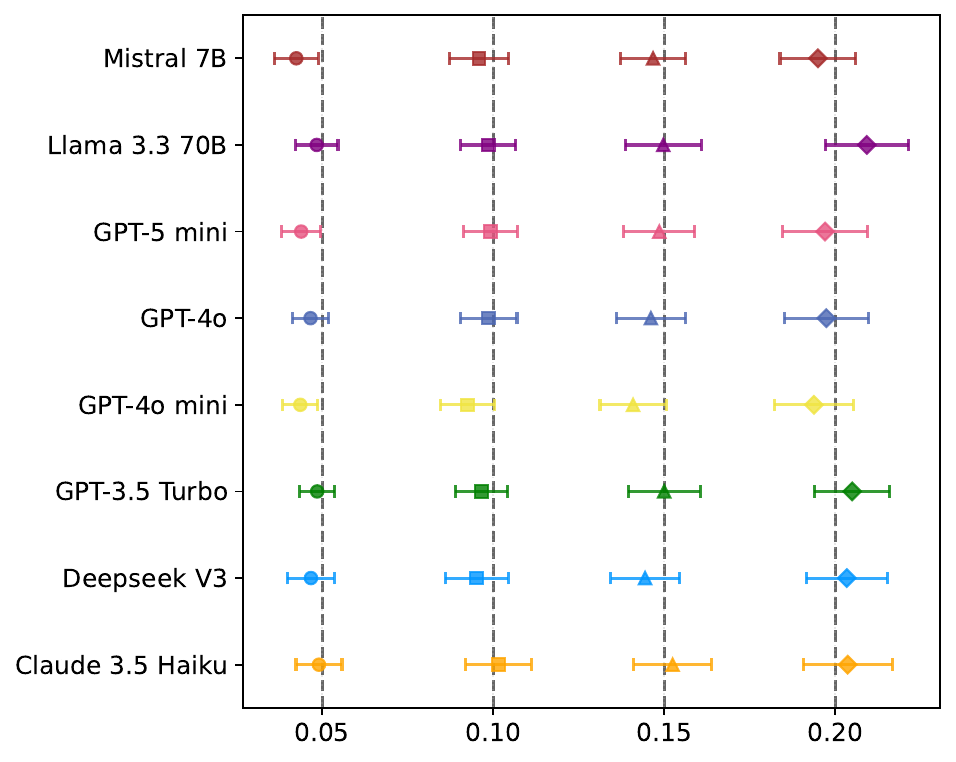}}
    \hfill\subcaptionbox{EEDI Dataset}
    {\includegraphics[width=0.46\linewidth]{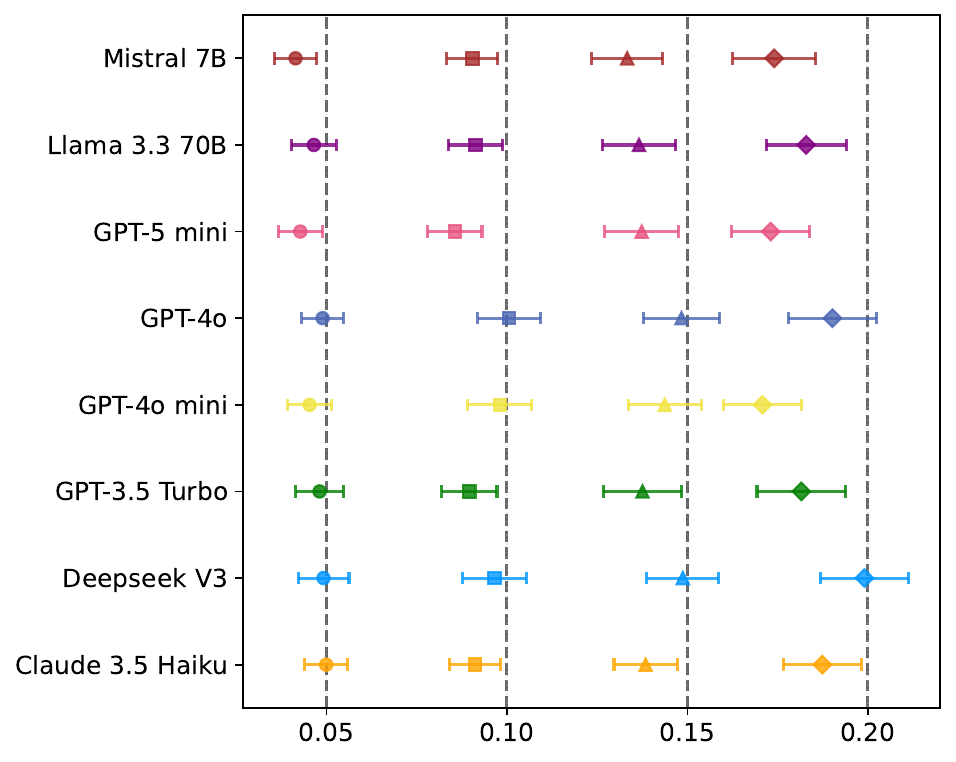}}
	}
	{Miscoverage Probability Proxy $\coverage_{\test}(\widehat{k})$ for Different LLMs and Target Miscoverage Levels $\alpha$. \label{fig-miscoverage-simple}}
	{Horizontal axis: target miscoverage level $\alpha$. Vertical axis: LLM. Circles, squares, triangles and diamonds represent $\coverage_{\test}(\widehat{k})$ for $\alpha=0.05,0.1,0.15,0.2$, respectively.}
\end{figure}

\paragraph{Sharpness of selected sample size.} To complement the coverage validity of the selected confidence interval, we demonstrate that our method does not select an overly conservative confidence interval $\simCIBern(\widehat{k})$ with an overly small sample size $\widehat{k}$. To that end, we compare $\widehat{k}$ with an oracle sample size $k_{\test}^* = \max\{k \in \ZZ_+ : \coverage_{\test}(k) \le \alpha \}$, which is the maximum sample size that guarantees $(1-\alpha)$ proxy coverage over the testing set of survey questions. \myCref{fig-sharpness-simple} plots the histogram of the relative error $\big|\widehat{k} - k^*_{\test}\big| / k_{\test}^*$ over $100$ random train-test splits, for GPT-4o on the OpinionQA dataset. We observe that the relative error between $\widehat{k}$ and $k_{\test}^*$ is within $0.25$ for at least $95\%$ of the random train-test splits. This shows that our method consistently selects a sample size $\widehat{k}$ close to $k^*_{\test}$. In \myCref{sec-appendix-sharpness-simple}, we provide results on the relative error $\big|\widehat{k} - k^*_{\test}\big| / k_{\test}^*$ for other LLMs and $\alpha$'s.

\begin{figure}[h]
\FIGURE
{\includegraphics[scale=0.5]{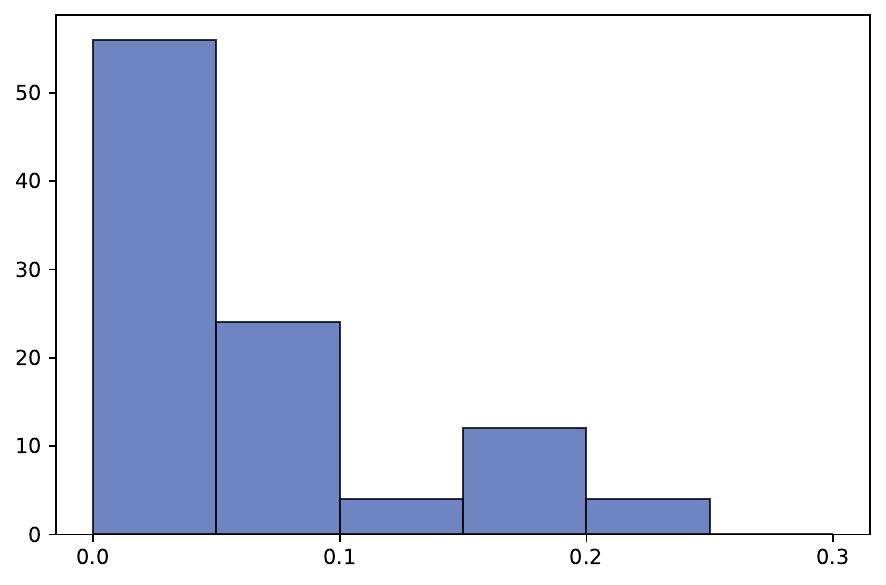}}
{Histogram of the Relative Error $\big|\widehat{k} - k^*_{\test}\big| / k_{\test}^*$ for GPT-4o on the OpinionQA Dataset. \label{fig-sharpness-simple}}
{Horizontal axis: relative error. Vertical axis: frequency. The histogram is over $100$ train-test splits.}
\end{figure}

Having validated our method, we now use the other two metrics, the interval width and the estimated hidden population size, to reveal the simulation fidelity of the LLMs.

\paragraph{Interval width.} In \myCref{fig-width-simple}, we plot the widths of the selected confidence intervals $\simCIBern(\widehat{k})$ for different LLMs and target miscoverage levels $\alpha$. As discussed in \myCref{sec-warmup}, the interval width reflects the human-LLM discrepancy. A shorter interval indicates that the LLM-generated responses support more precise inference about the human population, while a wider interval indicates lower LLM simulation fidelity. In simulating social opinions (\myCref{fig-width-simple-OpinionQA}), GPT-4o has the smallest misalignment gap among the LLMs considered. In simulating middle-school student answers to mathematics questions (\myCref{fig-width-simple-EEDI}), all tested LLMs yield wide confidence intervals, indicating substantial human-LLM misalignment.

We reiterate that our goal is to quantify the uncertainty induced by human-LLM discrepancy. Our method does not aim to improve an LLM's simulation fidelity. Instead, it uses confidence intervals to gauge this discrepancy. Thus, wide confidence intervals should be interpreted as evidence of low simulation fidelity among the tested LLMs, rather a failure of uncertainty quantification.

\begin{figure}[h]
	\FIGURE{
    \subcaptionbox{OpinionQA Dataset \label{fig-width-simple-OpinionQA}}
    {\includegraphics[width=0.49\linewidth]{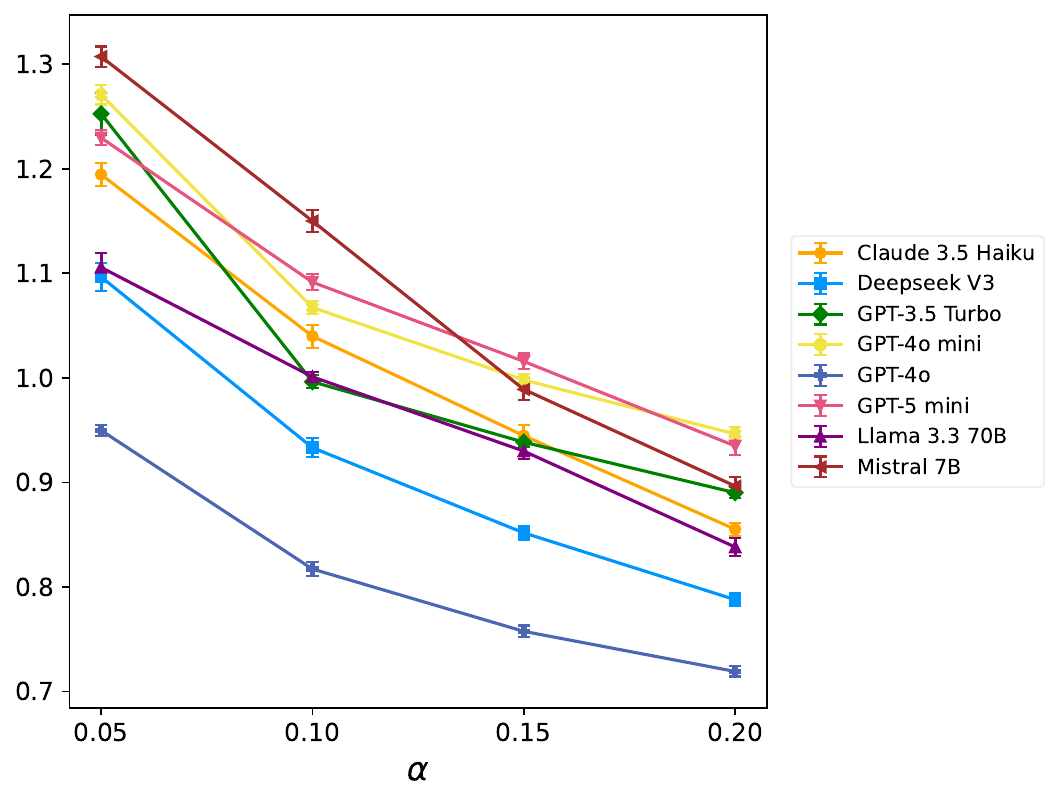}}
    \hfill\subcaptionbox{EEDI Dataset \label{fig-width-simple-EEDI}}
    {\includegraphics[width=0.49\linewidth]{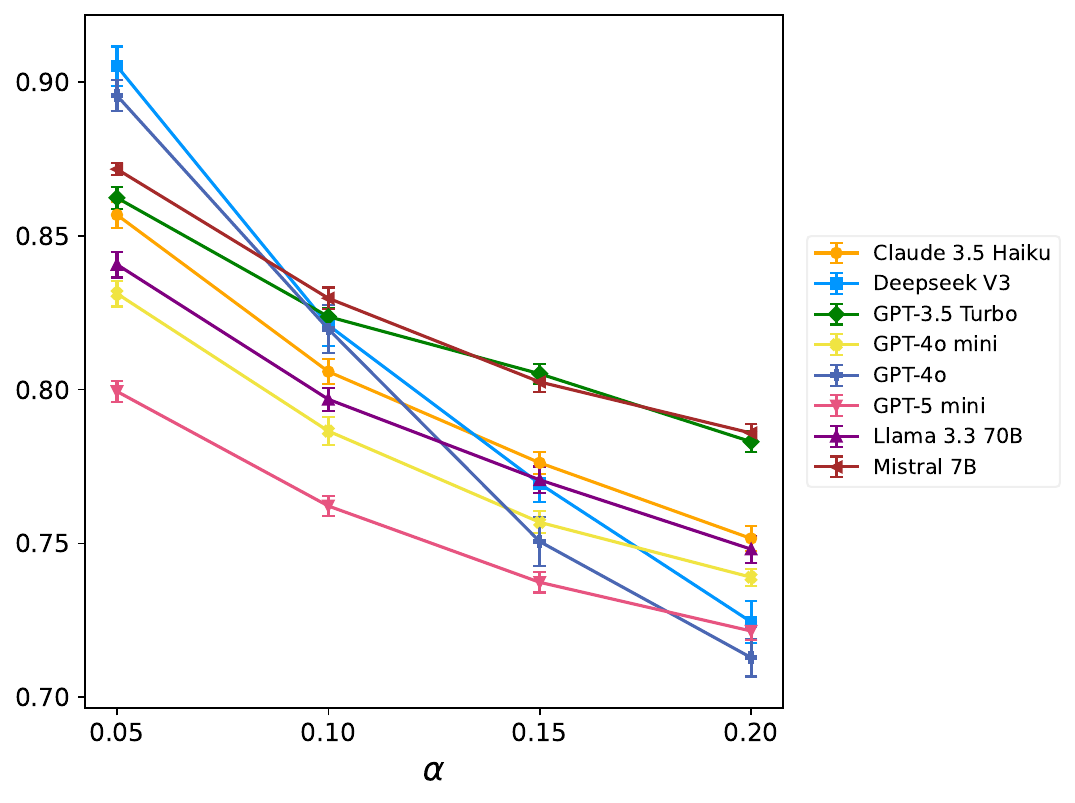}}
	}
	{Widths of Confidence Intervals $\simCIBern(\widehat{k})$ for Different LLMs and Target Levels $\alpha$. \label{fig-width-simple}}
	{Horizontal axis: target miscoverage level $\alpha$. Vertical axis: width of $\simCIBern(\widehat{k})$. The results are averaged over $100$ random train-test splits. The half-width of each error bar is 1.96 times the standard error.}
\end{figure}

\paragraph{Estimated hidden population size $\widehat{\kappa} = \widehat{k}/C$.} In \myCref{fig-kappa-simple}, we report $\widehat{\kappa} = \widehat{k}/C$ for various LLMs on the OpinionQA and EEDI datasets, for $\alpha=0.05$; in \myCref{sec-ablation-C} we examine the sensitivity of $\widehat{k}/C$ to the dilation factor $C$, where the estimates are shown to be stable. As discussed in \myCref{sec-interpretations}, the hidden population size $\kappa$ reflects the effective number of human respondents represented by the LLM. A larger $\widehat{\kappa}$ indicates that the LLM captures more information about the human population and therefore has higher simulation fidelity. On the OpinionQA dataset (\myCref{fig-kappa-simple-opinionQA}), GPT-4o has the largest estimated $\widehat{\kappa}$ of around $58$, followed by DeepSeek-V3 and GPT-3.5 Turbo. On the EEDI dataset (\myCref{fig-kappa-simple-EEDI}), the estimated $\widehat{\kappa}$ values are uniformly smaller, around or below $20$. These small effective sample sizes indicate limited simulation fidelity for middle-school mathematics questions. Comparing across the two datasets shows that the tested LLMs have substantially lower simulation fidelity for middle-school mathematics questions than for social opinions.

\begin{figure}[H]
	\FIGURE{
    \subcaptionbox{OpinionQA Dataset\label{fig-kappa-simple-opinionQA}}
    {\includegraphics[width=0.45\linewidth]{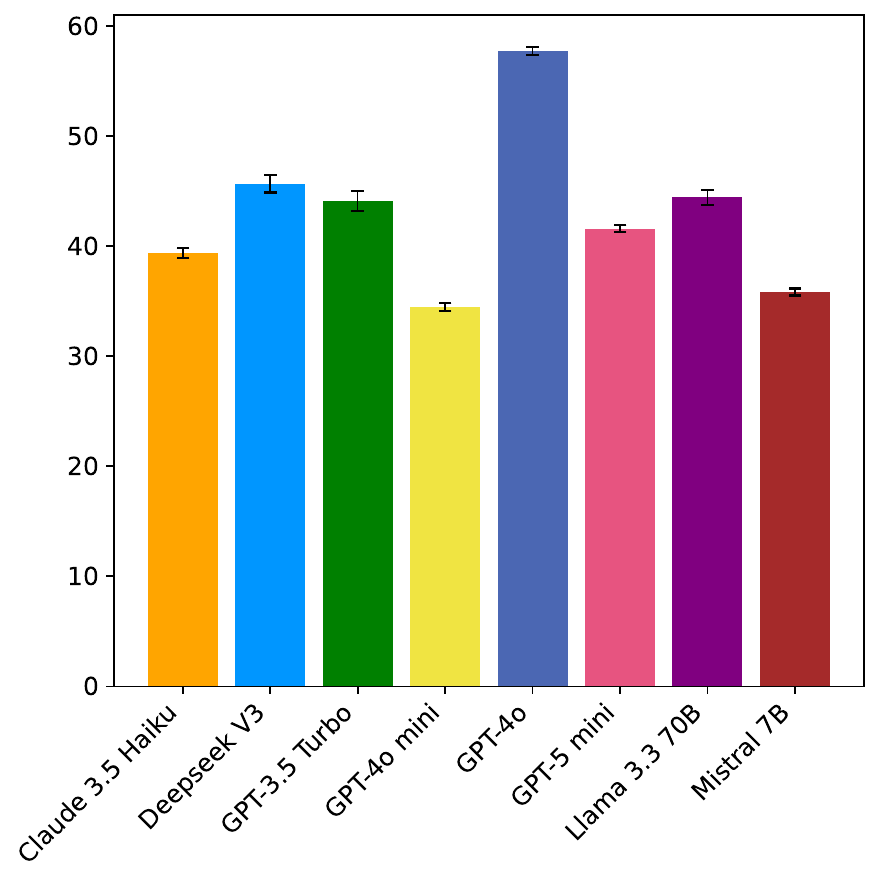}}
    \qquad\subcaptionbox{EEDI Dataset\label{fig-kappa-simple-EEDI}}
    {\includegraphics[width=0.45\linewidth]{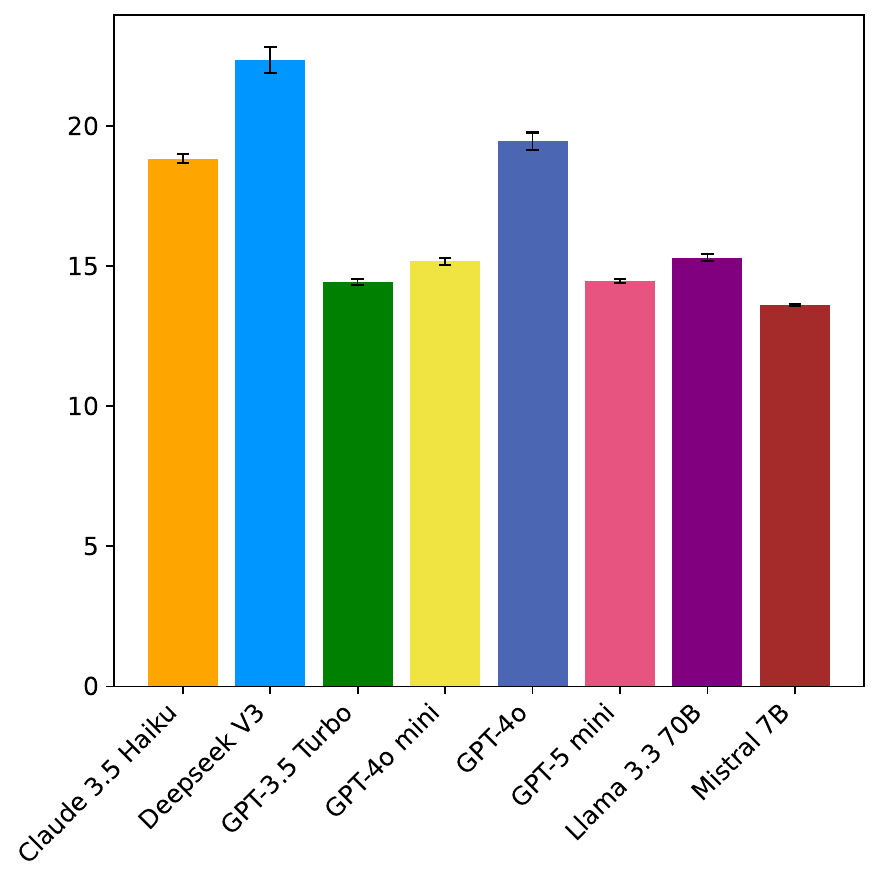}}
	}
	{Estimated Hidden Population Sizes $\widehat{\kappa}=\widehat{k}/C$ of Different LLMs. \label{fig-kappa-simple}}
	{The results are averaged over $100$ train-test splits. The half-width of each error bar is 1.96 times the standard error.}
\end{figure}

Our experiment results highlight a pitfall of using LLM-generated data. The ease of LLM simulation makes it tempting to generate a large number of responses. However, increasing the number of synthetic responses reduces only the simulation noise and does not eliminate human-LLM misalignment. Our results show that the tested LLMs have limited simulation fidelity, and that the fidelity is highly heterogeneous across both LLMs and domains. Consequently, na\"{i}vely generating too many synthetic samples can lead to overconfident and biased inference. Our framework provides a disciplined approach for quantifying this uncertainty and adaptively selecting an appropriate simulation sample size based on the LLM's fidelity.

\section{Discussions}\label{sec-discussions}

We developed a general framework for converting imperfect LLM-based survey simulations into statistically valid confidence sets for population parameters of human responses. These confidence sets quantify the uncertainty induced by human-LLM misalignment, and provide valid inference for the true human population parameter. Our approach identifies a simulation sample size for valid uncertainty quantification. The selected sample size is further shown to correspond to an effective human sample size and a fidelity measure of the LLM. Numerical experiments on multiple LLMs and real datasets verified the coverage guarantees of our approach and revealed substantial heterogeneity in the simulation fidelity across LLMs and domains.

Several future directions are worth exploring. First, our approach does not explicitly minimize the size of the confidence set. A natural question is whether we can incorporate a size minimization procedure to produce smaller confidence sets while maintaining valid coverage. Second, it would be interesting to see if our approach can be combined with debiasing methods to give more informative confidence sets. Finally, as prompt engineering is known to have crucial effects on the quality of LLM simulations, it is worth investigating the impacts of prompts on the selected simulation sample size $\widehat{k}$, and how prompt engineering can be leveraged to improve the fidelity of LLM simulations.

\section*{Acknowledgement}
Chengpiao Huang and Kaizheng Wang's research is supported by NSF grants DMS-2210907 and DMS-2515679, and a Data Science Institute seed grant SF-181 at Columbia University.

\newpage 

\appendix
\crefalias{section}{appendix}
\crefalias{subsection}{appendix}

\section{Proofs of Coverage Guarantees} 

In this appendix, we prove the coverage guarantees (\myCref{thm-coverage-1D} and \myCref{thm-coverage}) for our methods.

\subsection{Proof of \myCref{thm-coverage-1D}}\label{sec-thm-coverage-1D-proof}

We will prove the following stronger guarantee.

\begin{lemma}[Conditional coverage]\label{lem-conditional-coverage-1D}
Consider the setting of \myCref{thm-coverage-1D}. Let $\delta\in(0,1)$. With probability at least $1-\delta$,
\begin{equation}\label{eqn-conditional-coverage-1D}
\PP\Big( \mean \in \simCI(\widehat{k}) \Bigm| \widehat{k} \Big) \ge 1 - (1-2\eta)^{-1} \left( \alpha + \sqrt{\frac{2\log(1/\delta)}{m}} \right).
\end{equation}
\end{lemma}

By \myCref{lem-conditional-coverage-1D}, we obtain
\begin{align*}
\PP\Big( \mean \in \simCI(\widehat{k}) \Big)
&=
\EE\left[ \PP \Big( \mean \in \simCI(\widehat{k}) \Bigm| \widehat{k} \Big) \right]
=
\int_0^{\infty} \PP \left( \PP \Big( \mean \in \simCI(\widehat{k}) \Bigm| \widehat{k} \Big) > t \right) \,dt \\
&\ge 
\int_0^{1-\alpha/(1-2\eta)} \left[ 1 - \exp\left( -\frac{m(1-2\eta)^2}{2} \Big(t - \Big(1-\frac{\alpha}{1-2\eta} \Big) \Big)^2 \right) \right] \,dt \\
&\ge 
1- \frac{\alpha}{1-2\eta} - \int_{-\infty}^{1-\alpha/(1-2\eta)} \exp\left( -\frac{m(1-2\eta)^2}{2} \Big(t - \Big(1-\frac{\alpha}{1-2\eta} \Big) \Big)^2 \right) \,dt \\
&\ge 
1- (1-2\eta)^{-1} \bigg( \alpha + \sqrt{\frac{2}{m}} \bigg).
\end{align*}

We will now prove \myCref{lem-conditional-coverage-1D}. Define $\varepsilon = \sqrt{2\log(1/\delta) / m}$ and a deterministic oracle sample size
\begin{equation}\label{eqn-oracle-k-1D}
\bar{k} = \inf \left\{ k\in [K] : \PP \Big( \mean \not\in \simCI(k) \Big)  
> (1-2\eta)^{-1} \left(\alpha + \varepsilon\right) \right\}.
\end{equation}
If $\bar{k} = \inf\emptyset$ does not exist, then there is nothing to prove. Now suppose that $\bar{k} \in [K]$ exists. We will prove that with probability at least $1-\delta$, it holds that $\coverage(\bar{k}) > \alpha/2$. When this event happens, we have $\widehat{k} < \bar{k}$, which implies $\PP\big( \mean \not\in \simCI(\widehat{k}) \bigm| \widehat{k} \big) \le (1-2\eta)^{-1}(\alpha+\varepsilon)$ and thus \eqref{eqn-conditional-coverage-1D}, thanks to the independence of $\widehat{k}$ and $(\testfunction,\simdataset)$.

By Hoeffding's inequality (e.g., Theorem 2.8 in \cite{BLM13}) and the conditional independence of $(\dataset_1,\simdataset_1),...,(\dataset_m,\simdataset_m)$ given $(\testfunction_1,...,\testfunction_m)$,
\begin{equation}\label{eqn-Hoeffding-1D}
\PP\left( \coverage(\bar{k}) \ge \frac{1}{m} \sum_{j=1}^m \PP\Big( \responsebar_j \not\in \simCI_j(\bar{k}) \Big) -  \sqrt{\frac{\log(1/\delta)}{2m}} \right) \ge 1-\delta.
\end{equation}
We now bound $\PP\big( \responsebar_j \not\in \simCI_j(\bar{k}) \big)$. For each $j\in[m]$ and $k\in[K]$,
\[
\ind \left\{ \responsebar_j \not\in \simCI_j(k) \right\}
\ge 
\ind \left\{ \responsebar_j < \mean_{j} \text{ and } \mean_{j} < \min \simCI_j(k) \right\}
+
\ind \left\{ \responsebar_j \ge \mean_{j} \text{ and } \mean_{j} > \max \simCI_j(k) \right\}.
\]
By the conditional independence of $\dataset_j$ and $\simdataset_j$ given $\testfunction_j$,
\begin{align*}
\PP \Big( \responsebar_j < \mean_{j} \text{ and } \mean_{j} < \min \simCI_j(k) \Big)
&=
\EE\Big[ \PP \Big( \responsebar_j < \mean_{j} \Bigm| \testfunction_j \Big) \cdot \PP \Big( \mean_{j} < \min \simCI_j(k)  \Bigm| \testfunction_j \Big) \Big] \\
&\ge
\EE\left[ \left(\frac{1}{2} - \eta \right) \cdot \PP \Big( \mean_{j} < \min \simCI_j(k)  \Bigm| \testfunction_j \Big) \right] \\
&=
\left(\frac{1}{2} - \eta \right) \cdot \PP\Big( \mean_{j} < \min \simCI_j(k)  \Big).
\end{align*}
Similarly,
\[
\PP \Big( \responsebar_j \ge \mean_{j} \text{ and } \mean_{j} > \max \simCI_j(k) \Big)
\ge
\left( \frac{1}{2} - \eta \right) \cdot \PP\Big( \mean_{j} > \max \simCI_j(k) \Big).
\]
Therefore,
\begin{align}
\PP \Big( \responsebar_j \not\in \simCI_j(k) \Big)
&\ge 
\left( \frac{1}{2} - \eta \right) \cdot \left[ \PP\Big( \mean_{j} < \min \simCI_j(k) \Big) + \PP\Big( \mean_{j} > \max \simCI_j(k) \Big) \right] \notag \\
&=
\left( \frac{1}{2} - \eta \right) \cdot \PP \Big( \mean_{j} \not\in \simCI_j(k) \Big) \notag \\
&=
\left( \frac{1}{2} - \eta \right) \cdot \PP \Big( \mean \not\in \simCI(k) \Big). \label{eqn-proxy-to-oracle-precise-1D}
\end{align}
When the event in \eqref{eqn-Hoeffding-1D} happens,
\begin{align*}
\coverage(\bar{k}) 
&
\ge \frac{1}{m} \sum_{j=1}^m \PP\Big( \responsebar_j \not\in \simCI_j(\bar{k}) \Big) -  \sqrt{\frac{\log(1/\delta)}{2m}} \\
&\ge 
\left( \frac{1}{2} - \eta \right) \cdot \PP \Big( \mean \not\in \simCI(\bar{k}) \Big) -  \sqrt{\frac{\log(1/\delta)}{2m}} \tag{by \eqref{eqn-proxy-to-oracle-precise-1D}} \\
&> 
\frac{\alpha}{2}. \tag{by definition of $\bar{k}$}
\end{align*}

\subsection{Proof of \myCref{thm-coverage}}\label{sec-thm-coverage-proof}

We will prove the following stronger guarantee.

\begin{lemma}[Conditional coverage]\label{lem-conditional-coverage}
Consider the setting of \myCref{thm-coverage}. Let $\delta\in(0,1)$. With probability at least $1-\delta$,
\begin{equation}\label{eqn-conditional-coverage}
\PP\Big( \statistic(\testfunction)  \in \simCIalt(\widehat{k}) \Bigm| \widehat{k} \Big) \ge 1 - \alpha - \gamma^{-1} \sqrt{\frac{\log(1/\delta)}{2m}}.
\end{equation}
\end{lemma}

By \myCref{lem-conditional-coverage}, we obtain
\begin{align*}
\PP\Big( \statistic (\testfunction)  \in \simCIalt(\widehat{k}) \Big)
&=
\EE\left[ \PP \Big( \statistic (\testfunction)  \in \simCIalt(\widehat{k}) \Bigm| \widehat{k} \Big) \right] \\
&=
\int_0^{\infty} \PP \left( \PP \Big( \statistic (\testfunction)  \in \simCIalt(\widehat{k}) \Bigm| \widehat{k} \Big) > t \right) \,dt \\
&\ge 
\int_0^{1-\alpha} \left[ 1 - \exp\left( -2m\gamma^2 \big(t - (1-\alpha) \big)^2 \right) \right] \,dt \\
&\ge 
1-\alpha - \int_{-\infty}^{1-\alpha} \exp\left( -2m\gamma^2 \big(t - ( 1-\alpha) \big)^2 \right) \,dt \\
&\ge 
1-\alpha - \gamma^{-1} \sqrt{\frac{1}{2m}}.
\end{align*}

We now prove \myCref{lem-conditional-coverage}. Define $\varepsilon = \gamma^{-1} \sqrt{\frac{\log(1/\delta)}{2m}} $ and a deterministic oracle sample size
\begin{equation}\label{eqn-oracle-k}
\bar{k} = \inf \left\{ k\in [K] : \PP \Big( \statistic (\testfunction)  \not\in \simCIalt(k) \Big)  
> \alpha + \varepsilon \right\}.
\end{equation}
If $\bar{k} = \inf \emptyset$ does not exist, then there is nothing to prove. Now suppose $\bar{k} \in [K]$ exists. We will prove that with probability at least $1-\delta$, it holds that $\coveragealt(\bar{k}) > \gamma\alpha$. When this event happens, we have $\widehat{k} < \bar{k}$, which implies $\PP\big( \statistic (\testfunction)  \not\in \simCIalt(\widehat{k}) \bigm| \widehat{k} \big) \le \alpha+\varepsilon$ and thus \eqref{eqn-conditional-coverage}, thanks to the independence of $\widehat{k}$ and $(\testfunction,\simdataset)$.

By Hoeffding's inequality (e.g., Theorem 2.8 in \cite{BLM13}) and the conditional independence of $(\dataset_1,\simdataset_1),...,(\dataset_m,\simdataset_m)$ given $(\testfunction_1,...,\testfunction_m)$,
\begin{equation}\label{eqn-Hoeffding}
\PP\left( \coveragealt(\bar{k}) \ge \frac{1}{m} \sum_{j=1}^m \PP\Big( \CIalt_j \not\subseteq \simCIalt_j(\bar{k}) \Big) -  \sqrt{\frac{\log(1/\delta)}{2m}} \right) \ge 1-\delta.
\end{equation}
We now bound $\PP\big( \CIalt_j \not\subseteq \simCIalt_j(\bar{k}) \big)$. For each $j\in[m]$ and $k\in[K]$,
\[
\ind \left\{ \CIalt_j \not\subseteq \simCIalt_j(k) \right\}
\ge 
\ind \left\{ \statistic (\testfunction_j) \in \CIalt_j \text{ and } \statistic (\testfunction_j) \not\in \simCIalt_j(k) \right\}.
\]
By the conditional independence of $\dataset_j$ and $\simdataset_j$ given $\testfunction_j$,
\begin{align}
\PP \Big( \CIalt_j \not\subseteq \simCIalt_j(k) \Big)
&\ge
\EE\Big[ \PP \Big( \statistic (\testfunction_j) \in \CIalt_j \text{ and } \statistic (\testfunction_j) \not\in \simCIalt_j(k)  \Bigm| \testfunction_j \Big) \Big] \notag \\
&=
\EE\Big[ \PP \Big( \statistic (\testfunction_j) \in \CIalt_j \Bigm| \testfunction_j \Big) \cdot \PP \Big( \statistic (\testfunction_j) \not\in \simCIalt_j(k)  \Bigm| \testfunction_j \Big) \Big] \notag \\
&\ge
\EE\Big[ \gamma \cdot \PP \Big( \statistic (\testfunction_j) \not\in \simCIalt_j (k) \Bigm| \testfunction_j \Big) \Big] \notag \\
&=
\gamma \cdot \PP\Big( \statistic (\testfunction_j) \not\in \simCIalt_j(k)  \Big) \notag \\
&=
\gamma \cdot \PP\Big( \statistic (\testfunction)  \not\in \simCIalt(k)  \Big). \label{eqn-proxy-to-oracle-precise}
\end{align}
Therefore, when the event in \eqref{eqn-Hoeffding} happens,
\begin{align*}
\coveragealt(\bar{k}) 
&
\ge \frac{1}{m} \sum_{j=1}^m \PP \Big( \CIalt_j \not\subseteq \simCIalt_j(\bar{k}) \Big) -  \sqrt{\frac{\log(1/\delta)}{2m}} \\
&\ge 
\gamma \cdot \PP\Big( \statistic (\testfunction)  \not\in \simCIalt(\bar{k}) \Big) -  \sqrt{\frac{\log(1/\delta)}{2m}} \tag{by \eqref{eqn-proxy-to-oracle-precise}} \\
&> 
\gamma\alpha. \tag{by definition of $\bar{k}$}
\end{align*}

\section{Characterizations of the Selected Simulation Sample Size}\label{sec-effective-sample-size-full}

In this appendix, we characterize the selected simulation sample sizes $\widehat{k} = \widehat{k}(C)$ and $\widehat{k}_{\KL} = \widehat{k}_{\KL}(C)$ selected for the CLT-based and concentration-based confidence intervals. The results show that they are approximately inversely proportional to the human-LLM discrepancy measures, which is similar to the behavior of their oracle counterparts $k^*(C)$ and $k^*_{\KL}(C)$ in \myCref{thm-oracle-effective-sample-size-chi} and \myCref{thm-oracle-effective-sample-size-KL}, respectively. These results establish $\widehat{k}$ and $\widehat{k}_{\KL}$ as meaningful empirical measures of the LLM's simulation fidelity.

\subsection{CLT-Based Confidence Interval}

\myCref{thm-effective-sample-size-chi} below characterizes the behavior of $\widehat{k}(C)$ defined by \eqref{eqn-selected-sample-size}.

\begin{theorem}[Simulation sample size]\label{thm-effective-sample-size-chi}
Let Assumption \myref{assumption-nondegeneracy-synthetic} hold. Suppose the $m$ calibration questions have the same number of human responses: $n_j=n$ for each $j\in[m]$. Take $\delta\in(0,1)$, and suppose $m\ge 32\alpha^{-1}\log(2K/\delta)$. There exist deterministic quantities $k_-(\alpha, C)$ and $k_+(\alpha, C)$ such that for every $C > 1$, the following holds with probability at least $1-\delta$:
\[
\min\left\{ k_-(\alpha, C), \, K \right\} \le \widehat{k}(C) \le  \min \left\{ k_+(\alpha, C), \, K \right\}.
\]
Here, $k_-(\alpha, C)$ and $k_+(\alpha, C)$ satisfy
\begin{align*}
& \lim_{C \to\infty} \frac{k_-(\alpha, C)}{C} = \left( \frac{z_{\alpha/2}}{ \quantile_{1-\alpha/8}^{1/2}\left( \discchireal(\testfunction) \right) + c_0 n^{-1/2} } \right)^2, \\[4pt]
& \lim_{C\to\infty} \frac{k_+(\alpha, C)}{C}
=
\left( \frac{z_{\alpha/2} }{\big( \quantile_{1-\alpha}^{1/2}\left( \discchireal(\testfunction) \right) - c_0n^{-1/2} \big)_+ } \right)^2,
\end{align*}
where $c_0$ is a constant depending on $\alpha$ and $\eta$.
\end{theorem}

\begin{proof}[Proof of \myCref{thm-effective-sample-size-chi}]
See \myCref{sec-thm-effective-sample-size-chi-proof}.
\end{proof}

\myCref{thm-effective-sample-size-chi} shows that $\widehat{k}(C)$ exhibits the same fundamental behavior as the oracle sample size $k^*(C)$. With high probability, $\widehat{k}(C)/C$ is approximately reciprocal to the discrepancy $\quantile_{1-\alpha}\left(\discchireal(\testfunction)\right)$, up to an additional $O(1/n)$ error term due to the finite sample approximation error of $\responsebar_j\approx\mean_j$. Combined with \myCref{thm-kappa-IT-chi}, we see that under the Mechanical Turk model, $\widehat{k}(C)/C \gtrsim \kappa$ approximately, serving as an empirical estimate of the hidden population size $\kappa$.

\subsection{Concentration-Based Confidence Interval}

We now turn to sample size $\widehat{k}_{\KL}(C)$, which has a similar definition of $\widehat{k}(C)$, with $\simCI_j(i;C)$ replaced by $\simCIKL_j(i;C)$ in the definition \eqref{eqn-selected-sample-size}. \myCref{thm-effective-sample-size-KL} below characterizes the asymptotic behavior of $\widehat{k}_{\KL}(C)$.

\begin{theorem}[Simulation sample size]\label{thm-effective-sample-size-KL}
Let Assumption \myref{assumption-nondegeneracy-both} hold. Suppose the $m$ calibration questions have the same number of human responses: $n_j=n$ for each $j\in[m]$. Take $\delta\in(0,1)$, and suppose $m\ge 32\alpha^{-1}\log(2K/\delta)$. There exist a universal constant $c_1>0$ and deterministic quantities $k_-(\alpha, C)$ and $k_+(\alpha, C)$ such that for every $C>1$, when $n \ge c_1\log(1+\alpha^{-1})\cdot\eta^{-1}(1-\eta)^{-1}$, the following holds with probability at least $1-\delta$:
\[
\min\left\{ k_-(\alpha, C), \, K \right\} \le \widehat{k}_{\KL}(C) \le  \min \left\{ k_+(\alpha, C), \, K \right\}.
\]
Here, $k_-(\alpha, C)$ and $k_+(\alpha, C)$ satisfy
\begin{align*}
& \lim_{C \to\infty} \frac{k_-(\alpha, C)}{C} = \frac{\log(2/\alpha)}{ \Big( \quantile_{1-\alpha/8}^{1/2}\left( \discKLreal(\testfunction) \right) + c_2 n^{-1/2} \Big)^2 }, \\[4pt]
& \lim_{C\to\infty} \frac{k_+(\alpha, C)}{C}
=
\frac{ \log(2/\alpha) }{\Big( \quantile_{1-\alpha}^{1/2}\left( \discKLreal(\testfunction) \right) - c_2n^{-1/2} \Big)_+^2 },
\end{align*}
where $c_2>0$ is a constant depending on $\alpha$ and $\eta$.
\end{theorem}

\begin{proof}[Proof of \myCref{thm-effective-sample-size-KL}]
See \myCref{sec-thm-effective-sample-size-KL-proof}.
\end{proof}

Similar to the CLT-based interval case, \myCref{thm-effective-sample-size-KL} confirms that $\widehat{k}_{\KL}(C)/C$ also has the same fundamental behavior as its oracle counterpart $k^*_{\KL}(C)/C$, exhibiting an inverse relation with the KL divergence discrepancy measure $\quantile_{1-\alpha}( \discKLreal(\testfunction) )$, up to a finite-sample approximation error of $O(1/n)$.

\section{Results Related to the Dilation Factor}\label{sec-dilation}

\subsection{Necessity of the Dilation Factor}\label{sec-impossibility-exact-CLT}

This appendix shows that when there is distribution shift between the real and synthetic distributions, the unscaled CLT-based confidence interval constructed from the synthetic data may fail to cover the true mean response with the prescribed coverage probability, regardless of the sample size.

Consider a true distribution $N(\mean,1)$ and a synthetic distribution $N(\simmean,1)$. Suppose we draw $k$ i.i.d.~synthetic samples $\{ x_i \}_{i=1}^k \sim N(\simmean,1)$ and construct the standard $(1-\alpha)$ confidence interval for $\mean$, given by 
\[
\simCIJ(k) = \left[ \bar{x}_k - \frac{z_{\alpha/2}}{\sqrt{k}}, ~ \bar{x}_k + \frac{z_{\alpha/2}}{\sqrt{k}} \right],
\] 
where $\bar{x}_k = \frac{1}{k} \sum_{i=1}^k x_i$. Let $\Delta = \mean - \simmean$, which represents the discrepancy between the true distribution and the synthetic distribution. Then
\begin{align*}
\PP\Big( \mean \in \simCIJ(k) \Big)
=
\PP\left( | \bar{x}_k - \mean | \le \frac{z_{\alpha/2}}{\sqrt{k}} \right) 
&=
\PP\left( | \sqrt{k} (\bar{x}_k - \simmean) - \sqrt{k} (\mean - \simmean) | \le z_{\alpha/2} \right) \\[4pt]
&=
\PP\left( \sqrt{k} \Delta - z_{\alpha/2} \le \sqrt{k} (\bar{x}_k - \simmean) \le \sqrt{k} \Delta + z_{\alpha/2} \right).
\end{align*}
Note that $\sqrt{k} (\bar{x}_k - \simmean) \sim N(0,1)$. Thus, whenever $\Delta\neq 0$, it holds that $\PP\left( \mean \in \simCIJ(k) \right) < 1-\alpha$, failing to attain $(1-\alpha)$ coverage probability, regardless of the sample size $k$.

\subsection{A Bootstrap Interpretation of the Dilation Factor}\label{sec-dilation-bootstrap}

In \myCref{sec-bootstrap}, we made a conceptual comparison between our framework and the classical parametric bootstrap. In this appendix, we establish a more quantitative connection. We will demonstrate that the structure of our dilated confidence interval arises naturally from the Mechanical Turk model when viewed through the lens of bootstrap resampling \citep{Efr79}. This derivation mathematically connects our dilation factor $C$ to the number of bootstrap resamples, justifying the relation $\widehat{k}\approx C\kappa$ between the total number of samples drawn $\widehat{k}$ and the effective sample size $\kappa$.
 
Under the Mechanical Turk model (Assumption \myref{assumption-MTurk}), the LLM simulates a hidden population of $\kappa$ human agents. If we simulate a total of $k$ synthetic samples where $k>\kappa$, we are effectively resampling from this hidden population. For simplicity, assume $k = B\kappa$, where we draw $B$ i.i.d.~responses $\{\simresponsetilde_{i,j}\}_{j=1}^B$ from each of the $\kappa$ agents $i=1,...,\kappa$. In this setup, $B$ is directly analogous to the number of resamples in a standard bootstrap procedure.

Fix a deterministic test question $\testfunction$. Our goal is to use the $k=B\kappa$ samples $\{\simresponsetilde_{i,j}\}_{i\in[\kappa],\,j\in[B]}$ to construct a confidence interval for the mean response $\mean = \EE_{\profile\sim\distribution} \performancefunction(\profile,\testfunction)$. Let $\simresponsebar_{k}$ denote the sample mean. To construct a confidence interval around $\simresponsebar_{k}$, we derive its asymptotic distribution in the following lemma.

\begin{lemma}[Central limit theorem]\label{lem-MTurk-CLT}
Let $\{\profile_i\}_{i=1}^{\infty}\sim\distribution$ be i.i.d. For each $i\in\ZZ_+$, given $\profile_i$, let $\datasetB_i = \{\simresponsetilde_{i,j}\}_{j=1}^B \sim \Bernoulli ( \performancefunction(\profile_i, \testfunction) )$ be i.i.d. Assume $\{ \datasetB_i \}_{i=1}^{\infty}$ are independent. Define $\tau^2 = \var_{\profile\sim\distribution} \performancefunction(\profile,\testfunction)$ and $\sigma^2 = \EE_{\profile \sim \distribution} \big[ \performancefunction (\profile, \testfunction) \left( 1 - \performancefunction (\profile, \testfunction) \right) \big]$. Let $k= \kappa B$, 
\[
\simresponsebar_{k} = \frac{1}{\kappa B} \sum_{i=1}^{\kappa} \sum_{j=1}^B \simresponsetilde_{i,j}
\quad\text{and}\quad
\simsamplesd_k = \sqrt{\simresponsebar_{k}\left( 1 - \simresponsebar_{k} \right)}.
\]
Then as $\kappa\to\infty$,
\[
\frac{\simresponsebar_{k} - \mean}{\simsamplesd_k/\sqrt{k}} ~ \xrightarrow{~d~} ~ N \left( 0, \, \frac{B\tau^2+\sigma^2}{\tau^2+\sigma^2} \right).
\]
\end{lemma}

\begin{proof}[Proof of \myCref{lem-MTurk-CLT}]
See \myCref{sec-lem-MTurk-CLT-proof}.
\end{proof}

\myCref{lem-MTurk-CLT} leads us to construct the following confidence interval:
\begin{equation}\label{eqn-MTurk-CI-dialted}
\left[
\simresponsebar_{k} - z_{\alpha/2}\simsamplesd_k\sqrt{\frac{C}{k}},
~
\simresponsebar_{k} + z_{\alpha/2}\simsamplesd_k\sqrt{\frac{C}{k}}
\right],
\end{equation}
where $C = \frac{B\tau^2+\sigma^2}{\tau^2+\sigma^2}\in[1,B]$ dilates the width of the confidence interval. This is exactly the form of the dilated confidence interval. The dilation factor $C$ is revealed to arise naturally in the resampling process. When $B$ is large, $C\approx B$, and we can identify the dilation factor $C$ with the number $B$ of resamples. 

This derivation provides a quantitative support for the comparison to bootstrap in \myCref{sec-bootstrap}. Without any knowledge of $\kappa$ and $B$, our method reparametrizes $B$ in terms of $C$, and selects a sample size $\widehat{k}$ for which \eqref{eqn-MTurk-CI-dialted} is a valid confidence interval. Thus, we expect $\widehat{k}\approx B\kappa$, which is $C\kappa$ when $C$ is large. In other words, $\widehat{k}/C\approx\kappa$ for $C$ large. This is consistent with our results in \myCref{sec-effective-sample-size}.

\section{Proofs for Effective Human Sample Size Analysis}

\subsection{Proof of \myCref{thm-kappa-IT-chi}}\label{sec-thm-kappa-IT-chi-proof}
Let $q(\profiles_{\kappa}) = \quantile_{1-\alpha}\big( \discchi(\testfunction) \mid \profiles_{\kappa} \big)$. Take $m= \big\lceil \frac{\log(\delta/2)}{\log(1-\alpha)} \big\rceil$ i.i.d.~samples $\testfunction_1,...,\testfunction_m\sim\testfunctiondist$, independent of $\profiles_{\kappa}$. Then $m \le -\frac{\log(2/\delta)}{\log(1-\alpha)} + 1 \le \frac{\log(2/\delta)}{\alpha} + 1$. By the definition of quantile, for all $\varepsilon > 0$,
\[
\PP\left(  \discchi(\testfunction_j) \le  q(\profiles_{\kappa}) - \varepsilon \mid \profiles_{\kappa} \right) < 1 - \alpha, \qquad \forall j\in[m].
\]
By independence,
\[
\PP\left(  \max_{j\in[m]} \discchi(\testfunction_j) \le q(\profiles_{\kappa}) - \varepsilon \Bigm| \profiles_{\kappa} \right)
=
\prod_{j=1}^m \PP\left(  \discchi(\testfunction_j) \le q(\profiles_{\kappa}) - \varepsilon \mid \profiles_{\kappa} \right) < (1-\alpha)^m \le \frac{\delta}{2},
\]
which implies
\begin{equation}\label{eqn-kappa-IT-chi-proof-1}
\PP\left(  q(\profiles_{\kappa}) < \max_{j\in[m]} \discchi(\testfunction_j) + \varepsilon \right) > 1 - \frac{\delta}{2}.
\end{equation}

We now derive a bound for $\max_{j\in[m]}\discchi(\testfunction_j)$. By Hoeffding's inequality (e.g., Theorem 2.8 in \cite{BLM13}), for every $j\in[m]$,
\[
\PP\left( \big| \mean(\testfunction_j) - \widehat{\mean}_{\kappa}(\testfunction_j) \big| \le \sqrt{\frac{\log(4m/\delta)}{2\kappa}} \Bigm| \testfunction_j \right) \ge 1 - \frac{\delta}{2m}.
\]
When this happens, if $\kappa \ge 2\left( \frac{\log(4m/\delta)}{r(1-r)} \right)^2$, then
\[
\widehat{\mean}_{\kappa}(\testfunction_j) \big(1 - \widehat{\mean}_{\kappa}(\testfunction_j)\big)
\ge 
\mean(\testfunction_j)\big(1-\mean(\testfunction)\big) - \big| \mean(\testfunction_j) - \widehat{\mean}_{\kappa}(\testfunction_j) \big|
\ge 
r(1-r) - \sqrt{\frac{\log(4m/\delta)}{2\kappa}}
\ge \frac{r(1-r)}{2},
\]
so
\[
\discchi (\testfunction_j) = \frac{\big| \mean(\testfunction_j) - \widehat{\mean}_{\kappa}(\testfunction_j) \big|^2}{\widehat{\mean}_{\kappa}\big(\testfunction_j) (1 - \widehat{\mean}_{\kappa}(\testfunction_j) )} \le \frac{\log(4m/\delta)}{r(1-r)} \cdot \frac{1}{\kappa} \le  \frac{c}{\kappa},
\]
where
\[
c  = \frac{\log\big[ 4\delta^{-1}\big(1 + \alpha^{-1} \log(2/\delta)\big)  \big]}{r(1-r)}.
\]
By a union bound over $j\in[m]$,
\begin{equation}\label{eqn-kappa-IT-chi-proof-2}
\PP\left( \max_{j\in[m]} \discchi (\testfunction_j) \le \frac{c}{\kappa} \right) \ge 1 - \frac{\delta}{2}.
\end{equation}

Combining \eqref{eqn-kappa-IT-chi-proof-1} and \eqref{eqn-kappa-IT-chi-proof-2} with a union bound, we obtain that
\begin{equation}\label{eqn-kappa-IT-chi-proof-3}
\PP\left(
q(\profiles_{\kappa}) <
\varepsilon + \frac{c}{\kappa} \right) \ge 1 - \delta.
\end{equation}
Since this holds for all $\varepsilon> 0$, then by the continuity of the probability measure $\PP$,
\[
\PP\left(
q(\profiles_{\kappa}) \le \frac{c}{\kappa} \right) = \lim_{\varepsilon\to 0^+}\PP\left( q(\profiles_{\kappa}) <
\varepsilon + \frac{c}{\kappa} \right) \ge 1 - \delta.
\]
This finishes the proof.

\subsection{Proof of \myCref{thm-oracle-effective-sample-size-chi}}\label{sec-thm-oracle-effective-sample-size-chi-proof}

For notational convenience, we define $Q_{\beta}=\quantile_{\beta}\left(\discchireal(\testfunction)\right)$. \myCref{thm-oracle-effective-sample-size-chi} is a consequence of the following lemma.

\begin{lemma}\label{lem-oracle-effective-sample-size-chi}
In the setting of \myCref{thm-oracle-effective-sample-size-chi},
\[
\frac{z_{\alpha/2}^2}{Q_{1-\alpha/2}}
\le 
\liminf_{C \to\infty} \frac{k^*(C)}{C}
\le
\limsup_{C\to\infty}\frac{k^*(C)}{C}
\le 
\frac{z_{\alpha/2}^2}{Q_{1-2\alpha}}.
\]
\end{lemma}

\begin{proof}[Proof of \myCref{lem-oracle-effective-sample-size-chi}]
See \myCref{sec-lem-oracle-effective-sample-size-chi-proof}.
\end{proof}

Given \myCref{lem-oracle-effective-sample-size-chi}, we can construct $\varepsilon(\alpha, C)$ as follows. Define
\[
\varepsilon_-(\alpha, C) = \left( \frac{z_{\alpha/2}^2}{Q_{1-\alpha/2}} - \inf_{C'\ge C}\frac{k^*(C')}{C'}  \right)_+
\qquad\text{and}\qquad
\varepsilon_+(\alpha, C) = \left( \sup_{C'\ge C}\frac{k^*(C')}{C'} - \frac{z_{\alpha/2}^2}{Q_{1-2\alpha}} \right)_+.
\]
Then for all $C$,
\[
\frac{z_{\alpha/2}^2}{Q_{1-\alpha/2}} - \varepsilon_-(\alpha, C) \le \inf_{C'\ge C}\frac{k^*(C')}{C'} 
\le 
\frac{k^*(C)}{C}
\le
\sup_{C'\ge C}\frac{k^*(C')}{C'} 
\le 
\frac{z_{\alpha/2}^2}{Q_{1-2\alpha}} + \varepsilon_+(\alpha, C).
\]
Let
\[
\varepsilon(\alpha, C) = \max\left\{ \varepsilon_-(\alpha, C), ~ \varepsilon_+(\alpha, C) \right\},
\]
then for all $C$,
\[
\frac{z_{\alpha/2}^2}{Q_{1-\alpha/2}} - \varepsilon(\alpha, C) 
\le
\frac{k^*(C)}{C}
\le 
\frac{z_{\alpha/2}^2}{Q_{1-2\alpha}} + \varepsilon(\alpha, C).
\]
By \myCref{lem-oracle-effective-sample-size-chi},
\begin{align*}
&\lim_{C\to\infty} \varepsilon_-(\alpha, C) = \left( \frac{z_{\alpha/2}^2}{Q_{1-\alpha/2}} - \liminf_{C\to\infty}\frac{k^*(C)}{C}  \right)_+ = 0, \\[4pt]
&\lim_{C\to\infty}\varepsilon_+(\alpha, C) = \left( \limsup_{C\to\infty}\frac{k^*(C)}{C} - \frac{z_{\alpha/2}^2}{Q_{1-2\alpha}} \right)_+ = 0,
\end{align*}
which implies $\lim_{C\to\infty} \varepsilon(\alpha, C) = 0$. This proves \myCref{thm-oracle-effective-sample-size-chi}.

\subsection{Proof of \myCref{lem-oracle-effective-sample-size-chi}}\label{sec-lem-oracle-effective-sample-size-chi-proof}

Define the standard deviation $\simsd(\testfunction) = \sqrt{\simmean(\testfunction) \cdot \left( 1 - \simmean(\testfunction) \right) }$ of the synthetic response distribution. By Assumption \myref{assumption-nondegeneracy-synthetic}, the following event has probability at least $1-\alpha/8$:
\[
\cA_1= \left\{ \simmean(\testfunction) \in [\eta, 1-\eta] \right\}.
\]
Thanks to Assumption \myref{assumption-nondegeneracy-synthetic}, $Q_{\beta}$ is finite for all $\beta\in[0,1-\alpha/8]$. By definition, the following event has probability at least $1-\alpha/2$:
\[
\cA_2 = \left\{ \discchireal(\testfunction) \le Q_{1-\alpha/2} \right\}.
\]
By Bernstein's inequality (e.g., Inequality (2.10) in \cite{BLM13}), the following event has probability at least $1-\alpha/4$:
\[
\cA_3(k) = \left\{ \big| \simresponsebar_{k} - \simmean(\testfunction) \big|
\le 
\simsd(\testfunction)\sqrt{\frac{2\log(8/\alpha)}{k}} + \frac{2\log(8/\alpha)}{3k} \right\}.
\]
By Theorem 10 in \cite{MPo09}, the following event has probability at least $1-\alpha/8$:
\[
\cA_4(k) = \left\{ \bigg| \simsamplesd_{k} - \simsd(\testfunction) \sqrt{\frac{k-1}{k}} \bigg| \le \sqrt{\frac{2\log(16/\alpha)}{k}} \right\}.
\]
By a union bound, the event $\cA(k) = \cA_1 \cap \cA_2 \cap \cA_3(k) \cap \cA_4(k)$ has probability at least $1-\alpha$. 

\paragraph{Lower bound.} We first derive a lower bound for $\liminf_{C\to\infty} k^*(C)/C$. For a fixed $k\in\ZZ_+$ and $C$, we have
\begin{align}
\mean(\testfunction) \in  \simCI(k;C)
&\quad \Leftrightarrow \quad
\big| \mean(\testfunction) - \simresponsebar_{k} \big| \le z_{\alpha/2} \cdot  \simsamplesd_{k} \sqrt{\frac{C}{k}} \notag \\[4pt]
& \quad \Leftarrow \quad
| \mean(\testfunction) - \simmean(\testfunction) | 
+ 
\big| \simresponsebar_{k} - \simmean(\testfunction) \big| \notag \\[4pt]
&\qquad\qquad \le 
z_{\alpha/2} \left( \simsd(\testfunction) - \bigg| \simsamplesd_{k} - \simsd(\testfunction) \sqrt{\frac{k-1}{k}} \bigg| \right) \sqrt{\frac{C}{k}} , \label{eqn-lem-oracle-effective-sample-size-chi-lower-1}
\end{align}
When $\cA(k)$ happens, \eqref{eqn-lem-oracle-effective-sample-size-chi-lower-1} is implied by
\begin{align}
	& | \mean(\testfunction) - \simmean(\testfunction) |  + \left( \simsd(\testfunction) \sqrt{\frac{2\log(8/\alpha)}{k}} + \frac{2\log(8/\alpha)}{3k} \right) \notag \\
	&\qquad \le 
	z_{\alpha/2} \left( \simsd(\testfunction) - \sqrt{\frac{2\log(16/\alpha)}{k}} \right) \sqrt{\frac{C}{k}} \notag \\[6pt]
	&\Leftarrow \quad \sqrt{\discchireal(\testfunction)}  + \left( \sqrt{\frac{2\log(8/\alpha)}{k}} + \frac{1}{\simsd(\testfunction)} \frac{2\log(8/\alpha)}{3k} \right) \notag \\
	&\qquad \le 
	z_{\alpha/2} \left( 1 - \frac{1}{\simsd(\testfunction)} \sqrt{\frac{2\log(16/\alpha)}{k}} \right) \sqrt{\frac{C}{k}} \tag{divide by $\simsd(\testfunction)$} \\[6pt]
	&\Leftarrow \quad \sqrt{Q_{1-\alpha/2}}  + \left( \sqrt{\frac{2\log(8/\alpha)}{k}} + \sqrt{\frac{1}{\eta(1-\eta)}} \frac{2\log(8/\alpha)}{3k} \right) \notag \\
	&\qquad \le 
	z_{\alpha/2} \left( 1 - \sqrt{\frac{1}{\eta(1-\eta)}} \sqrt{\frac{2\log(16/\alpha)}{k}} \right) \sqrt{\frac{C}{k}} \notag \\[6pt] 
	&\Leftarrow \quad \sqrt{Q_{1-\alpha/2}} + \frac{c_1}{\sqrt{k}} \le z_{\alpha/2}\left( 1 - \frac{c_2}{\sqrt{k}} \right) \sqrt{\frac{C}{k}} \label{eqn-lem-oracle-effective-sample-size-chi-lower-2}
\end{align}
for some constants $c_1,c_2>0$ that depend on $\alpha$ and $\eta$. When $C$ is sufficiently large, solving \eqref{eqn-lem-oracle-effective-sample-size-chi-lower-2} for $k$ yields
\[
k \le \frac{ C \left(z_{\alpha/2} - C^{-1/2} c_1 + \sqrt{ \left( z_{\alpha/2} - C^{-1/2}c_1 \right)^2 - 4 C^{-1/2}c_2 z_{\alpha/2} Q_{1-\alpha/2}^{1/2}}\right)^2}{4Q_{1-\alpha/2}}.
\]
Let
\[
k_-^*(C) = \Bigg\lfloor  \frac{ C \left(z_{\alpha/2} - C^{-1/2} c_1 + \sqrt{ \left( z_{\alpha/2} - C^{-1/2} c_1 \right)^2 - 4 C^{-1/2}c_2 z_{\alpha/2} Q_{1-\alpha/2}^{1/2}}\right)^2}{4Q_{1-\alpha/2}} \Bigg\rfloor.
\]
By \eqref{eqn-lem-oracle-effective-sample-size-chi-lower-1} and \eqref{eqn-lem-oracle-effective-sample-size-chi-lower-2},
\[
\PP\left( \mean(\testfunction) \in  \simCI(k^*_-(C);C) \right) \ge \PP\left( \cA(k_-^*(C)) \right) \ge 1-\alpha.
\]
Then $k^*(C) \ge k_-^*(C)$, and thus
\begin{equation}\label{eqn-lem-oracle-effective-sample-size-chi-lower-final}
\liminf_{C \to\infty} \frac{k^*(C)}{C} \ge \lim_{C\to\infty} \frac{k_-^*(C)}{C} = \frac{z_{\alpha/2}^2}{Q_{1-\alpha/2}}.
\end{equation}
This proves the lower bound.

\paragraph{Upper bound.} We next derive an upper bound for $\limsup_{C\to\infty} k^*(C)/C$. Take an arbitrary $\epsilon\in(0,Q_{1-2\alpha})$. By definition, the following event has probability strictly greater than $2\alpha$:
\[
\cA_5 = \left\{ \discchireal(\testfunction) \ge Q_{1-2\alpha} - \epsilon \right\}.
\]
By a union bound, the event $\cA'(k) = \cA(k)\cap\cA_5$ has probability strictly greater than $\alpha$.

For a fixed $k\in\ZZ_+$, we have
\begin{align}
\mean(\testfunction) \not\in  \simCI(k;C)
&\quad \Leftrightarrow \quad
\big| \mean(\testfunction) - \simresponsebar_{k} \big| > z_{\alpha/2} \cdot  \simsamplesd_{k} \sqrt{\frac{C}{k}} \notag \\[4pt]
& \quad \Leftarrow \quad
| \mean(\testfunction) - \simmean(\testfunction) | 
-
\big| \simresponsebar_{k} - \simmean(\testfunction) \big| \notag \\[4pt]
&\qquad\qquad >
z_{\alpha/2} \left( \simsd(\testfunction) + \bigg| \simsamplesd_{k} - \simsd(\testfunction) \sqrt{\frac{k-1}{k}} \bigg| \right) \sqrt{\frac{C}{k}} , \label{eqn-lem-oracle-effective-sample-size-chi-upper-1}
\end{align}
When $\cA'(k)$ happens, \eqref{eqn-lem-oracle-effective-sample-size-chi-upper-1} is implied by
\begin{align}
	& | \mean(\testfunction) - \simmean(\testfunction) |  - \left( \simsd(\testfunction) \sqrt{\frac{2\log(8/\alpha)}{k}} + \frac{2\log(8/\alpha)}{3k} \right) \notag \\
	&\qquad >
	z_{\alpha/2} \left( \simsd(\testfunction) + \sqrt{\frac{2\log(16/\alpha)}{k}} \right) \sqrt{\frac{C}{k}} \notag \\[6pt]
	&\Leftarrow \quad \sqrt{\discchireal(\testfunction)}  - \left( \sqrt{\frac{2\log(8/\alpha)}{k}} + \frac{1}{\simsd(\testfunction)} \frac{2\log(8/\alpha)}{3k} \right) \notag \\
	&\qquad > 
	z_{\alpha/2} \left( 1 + \frac{1}{\simsd(\testfunction)} \sqrt{\frac{2\log(16/\alpha)}{k}} \right) \sqrt{\frac{C}{k}} \tag{divide by $\simsd(\testfunction)$} \\[6pt]
	&\Leftarrow \quad \sqrt{Q_{1-2\alpha}-\epsilon}  - \left( \sqrt{\frac{2\log(8/\alpha)}{k}} + \sqrt{\frac{1}{\eta(1-\eta)}} \frac{2\log(8/\alpha)}{3k} \right) \notag \\
	&\qquad >
	z_{\alpha/2} \left( 1 + \sqrt{\frac{1}{\eta(1-\eta)}} \sqrt{\frac{2\log(16/\alpha)}{k}} \right) \sqrt{\frac{C}{k}} \notag \\[6pt] 
	&\Leftarrow \quad \sqrt{Q_{1-2\alpha}-\epsilon} - \frac{c_3}{\sqrt{k}} > z_{\alpha/2}\left( 1 + \frac{c_4}{\sqrt{k}} \right) \sqrt{\frac{C}{k}} \label{eqn-lem-oracle-effective-sample-size-chi-upper-2}
\end{align}
for some constants $c_3,c_4>0$ that depend on $\alpha$ and $\eta$. Solving \eqref{eqn-lem-oracle-effective-sample-size-chi-upper-2} for $k$ gives
\[
k >  \frac{C\left( 2 + C^{-1/2} c_3 + \sqrt{\left( z_{\alpha/2} + C^{-1/2}c_3 \right)^2 + 4C^{-1/2}c_4z_{\alpha/2} ( Q_{1-2\alpha} - \epsilon )^{1/2} } \right)^2}{ 4 (Q_{1-2\alpha} - \epsilon )} .
\]
Let
\[
k_+^*(C,\epsilon) = \Bigg\lfloor \frac{C\left( z_{\alpha/2} + C^{-1/2} c_3 + \sqrt{\left( z_{\alpha/2} + C^{-1/2}c_3 \right)^2 + 4C^{-1/2}c_4z_{\alpha/2} ( Q_{1-2\alpha} - \epsilon )^{1/2} } \right)^2}{ 4 (Q_{1-2\alpha} - \epsilon )} \Bigg\rfloor + 1.
\]
By \eqref{eqn-lem-oracle-effective-sample-size-chi-upper-1} and \eqref{eqn-lem-oracle-effective-sample-size-chi-upper-2},
\[
\PP\left( \mean(\testfunction) \not\in  \simCI(k_+^*(C,\epsilon);C) \right) \ge \PP \left(\cA'(k_+^*(C,\epsilon)) \right) > \alpha
.\]
Then $k^*(C) < k_+^*(C,\epsilon)$, and thus
\[
\limsup_{C\to\infty}\frac{k^*(C)}{C}
\le 
\lim_{C\to\infty} \frac{k_+^*(C,\epsilon)}{C}
=
\frac{z_{\alpha/2}^2}{Q_{1-2\alpha} - \epsilon}.
\]
Since this holds for all $\epsilon\in(0,Q_{1-2\alpha})$, then
\begin{equation}\label{eqn-lem-oracle-effective-sample-size-chi-upper-final}
\limsup_{C\to\infty}\frac{k^*(C)}{C}
\le 
\frac{z_{\alpha/2}^2}{Q_{1-2\alpha}}.
\end{equation}
This proves the upper bound.

Combining \eqref{eqn-lem-oracle-effective-sample-size-chi-lower-final} and \eqref{eqn-lem-oracle-effective-sample-size-chi-upper-final} finishes the proof.

\subsection{Proof of \myCref{thm-effective-sample-size-chi}}\label{sec-thm-effective-sample-size-chi-proof}

Take i.i.d.~samples $\{y_i\}_{i=1}^n\sim\Bernoulli(\mean(\testfunction))$, and set $\responsebar = \frac{1}{n} \sum_{i=1}^n y_i$. By Hoeffding's inequality and a union bound, with probability at least $1-\delta$,
\begin{equation}\label{eqn-thm-effective-sample-size-chi-1}
\left|\frac{1}{m} \sum_{j=1}^m \ind \left\{ \responsebar_j \in \simCI_j(k;C) \right\}
-
\PP\big( \responsebar \in \simCI(k;C) \big) \right|
\le 
\sqrt{\frac{\log(2K/\delta)}{2m}}, \quad \forall k\in[K].
\end{equation}
From now on suppose that \eqref{eqn-thm-effective-sample-size-chi-1} happens. When $m\ge 32\alpha^{-1}\log(2K/\delta)$, we have $\sqrt{\frac{\log(K/\delta)}{2m}}\le \alpha/8$. In this case,
\begin{align}
&\widehat{k}(C)
\ge 
\max\left\{ 0\le k\le K:
\PP\big( \responsebar \in \simCI(i;C) \big) \ge 1 - 3\alpha/8,\ \forall i \le k
\right\}, \label{eqn-thm-effective-sample-size-chi-upper-0} \\[4pt]
&\widehat{k}(C)
\le 
\max\left\{ 0\le k\le K:
\PP\big( \responsebar \in \simCI(i;C) \big) \ge 1 - 5\alpha/8,\ \forall i \le k
\right\}. \label{eqn-thm-effective-sample-size-chi-lower-0}
\end{align}
By definition,
\[
\PP\big( \responsebar \in \simCI(k;C) \big)
=
\PP\left( \big| \responsebar - \simresponsebar_{k} \big| \le z_{\alpha/2} \cdot  \simsamplesd_{k} \sqrt{\frac{C}{k}} \right).
\]

Define the standard deviation $\simsd(\testfunction) = \sqrt{\simmean(\testfunction) \cdot \left( 1 - \simmean(\testfunction) \right) }$ of the synthetic response distribution. By Assumption \myref{assumption-nondegeneracy-synthetic}, the following event has probability at least $1-\alpha/8$:
\[
\cB_1= \left\{ \simmean(\testfunction) \in [\eta, 1-\eta] \right\}.
\]
For notational convenience, we define $Q_{\beta}=\quantile_{\beta}\left(\discchireal(\testfunction)\right)$. Thanks to Assumption \myref{assumption-nondegeneracy-synthetic}, $Q_{\beta}$ is finite for all $\beta\in[0,1-\alpha/8]$. By definition, the following event has probability at least $1-\alpha/8$:
\[
\cB_2 = \left\{ \discchireal(\testfunction) \le Q_{1-\alpha/8} \right\}.
\]
By Bernstein's inequality (e.g., Inequality (2.10) in \cite{BLM13}), the following event has probability at least $1-\alpha/24$:
\[
\cB_3(k) = \left\{ \big| \simresponsebar_{k} - \simmean(\testfunction) \big|
\le 
\simsd(\testfunction)\sqrt{\frac{2\log(48/\alpha)}{k}} + \frac{2\log(48/\alpha)}{3k} \right\}.
\]
By Theorem 10 in \cite{MPo09}, the following event has probability at least $1-\alpha/24$:
\[
\cB_4(k) = \left\{ \bigg| \simsamplesd_{k} - \simsd(\testfunction) \sqrt{\frac{k-1}{k}} \bigg| \le \sqrt{\frac{2\log(48/\alpha)}{k}} \right\}.
\]
By Hoeffding's inequality (e.g., Theorem 2.8 in \cite{BLM13}), the following event has probability at least $1-\alpha/24$:
\[
\cB_5 = \left\{ \big| \responsebar - \mean(\testfunction) \big|
\le 
\sqrt{\frac{\log(48/\alpha)}{2n}} \right\}.
\]
By a union bound, the event $\cB(k) = \cB_1 \cap \cB_2 \cap \cB_3(k) \cap \cB_4(k) \cap \cB_5$ has probability at least $1-3\alpha/8$.

\paragraph{Lower bound.} We first derive a lower bound for $\widehat{k}(C)$. For a fixed $k\in\ZZ_+$ and $C$, we have
\begin{align}
\responsebar \in  \simCI(k;C)
&\quad \Leftrightarrow \quad
\big| \responsebar - \simresponsebar_{k} \big| \le z_{\alpha/2} \cdot  \simsamplesd_{k} \sqrt{\frac{C}{k}} \notag \\[4pt]
& \quad \Leftarrow \quad
| \mean(\testfunction) - \simmean(\testfunction) |
+
| \responsebar - \mean(\testfunction) | 
+ 
\big| \simresponsebar_{k} - \simmean(\testfunction) \big| \notag \\[4pt]
&\qquad\qquad \le 
z_{\alpha/2} \left( \simsd(\testfunction) - \bigg| \simsamplesd_{k} - \simsd(\testfunction) \sqrt{\frac{k-1}{k}} \bigg| \right) \sqrt{\frac{C}{k}} , \label{eqn-thm-effective-sample-size-chi-lower-1}
\end{align}
When $\cB(k)$ happens, \eqref{eqn-thm-effective-sample-size-chi-lower-1} is implied by
\begin{align}
	& | \mean(\testfunction) - \simmean(\testfunction) | + \sqrt{\frac{\log(48/\alpha)}{2n}}  + \left( \simsd(\testfunction) \sqrt{\frac{2\log(48/\alpha)}{k}} + \frac{2\log(48/\alpha)}{3k} \right) \notag \\
	&\qquad \le 
	z_{\alpha/2} \left( \simsd(\testfunction) - \sqrt{\frac{2\log(48/\alpha)}{k}} \right) \sqrt{\frac{C}{k}} \notag \\[6pt]
	&\Leftarrow \quad \sqrt{\discchireal(\testfunction)}  + \left( \frac{1}{\simsd(\testfunction)} \sqrt{\frac{\log(48/\alpha)}{2n}} + \sqrt{\frac{2\log(48/\alpha)}{k}} +  \frac{1}{\simsd(\testfunction)} \frac{2\log(48/\alpha)}{3k} \right) \notag \\
	&\qquad \le 
	z_{\alpha/2} \left( 1 - \frac{1}{\simsd(\testfunction)} \sqrt{\frac{2\log(48/\alpha)}{k}} \right) \sqrt{\frac{C}{k}} \tag{divide by $\simsd(\testfunction)$} \\[6pt]
	&\Leftarrow \quad \sqrt{Q_{1-\alpha/8}}  + \left( \sqrt{\frac{1}{\eta(1-\eta)}} \sqrt{\frac{\log(48/\alpha)}{2n}} + \sqrt{\frac{2\log(48/\alpha)}{k}} + \sqrt{\frac{1}{\eta(1-\eta)}} \frac{2\log(48/\alpha)}{3k} \right) \notag \\
	&\qquad \le 
	z_{\alpha/2} \left( 1 - \sqrt{\frac{1}{\eta(1-\eta)}} \sqrt{\frac{2\log(48/\alpha)}{k}} \right) \sqrt{\frac{C}{k}} \notag \\[6pt] 
	&\Leftarrow \quad \sqrt{Q_{1-\alpha/8}} + \frac{c_0}{\sqrt{n}} + \frac{c_1}{\sqrt{k}} \le z_{\alpha/2}\left( 1 - \frac{c_2}{\sqrt{k}} \right) \sqrt{\frac{C}{k}} \label{eqn-thm-effective-sample-size-chi-lower-2}
\end{align}
for some constants $c_0,c_1,c_2>0$ that depend on $\alpha$ and $\eta$. When $C$ is sufficiently large, solving \eqref{eqn-thm-effective-sample-size-chi-lower-2} for $k$ yields
\[
k \le \frac{ C \left(z_{\alpha/2} - C^{-1/2} c_1 + \sqrt{ \left( z_{\alpha/2} - C^{-1/2}c_1 \right)^2 - 4 C^{-1/2}c_2 z_{\alpha/2} \big( Q_{1-\alpha/8}^{1/2} + c_0n^{-1/2} \big) }\right)^2}{4\big( Q_{1-\alpha/8}^{1/2} + c_0 n^{-1/2} \big)^2 }.
\]
Let
\[
k_-(C) = \Bigg\lfloor  \frac{ C \left(z_{\alpha/2} - C^{-1/2} c_1 + \sqrt{ \left( z_{\alpha/2} - C^{-1/2} c_1 \right)^2 - 4 C^{-1/2}c_2 z_{\alpha/2} \big( Q_{1-\alpha/8}^{1/2} + c_0n^{-1/2} \big) }\right)^2}{ 4\big( Q_{1-\alpha/8}^{1/2} + c_0 n^{-1/2} \big)^2 } \Bigg\rfloor.
\]
Here we have dropped the dependence of $k_-(C)$ on $\alpha$ for notational convenience. By \eqref{eqn-thm-effective-sample-size-chi-lower-1} and \eqref{eqn-thm-effective-sample-size-chi-lower-2}, for all $i \le k_-(C)$,
\[
\PP\left( \responsebar \in \simCI(i;C) \right)  \ge \PP\left( \cB(i) \right) \ge 1 - 3\alpha/8.
\]
By \eqref{eqn-thm-effective-sample-size-chi-lower-0}, we have $\widehat{k}(C) \ge k_-(C)$, with
\[
\lim_{C \to\infty} \frac{k_-(C)}{C} = \left( \frac{z_{\alpha/2}}{ Q_{1-\alpha/8}^{1/2} + c_0 n^{-1/2} } \right)^2.
\]
This proves the lower bound.

\paragraph{Upper bound.} We next derive an upper bound for $\widehat{k}(C)$. Take an arbitrary $\epsilon\in(0,Q_{1-\alpha})$. By definition, the following event has probability strictly greater than $\alpha$:
\[
\cB_5 = \left\{ \discchireal(\testfunction) \ge Q_{1-\alpha} - \epsilon \right\}.
\]
By a union bound, the event $\cB'(k) = \cB(k)\cap\cB_5$ has probability strictly greater than $5\alpha/8$.

For a fixed $k\in\ZZ_+$, we have
\begin{align}
\responsebar \not\in  \simCI(k;C)
&\quad \Leftrightarrow \quad
\big| \responsebar - \simresponsebar_{k} \big| > z_{\alpha/2} \cdot  \simsamplesd_{k} \sqrt{\frac{C}{k}} \notag \\[4pt]
& \quad \Leftarrow \quad
| \mean(\testfunction) - \simmean(\testfunction) | 
- | \responsebar - \mean(\testfunction) | -
\big| \simresponsebar_{k} - \simmean(\testfunction) \big| \notag \\[4pt]
&\qquad\qquad >
z_{\alpha/2} \left( \simsd(\testfunction) + \bigg| \simsamplesd_{k} - \simsd(\testfunction) \sqrt{\frac{k-1}{k}} \bigg| \right) \sqrt{\frac{C}{k}} , \label{eqn-thm-effective-sample-size-chi-upper-1}
\end{align}
When $\cB'(k)$ happens, \eqref{eqn-thm-effective-sample-size-chi-upper-1} is implied by
\begin{align}
	& | \mean(\testfunction) - \simmean(\testfunction) | - \sqrt{\frac{\log(48/\alpha)}{2n}}  - \left( \simsd(\testfunction) \sqrt{\frac{2\log(48/\alpha)}{k}} + \frac{2\log(48/\alpha)}{3k} \right) \notag \\
	&\qquad >
	z_{\alpha/2} \left( \simsd(\testfunction) + \sqrt{\frac{2\log(48/\alpha)}{k}} \right) \sqrt{\frac{C}{k}} \notag \\[6pt]
	&\Leftarrow \quad \sqrt{\discchireal(\testfunction)}   - \frac{1}{\simsd(\testfunction)} \sqrt{\frac{\log(48/\alpha)}{2n}} - \left( \sqrt{\frac{2\log(48/\alpha)}{k}} + \frac{1}{\simsd(\testfunction)} \frac{2\log(48/\alpha)}{3k} \right) \notag \\
	&\qquad > 
	z_{\alpha/2} \left( 1 + \frac{1}{\simsd(\testfunction)} \sqrt{\frac{2\log(48/\alpha)}{k}} \right) \sqrt{\frac{C}{k}} \tag{divide by $\simsd(\testfunction)$} \\[6pt]
	&\Leftarrow \quad \sqrt{Q_{1-\alpha} -\epsilon} - \sqrt{\frac{1}{\eta(1-\eta)}} \sqrt{\frac{\log(48/\alpha)}{2n}}  - \left( \sqrt{\frac{2\log(48/\alpha)}{k}} + \sqrt{\frac{1}{\eta(1-\eta)}} \frac{2\log(48/\alpha)}{3k} \right) \notag \\
	&\qquad >
	z_{\alpha/2} \left( 1 + \sqrt{\frac{1}{\eta(1-\eta)}} \sqrt{\frac{2\log(16/\alpha)}{k}} \right) \sqrt{\frac{C}{k}} \notag \\[6pt] 
	&\Leftarrow \quad \sqrt{Q_{1-\alpha}-\epsilon} - \frac{c_0}{\sqrt{n}} - \frac{c_3}{\sqrt{k}} > z_{\alpha/2}\left( 1 + \frac{c_4}{\sqrt{k}} \right) \sqrt{\frac{C}{k}} \label{eqn-thm-effective-sample-size-chi-upper-2}
\end{align}
for some constants $c_0,c_3,c_4>0$ that depend on $\alpha$ and $\eta$. Solving \eqref{eqn-thm-effective-sample-size-chi-upper-2} for $k$ gives
\[
k >  \frac{C\left( z_{\alpha/2} + C^{-1/2} c_3 + \sqrt{\left( z_{\alpha/2} + C^{-1/2}c_3 \right)^2 + 4C^{-1/2}c_4z_{\alpha/2} \big[ ( Q_{1-\alpha} - \epsilon )^{1/2} - c_0 n^{-1/2} \big] } \right)^2}{ 4 \big( (Q_{1-\alpha} - \epsilon )^{1/2} - c_0n^{-1/2} \big)_+ } .
\]
Let
\[
\widetilde{k}_+(C,\epsilon) = \Bigg\lfloor \frac{C\left( z_{\alpha/2} + C^{-1/2} c_3 + \sqrt{\left( z_{\alpha/2} + C^{-1/2}c_3 \right)^2 + 4C^{-1/2}c_4z_{\alpha/2} ( Q_{1-\alpha} - \epsilon )^{1/2} } \right)^2}{ 4 \big( (Q_{1-\alpha} - \epsilon )^{1/2} - c_0n^{-1/2} \big)_+^2 } \Bigg\rfloor + 1.
\]
Here we have dropped the dependence of $\widetilde{k}_+(C,\epsilon)$ on $\alpha$ for notational convenience. When $\widetilde{k}_+(C,\epsilon)<\infty$, by \eqref{eqn-thm-effective-sample-size-chi-upper-1} and \eqref{eqn-thm-effective-sample-size-chi-upper-2},
\[
\PP\left( \responsebar \not\in \simCI(\widetilde{k}_+(C,\epsilon);C) \right)  \ge \PP\left( \cB'(\widetilde{k}_+(C,\epsilon)) \right) > 5\alpha/8,
\]
so $\PP\left( \responsebar \in \simCI(\widetilde{k}_+(C,\epsilon);C) \right)  < 1 -  5\alpha/8$. By \eqref{eqn-thm-effective-sample-size-chi-upper-0}, we have 
\[
\widehat{k}(C) < \widetilde{k}_+(C,\epsilon),
\] 
with
\begin{equation}\label{eqn-thm-effective-sample-size-chi-upper-diagonal-1}
\lim_{C\to\infty} \frac{\widetilde{k}_+(C,\epsilon)}{C}
=
\left[ \frac{z_{\alpha/2} }{\big( (Q_{1-\alpha} - \epsilon )^{1/2} - c_0n^{-1/2} \big)_+ } \right]^2.
\end{equation}
When $\widetilde{k}_+(C,\epsilon)=\infty$, we still have $\widehat{k}(C) < \widetilde{k}_+(C,\epsilon)$ and \eqref{eqn-thm-effective-sample-size-chi-upper-diagonal-1}. As $\epsilon\to 0^+$, 
\begin{equation}\label{eqn-thm-effective-sample-size-chi-upper-diagonal-2}
\left[ \frac{z_{\alpha/2} }{\big( (Q_{1-\alpha} - \epsilon )^{1/2} - c_0n^{-1/2} \big)_+ } \right]^2
\to
\left[ \frac{z_{\alpha/2} }{\big( Q_{1-\alpha}^{1/2} - c_0n^{-1/2} \big)_+ } \right]^2.
\end{equation}
Denote the left hand side of \eqref{eqn-thm-effective-sample-size-chi-upper-diagonal-2} by $U(\epsilon)$, then the right hand side is $U(0)$.

We will now use $\widetilde{k}_+(C,\epsilon)$ to construct $k_+(C)$ such that $\widehat{k}(C) \le k_+(C)$ and
\[
\lim_{C\to\infty} \frac{k_+(C)}{C} = U(0).
\]
Take a positive sequence $\{\epsilon_j\}_{j=1}^{\infty}$ such that $\epsilon_j\to 0$ as $j\to\infty$, and $\big| U(\epsilon_j) - U(0) \big| \le 1/j$. By \eqref{eqn-thm-effective-sample-size-chi-upper-diagonal-1}, there exists a strictly increasing sequence $\{C_j\}_{j=1}^{\infty}$ such that for each $j\in\ZZ_+$,
\[
\left|\frac{\widetilde{k}_+(C,\epsilon_j)}{C} - U(\epsilon_j) \right| \le \frac{1}{j} ,\qquad \forall C \ge C_j.
\]
Thus, by the triangle inequality,
\[
\left|\frac{\widetilde{k}_+(C,\epsilon_j)}{C} - U(0) \right| \le \frac{2}{j}, \qquad\forall C\ge C_j.
\]
Define $k_+(C)$ by
\[
k_+(C) = 
\begin{cases}
K, &\quad\text{if } C < C_1 \\
\widetilde{k}_+(C,\epsilon_j), &\quad\text{if } C_j \le C < C_{j+1}
\end{cases}.
\]
Then for all $C$, we have $\widehat{k} \le \min\{K, k_+(C)\}$. Moreover, for all $j\in\ZZ_+$, for all $C\ge C_j$, there exists $j'\ge j$ such that $C_{j'}\le C < C_{j'+1}$, so
\[
\left| \frac{k_+(C)}{C} - U(0) \right|
=
\left| \frac{\widetilde{k}_+(C,\epsilon_{j'})}{C} - U(0) \right|
\le 
\frac{2}{j'}
\le 
\frac{2}{j}.
\]
Taking $C\to\infty$ yields
\[
\limsup_{C\to\infty} \left| \frac{k_+(C)}{C} - U(0) \right| \le \frac{2}{j}.
\]
Since this holds for all $j\in\ZZ_+$, we conclude that
\[
\lim_{C\to\infty} \frac{k_+(C)}{C} = U(0) = \left[ \frac{z_{\alpha/2} }{\big( Q_{1-\alpha}^{1/2} - c_0n^{-1/2} \big)_+ } \right]^2.
\]
This completes the proof.

\subsection{Proof of \myCref{cor-sharpness}}\label{sec-cor-sharpness-proof}

By \myCref{thm-effective-sample-size-chi}, for $C$ sufficiently large,
\begin{align}
k_-(\alpha,C) 
&\ge 
\frac{Cz_{\alpha/2}^2}{2}\left[ \quantile_{1-\alpha/8}^{1/2}\left( \discchireal(\testfunction) \right) + c_0 n^{-1/2} \right]^{-2} \notag \\[4pt]
&\ge 
\frac{Cz_{\alpha/2}^2}{2} \min\left\{ \quantile^{-1}_{1-\alpha/8} \left( \discchireal(\testfunction) \right), ~ c_0^{-2} n \right\} \label{eqn-cor-sharpness-1} \\[4pt]
&\ge 
\min\left\{ \frac{Cz_{\alpha/2}^2}{2  \quantile_{1-\alpha/8} \left( \discchireal(\testfunction) \right)}, ~ n \right\}. \label{eqn-cor-sharpness-2}
\end{align}
With probability at least $1-\delta$, by \eqref{eqn-cor-sharpness-2},
\[
\widehat{k}(C) \ge \min \left\{ K, \, k_-(\alpha,C) \right\} \ge \min\left\{ K, \, n, \, \frac{Cz_{\alpha/2}^2}{2  \quantile_{1-\alpha/8} \left( \discchireal(\testfunction) \right)} \right\}.
\]
When this happens, by \eqref{eqn-cor-sharpness-1}, the selection confidence interval $\simCI\big(\widehat{k}(C);C\big)$ has width at most
\[
2z_{\alpha/2}\sqrt{\frac{C}{\widehat{k}(C)}}
\le
2\max\left\{ z_{\alpha/2}\sqrt{\frac{C}{K}}, ~ \sqrt{2}\quantile^{-1/2}_{1-\alpha/8} \left( \discchireal(\testfunction) \right), ~ \sqrt{2}c_0 n^{-1/2}  \right\}.
\]
This completes the proof.

\subsection{Proof of \myCref{lem-MTurk-CLT}}\label{sec-lem-MTurk-CLT-proof}

For each $i\in\ZZ_+$, we define $X_i = \frac{1}{B} \sum_{j=1}^B \simresponsetilde_{i,j}$, which is the average of the responses sampled from profile $\profile_i$. Then $\{X_i\}_{i=1}^{\infty}$ are independent, $\EE[X_i] = \mean$, and
\begin{align*}
\var(X_i)
&=
\var\left( \EE\left[ X_i \mid z_i \right] \right) + \EE\left[ \var\left( X_i \mid z_i \right) \right] \\[4pt]
&=
\var\left( \performancefunction( \profile_i, \testfunction ) \right) + \EE\left[ \frac{1}{B} \performancefunction( \profile_i, \testfunction ) \left( 1 - \performancefunction( \profile_i, \testfunction ) \right) \right]
=
\tau^2 + \frac{\sigma^2}{B}.
\end{align*}
For each $\kappa\in\ZZ_+$, we define $S_{\kappa} = \sum_{i=1}^{\kappa} X_i$, then $S_{\kappa} / \kappa = \simresponsebar_{k}$. We will apply the central limit theorem with Lyapunov's condition (e.g., Theorem 27.3 in \cite{Bil17}). To this end, we check that
\[
\limsup_{\kappa\to\infty} \left[\var(S_{\kappa})\right]^{-3/2}  \sum_{i=1}^{\kappa} \EE\left[ \left|X_i - \EE X_i \right|^3 \right]
\le 
\limsup_{\kappa\to\infty} \frac{8\kappa}{\kappa^{3/2}\left( \tau^2 + \frac{\sigma^2}{B} \right)^{3/2}}
=
0.
\]
Therefore, central limit theorem applies: as $\kappa\to\infty$,
\[
\sqrt{\kappa} \left( \simresponsebar_{k} - \mean \right) ~ \xrightarrow{~d~} ~ N \left( 0, \, \tau^2 + \frac{\sigma^2}{B} \right).
\]
By the law of large numbers, as $\kappa\to\infty$, $\simresponsebar_{k} = S_{\kappa} / \kappa \to \mean$ almost surely, so $\simsamplesd_k \to \sqrt{\mean(1-\mean)} = \sqrt{\tau^2 + \sigma^2}$ almost surely. By Slutsky's Theorem,
\[
\frac{\simresponsebar_{k} - \mean}{\simsamplesd_k/\sqrt{k}}
=
\sqrt{\kappa} \left( \simresponsebar_{k} - \mean \right) \cdot \frac{\sqrt{B}}{\simsamplesd_k}
 ~ \xrightarrow{~d~} ~ N \left(0, \, \frac{B\tau^2+\sigma^2}{\tau^2+\sigma^2} \right).
\]
This finishes the proof.

\section{Proofs for Concentration-Based Interval Analysis}

\subsection{Proof of \myCref{lem-Ber-concentration-KL}}\label{sec-lem-Ber-concentration-KL-proof}
For every $\theta\in(0,1)$, the function $r\mapsto \KLdiv{r}{\theta}$ is decreasing on $(0,\theta)$ and increasing on $(\theta,1)$. For every $\epsilon>0$, by the Chernoff bound for the Bernoulli distribution (e.g., Section 2.2 in \cite{BLM13}),
\begin{align*}
& \PP\left( \bar{x}_k < r \right) \le \exp\left( -k \cdot \KLdiv{r}{q} \right),\quad \forall r\in(0,q), \\[4pt]
& \PP\left( \bar{x}_k > r \right) \le \exp\left( -k \cdot \KLdiv{r}{q} \right),\quad \forall r\in(q,1).
\end{align*}
For $\epsilon>0$, let $\cC_L(q,\epsilon) = \left\{ r\in(0,1) : \KLdiv{r}{q} \le \epsilon  \right\}$, which is a closed interval containing $q$. Then,
\[
\PP\left( \bar{x}_k \not\in \cC_L(q,\epsilon) \right)
\le 
\PP\left( \bar{x}_k < \min \cC_L(q,\epsilon) \right) +  \PP\left( \bar{x}_k > \max \cC_L(q,\epsilon) \right)
\le 
2\exp(-k\epsilon).
\]
In other words,
\[
\PP\left( \KLdiv{\bar{x}_k}{q} \le \epsilon \right) \ge 1 - 2\exp(-k\epsilon).
\]
Setting $\epsilon = \log(2/\delta)/k$ finishes the proof.

\subsection{Tightness of the Concentration-Based Confidence Interval}\label{sec-lem-KL-Bernstein-proof}

The following lemma shows that the concentration-based confidence interval is always at least as tight as an interval based on the empirical-Bernstein bound \citep{MPo09}.

\begin{lemma}\label{lem-KL-Bernstein}
It always holds that 
\[
\left\{ p\in(0,1): \KLdiv{\bar{x}_k}{p} \le \frac{\log(2/\alpha)}{k} \right\}
\subseteq
\left[ \bar{x}_k - R_k,
~
\bar{x}_k + R_k
\right],
\] 
where
\[
R_k = s_k\sqrt{\frac{2\log(2/\alpha)}{k}} + \frac{2\log(2/\alpha)}{k}
\quad\text{and}\quad
s_k = \sqrt{\bar{x}_k(1-\bar{x}_k)}.
\]
\end{lemma}

\begin{proof}[Proof of \myCref{lem-KL-Bernstein}]
We will prove a more general result: for all $x\in(0,1)$ and $b>0$,
\[
\left\{ p\in(0,1) : \KLdiv{x}{p} \le b \right\} \subseteq \left[ x-R, x+R \right],
\quad\text{where}\quad
R = \sqrt{2x(1-x)b} + 2b.
\]
To prove it, since the function $z\mapsto z(1-z)$ is $1$-Lipschitz continuous on $[0,1]$, then for every $p\in(0,1)$,
\[
\KLdiv{x}{p} = \int_x^p \frac{u-x}{u(1-u)}\,du
\ge 
\int_x^p \frac{u-x}{x(1-x) + |p-x|} \,du = \frac{1}{2} \cdot \frac{(p-x)^2}{x(1-x) + |p-x|}.
\]
Thus,
\[
\left\{ p\in(0,1) : \KLdiv{x}{p} \le b \right\}
\subseteq
\left\{ p\in(0,1) : \frac{1}{2} \cdot \frac{(p-x)^2}{x(1-x) + |p-x|} \le b \right\} \subseteq [x-R,x+R].
\]
This completes the proof.
\end{proof}

\subsection{Proof of \myCref{thm-kappa-IT-KL}}\label{sec-thm-kappa-IT-KL-proof}
Let $q(\profiles_{\kappa}) = \quantile_{1-\alpha}\big( \discKL(\testfunction) \mid \profiles_{\kappa} \big)$. Take $m= \big\lceil \frac{\log(\delta/2)}{\log(1-\alpha)} \big\rceil$ i.i.d.~samples $\testfunction_1,...,\testfunction_m\sim\testfunctiondist$, independent of $\profiles_{\kappa}$. Then $m \le -\frac{\log(2/\delta)}{\log(1-\alpha)} + 1 \le \frac{\log(2/\delta)}{\alpha} + 1$. By the definition of quantile, for all $\epsilon > 0$,
\[
\PP\left(  \discKL(\testfunction_j) \le  q(\profiles_{\kappa}) - \epsilon \mid \profiles_{\kappa} \right) < 1 - \alpha, \qquad \forall j\in[m].
\]
By independence,
\[
\PP\left(  \max_{j\in[m]} \discKL(\testfunction_j) \le q(\profiles_{\kappa}) - \epsilon \Bigm| \profiles_{\kappa} \right)
=
\prod_{j=1}^m \PP\left(  \discKL(\testfunction_j) \le q(\profiles_{\kappa}) - \epsilon \mid \profiles_{\kappa} \right) < (1-\alpha)^m \le \frac{\delta}{2},
\]
which implies
\begin{equation}\label{eqn-kappa-IT-KL-proof-1}
\PP\left(  q(\profiles_{\kappa}) < \max_{j\in[m]} \discKL(\testfunction_j) + \epsilon \right) > 1 - \frac{\delta}{2}.
\end{equation}

We now derive a bound for $\max_{j\in[m]}\discKL(\testfunction_j)$. By \myCref{lem-Ber-concentration-KL}, for every $j\in[m]$,
\[
\PP\left( \discKL(\testfunction_j) \le \frac{\log(4m/\delta)}{\kappa} \Bigm| \testfunction_j \right) \ge 1 - \frac{\delta}{2m}.
\]
By a union bound over $j\in[m]$,
\begin{equation}\label{eqn-kappa-IT-KL-proof-2}
\PP\left( \max_{j\in[m]} \discKL (\testfunction_j) \le \frac{\log(4m/\delta)}{\kappa} \right) \ge 1 - \frac{\delta}{2}.
\end{equation}

Combining \eqref{eqn-kappa-IT-KL-proof-1} and \eqref{eqn-kappa-IT-KL-proof-2} with a union bound, we obtain that with probability at least $1-\delta$,
\begin{equation}\label{eqn-kappa-IT-KL-proof-3}
q(\profiles_{\kappa}) < \epsilon + \frac{\log(4m/\delta)}{\kappa} \le \epsilon + \frac{c'}{\kappa},
\end{equation}
where $c' = \log\big[ 4\delta^{-1}\big(1 + \alpha^{-1} \log(2/\delta)\big)  \big]$. Since this holds for all $\epsilon> 0$, then by the continuity of the probability measure $\PP$,
\[
\PP\left(
q(\profiles_{\kappa}) \le \frac{c'}{\kappa} \right) = \lim_{\epsilon\to 0^+}\PP\left( q(\profiles_{\kappa}) <
\epsilon + \frac{c'}{\kappa} \right) \ge 1 - \delta.
\]
This completes the proof.

\subsection{Proof of \myCref{thm-oracle-effective-sample-size-KL}}\label{sec-thm-oracle-effective-sample-size-KL-proof}

We write $Q_{\beta}=\quantile_{\beta}\left(\discKLreal(\testfunction)\right)$. \myCref{thm-oracle-effective-sample-size-KL} is a consequence of the following lemma.

\begin{lemma}\label{lem-oracle-effective-sample-size-KL}
In the setting of \myCref{thm-oracle-effective-sample-size-KL},
\[
\frac{\log(2/\alpha)}{Q_{1-\alpha/2}}
\le 
\liminf_{C \to\infty} \frac{k_{\KL}^*(C)}{C}
\le
\limsup_{C\to\infty}\frac{k_{\KL}^*(C)}{C}
\le 
\frac{\log(2/\alpha)}{Q_{1-2\alpha}}.
\]
\end{lemma}

\begin{proof}[Proof of \myCref{lem-oracle-effective-sample-size-KL}]
See \myCref{sec-lem-oracle-effective-sample-size-KL-proof}.
\end{proof}

Given \myCref{lem-oracle-effective-sample-size-KL}, we can derive \myCref{thm-oracle-effective-sample-size-KL} using the same argument as \myCref{sec-thm-oracle-effective-sample-size-KL-proof}.

\subsection{Proof of \myCref{lem-oracle-effective-sample-size-KL}}\label{sec-lem-oracle-effective-sample-size-KL-proof}

By definition, with probability at least $1-\alpha/2$,
\begin{equation}\label{eqn-lem-oracle-effective-sample-size-KL-1}
\discKLreal(\testfunction) \le Q_{1-\alpha/2}.
\end{equation}
Define a function $m(p,q) = \min\{p(1-p),q(1-q)\}$. By Assumption \myref{assumption-nondegeneracy-both}, with probability at least $1-\alpha/4$,
\begin{equation}\label{eqn-lem-oracle-effective-sample-size-KL-2}
m \big( \mean(\testfunction), \, \simmean(\testfunction)  \big) = \min \big\{ \mean(\testfunction)\left(1-\mean(\testfunction)\right), ~ \simmean(\testfunction)\left(1-\simmean(\testfunction)\right) \big\} \ge \eta(1-\eta).
\end{equation}
By \myCref{lem-Ber-concentration-KL-perturbed}, with probability at least $1-\alpha/4$,
\begin{multline}\label{eqn-lem-oracle-effective-sample-size-KL-3}
\left| \KLdiv{\simresponsebar_k}{\mean(\testfunction)} - \KLdiv{\simmean(\testfunction)}{\mean(\testfunction)} \right| \\[4pt]
\le 
\sqrt{\frac{\KLdiv{\simmean(\testfunction)}{\mean(\testfunction)}}{m \big( \mean(\testfunction), \, \simmean(\testfunction)  \big)}}\cdot \sqrt{\frac{\log(8/\alpha)}{k}} + \left( 1 + \frac{ \sqrt{2\KLdiv{\simmean(\testfunction)}{\mean(\testfunction)}} }{m \big( \mean(\testfunction), \, \simmean(\testfunction)  \big)}\right)\frac{\log(8/\alpha)}{k}. 
\end{multline}

\paragraph{Lower bound.} We first derive a lower bound for $\liminf_{C\to\infty} k_{\KL}^*(C)/C$. Applying a union bound to the three events in \eqref{eqn-lem-oracle-effective-sample-size-KL-1}, \eqref{eqn-lem-oracle-effective-sample-size-KL-2} and \eqref{eqn-lem-oracle-effective-sample-size-KL-3}, we see that the following inequality holds with probability at least $1-\alpha$:
\begin{equation}\label{eqn-lem-oracle-effective-sample-size-KL-lower-1}
\KLdiv{\simresponsebar_k}{\mean(\testfunction)}
\le 
Q_{1-\alpha/2}
+
\sqrt{\frac{Q_{1-\alpha/2}}{\eta(1-\eta)}}\cdot \sqrt{\frac{\log(8/\alpha)}{k}} + \left( 1 + \frac{ \sqrt{2Q_{1-\alpha/2}} }{\eta(1-\eta)}\right)\frac{\log(8/\alpha)}{k}.
\end{equation}
Denote this event by $\cE(k)$.

For a fixed $k\in\ZZ_+$ and $C\ge 1$, when $\cE(k)$ happens,
\begin{align}
&\quad \mean(\testfunction) \in \simCIKL(k;C) \notag \\[4pt]
\Leftrightarrow&\quad 
\KLdiv{ \simresponsebar_k}{p} \le \frac{C\log(2/\alpha)}{k} \notag \\[4pt]
\Leftarrow&\quad
Q_{1-\alpha/2}
+
\sqrt{\frac{Q_{1-\alpha/2}}{\eta(1-\eta)}}\cdot \sqrt{\frac{\log(8/\alpha)}{k}} + \left( 1 + \frac{ \sqrt{2Q_{1-\alpha/2}} }{\eta(1-\eta)}\right)\frac{\log(8/\alpha)}{k}
\le \frac{C\log(2/\alpha)}{k} \label{eqn-lem-oracle-effective-sample-size-KL-lower-2}
\end{align}
When $C\ge 1$ is sufficiently large, solving \eqref{eqn-lem-oracle-effective-sample-size-KL-lower-2} for $k$ yields
\[
k \le \frac{C\log(2/\alpha) - \xi}{Q_{1-\alpha/2}}\left( 1 + \sqrt{\frac{\log(8/\alpha) \cdot \eta^{-1}(1-\eta)^{-1}}{C\log(2/\alpha) - \xi}} \right)^{-2},
\]
where $\xi=\left( 1 + \frac{ \sqrt{2 Q_{1-\alpha/2}} }{\eta(1-\eta)}\right)\log(8/\alpha)$. Let 
\[
k^*_-(C) = \bigg\lfloor \frac{C\log(2/\alpha) - \xi}{Q_{1-\alpha/2}}\left( 1 + \sqrt{\frac{\log(8/\alpha) \cdot \eta^{-1}(1-\eta)^{-1}}{C\log(2/\alpha) - \xi}} \right)^{-2} \bigg\rfloor,
\] 
then $\PP\big(\mean(\testfunction) \in \simCIKL(k^*_-(C);C) \big) \ge \PP\big( \cE(k^*_-(C)) \big) \ge  1-\alpha$, and thus $k_{\KL}^*(C) \ge k^*_-(C)$. Therefore,
\[
\limsup_{C\to\infty} \frac{k_{\KL}^*(C)}{C} \ge \lim_{C\to\infty} \frac{k_-^*(C)}{C} = \frac{\log(2/\alpha)}{Q_{1-\alpha/2}}.
\]
This proves the lower bound.

\paragraph{Upper bound.} Next we derive an upper bound for $\limsup_{C\to\infty} k_{\KL}^*(C)/C$. Take arbitrary $\epsilon \in (0, Q_{1-2\alpha})$. Then with probability strictly greater than $2\alpha$,
\begin{equation}\label{eqn-lem-oracle-effective-sample-size-KL-4}
\discKLreal(\testfunction) \ge Q_{1-2\alpha} - \epsilon.
\end{equation}
Applying the union bound to the four events in \eqref{eqn-lem-oracle-effective-sample-size-KL-1}, \eqref{eqn-lem-oracle-effective-sample-size-KL-2}, \eqref{eqn-lem-oracle-effective-sample-size-KL-3} and \eqref{eqn-lem-oracle-effective-sample-size-KL-4}, we see that the following inequality holds with probability strictly greater than $\alpha$:
\begin{equation}\label{eqn-lem-oracle-effective-sample-size-KL-upper-1}
\KLdiv{\simresponsebar_k}{\mean(\testfunction)}
\ge 
Q_{1-2\alpha} - \epsilon
-
\sqrt{\frac{Q_{1-\alpha/2}}{\eta(1-\eta)}}\cdot \sqrt{\frac{\log(8/\alpha)}{k}} - \left( 1 + \frac{ \sqrt{2Q_{1-\alpha/2}} }{\eta(1-\eta)}\right)\frac{\log(8/\alpha)}{k}.
\end{equation}
Denote this event by $\cE'(k)$. 

For a fixed $k\in\ZZ_+$ and $C\ge 1$, when $\cE'(k)$ happens,
\begin{align}
&\quad \mean(\testfunction) \not\in \simCIKL(k;C) \notag \\[4pt]
\Leftrightarrow&\quad 
\KLdiv{ \simresponsebar_k}{p} > \frac{C\log(2/\alpha)}{k} \notag \\[4pt]
\Leftarrow&\quad
Q_{1-2\alpha}
-
\sqrt{\frac{Q_{1-\alpha/2}}{\eta(1-\eta)}}\cdot \sqrt{\frac{\log(8/\alpha)}{k}} - \left( 1 + \frac{ \sqrt{2Q_{1-\alpha/2}} }{\eta(1-\eta)}\right)\frac{\log(8/\alpha)}{k}
> \frac{C\log(2/\alpha)}{k} \label{eqn-lem-oracle-effective-sample-size-KL-upper-2}
\end{align}
When $C\ge 1$ is sufficiently large, solving \eqref{eqn-lem-oracle-effective-sample-size-KL-upper-2} yields
\[
k > \big( C\log(2/\alpha) + \xi \big) \left( \sqrt{Q_{1-2\alpha}-\epsilon} + \frac{1}{2} \sqrt{\frac{Q_{1-\alpha/2}\log(8/\alpha) \cdot \eta^{-1}(1-\eta)^{-1}}{C\log(2/\alpha) + \xi}} \right)^{-2},
\]
where $\xi=\left( 1 + \frac{ \sqrt{2 Q_{1-\alpha/2}} }{\eta(1-\eta)}\right)\log(8/\alpha)$. Let 
\[
k^*_+(C,\epsilon) = \bigg\lfloor \big( C\log(2/\alpha) + \xi \big) \left( \sqrt{Q_{1-2\alpha}-\epsilon} + \frac{1}{2} \sqrt{\frac{Q_{1-\alpha/2}\log(8/\alpha) \cdot \eta^{-1}(1-\eta)^{-1}}{C\log(2/\alpha) + \xi}} \right)^{-2} \bigg\rfloor + 1.
\] 
Then $\PP\Big( \mean(\testfunction) \not\in \simCIKL(k^*_+(C,\epsilon);C) \Big) \ge \PP\big( \cE'(k_+^*(C,\epsilon)) \big) > \alpha$, which implies \[
\PP\left( \mean(\testfunction) \in \simCIKL(k^*_+(C,\epsilon);C) \right) < 1 - \alpha.
\] 
Thus, $k_{\KL}^*(C) < k^*_+(C)$, and
\[
\liminf_{C\to\infty} \frac{k_{\KL}^*(C)}{C} \le \lim_{C\to\infty} \frac{k^*_+(C,\epsilon)}{C} = \frac{\log(2/\alpha)}{Q_{1-2\alpha}-\epsilon}.
\]
Since this holds for arbitrarily small $\epsilon>0$, then
\[
\limsup_{C\to\infty}\frac{k^*(C)}{C}
\le
\frac{\log(2/\alpha)}{Q_{1-2\alpha}}.
\]
This proves the upper bound.

\subsection{Proof of \myCref{thm-effective-sample-size-KL}}\label{sec-thm-effective-sample-size-KL-proof}
Take i.i.d.~samples $\{y_i\}_{i=1}^n\sim\Bernoulli(\mean(\testfunction))$, and set $\responsebar = \frac{1}{n} \sum_{i=1}^n y_i$. By Hoeffding's inequality and a union bound, with probability at least $1-\delta$,
\begin{equation}\label{eqn-thm-effective-sample-size-KL-0}
\left|\frac{1}{m} \sum_{j=1}^m \ind \left\{ \responsebar_j \in \simCIKL_j(k;C) \right\}
-
\PP\big( \responsebar \in \simCIKL(k;C) \big) \right|
\le 
\sqrt{\frac{\log(2K/\delta)}{2m}}, \quad \forall k\in[K].
\end{equation}
From now on assume that \eqref{eqn-thm-effective-sample-size-KL-0} happens. When $m\ge 32\alpha^{-1}\log(2K/\delta)$, we have $\sqrt{\frac{\log(K/\delta)}{2m}}\le \alpha/8$. Thus,
\begin{align}
&\widehat{k}(C)
\ge 
\max\left\{ 0\le k\le K:
\PP\big( \responsebar \in \simCIKL(i;C) \big) \ge 1 - 3\alpha/8,\ \forall i \le k
\right\}, \label{eqn-thm-effective-sample-size-KL-lower-0} \\[4pt]
&\widehat{k}(C)
\le 
\max\left\{ 0\le k\le K:
\PP\big( \responsebar \in \simCIKL(i;C) \big) \ge 1 - 5\alpha/8,\ \forall i \le k
\right\}. \label{eqn-thm-effective-sample-size-KL-upper-0}
\end{align}
By definition,
\[
\PP\big( \responsebar \in \simCIKL(k;C) \big)
=
\PP\left( \KLdiv{\simresponse_k}{\responsebar} \le \frac{C\log(2/\alpha)}{k} \right).
\]

For notational convenience, we let $Q_{\beta} = \quantile_{\beta} \left( \Delta(\testfunction) \right)$. By \myCref{lem-Ber-concentration-KL-perturbed-double}, there exists a universal constant $c_1>0$ such that when $n\ge c_1\log(1+\alpha^{-1})\cdot \eta^{-1}(1-\eta)^{-1}$, the following happens with probability at least $1-\alpha/8$,
\begin{align}
\left| \KLdiv{\simresponse_k}{\responsebar} - \Delta(\testfunction) \right|
&\lesssim 
\frac{\sqrt{v\left(\simmean(\testfunction)\right)\cdot \Delta(\testfunction)}}{m(\simmean(\testfunction),\mean(\testfunction))} \cdot \sqrt{\frac{\log(1+\alpha^{-1})}{k}}
+
\sqrt{\frac{\Delta(\testfunction)}{v\left(\mean(\testfunction)\right)}} \cdot \sqrt{\frac{\log(1+\alpha^{-1})}{n}} \notag \\[4pt]
&\qquad+
\left(1 + \frac{\sqrt{\Delta(\testfunction)}}{m(\simmean(\testfunction),\mean(\testfunction))} \right) \left[ \frac{\log(1+\alpha^{-1})}{k} + \frac{\log(1+\alpha^{-1})}{n} \right] \notag \\[4pt]
&\qquad +
\sqrt{\frac{v\left(\simmean(\testfunction)\right)}{v\left(\mean(\testfunction)\right)}}\frac{\log(1+\alpha^{-1})}{kn}. \label{eqn-thm-effective-sample-size-KL-1}
\end{align}
Here $v(p) = p(1-p)$, $m(p,q) = \min\{v(p),v(q)\}$, and $\lesssim$ hides a universal constant. 
By definition, with probability at least $1-\alpha/8$,
\begin{equation}\label{eqn-thm-effective-sample-size-KL-2}
\Delta(\testfunction) \le Q_{1-\alpha/8}.
\end{equation} 
By Assumption \myref{assumption-nondegeneracy-both}, with probability at least $1-\alpha/8$,
\begin{equation}\label{eqn-thm-effective-sample-size-KL-3}
m \big( \mean(\testfunction), \, \simmean(\testfunction)  \big) = \min \big\{ \mean(\testfunction)\left(1-\mean(\testfunction)\right), ~ \simmean(\testfunction)\left(1-\simmean(\testfunction)\right) \big\} \ge \eta(1-\eta).
\end{equation}

\paragraph{Lower bound.} We first derive a lower bound for $\widehat{k}_{\KL}(C)$. 
Applying a union bound to the events in \eqref{eqn-thm-effective-sample-size-KL-1}, \eqref{eqn-thm-effective-sample-size-KL-2} and \eqref{eqn-thm-effective-sample-size-KL-3}, we obtain that for every $k\in[K]$, with probability at least $1-3\alpha/8$,
\begin{multline}\label{eqn-thm-effective-sample-size-KL-lower-1}
\KLdiv{\simresponse_k}{\responsebar}
\le 
Q_{1-\alpha/8}
+
C_0 \bigg\{
\frac{\sqrt{Q_{1-\alpha/8}}}{\eta(1-\eta)} \cdot \sqrt{\frac{\log(1+\alpha^{-1})}{k}}
+
\sqrt{\frac{Q_{1-\alpha/8}}{\eta(1-\eta)}} \cdot \sqrt{\frac{\log(1+\alpha^{-1})}{n}} \\[4pt]
+
\left(1 + \frac{\sqrt{Q_{1-\alpha/8}}}{\eta(1-\eta)} \right) \left[ \frac{\log(1+\alpha^{-1})}{k} + \frac{\log(1+\alpha^{-1})}{n} \right]
+
\sqrt{\frac{1}{\eta(1-\eta)}}\frac{\log(1+\alpha^{-1})}{kn} \bigg\},
\end{multline}
where $C_0>0$ is a universal constant. Denote the right hand side of \eqref{eqn-thm-effective-sample-size-KL-lower-1} by $b(n,k)$, and let
\[
k_-(C) = \sup\left\{ k\in\ZZ_+ : b(n,i) \le \frac{C\log(2/\alpha)}{i},\ \forall i\le k \right\}.
\]
It is easy to see that $k_-(C)$ is finite and increasing in $C$, and as $C\to\infty$, $\widetilde{k}_-(C)\to\infty$. Moreover, $\PP\Big( \responsebar \in \simCIKL(i;C) \Big) \ge 1-3\alpha/8$ for all $i\le k_-(C)$. Therefore, by \eqref{eqn-thm-effective-sample-size-KL-lower-0},
\[
\widehat{k}(C) \ge \min \{  k_-(C), \, K\}.
\]

By definition,
\[
b\left( n,k_-(C)+1 \right) > \frac{C\log(2/\alpha)}{k_-(C)+1}.
\]
Taking $\limsup$ as $C\to\infty$ on both sides yields
\begin{multline*}
Q_{1-\alpha/8} + C_0 \left\{ \sqrt{\frac{Q_{1-\alpha/8}}{\eta(1-\eta)}} \cdot \sqrt{\frac{\log(1+\alpha^{-1})}{n}} + \left(1 + \frac{\sqrt{Q_{1-\alpha/8}}}{\eta(1-\eta)} \right)\frac{\log(1+\alpha^{-1})}{n} \right\} \\
\ge 
\log(2/\alpha)\cdot \limsup_{C\to\infty} \frac{C}{k_-(C)}.
\end{multline*}
Thus,
\begin{align*}
\liminf_{C\to\infty}\frac{k_-(C)}{C} 
&\ge 
\log(2/\alpha) \cdot \left( Q_{1-\alpha/8} + c_2'\sqrt{\frac{Q_{1-\alpha/8}}{\eta(1-\eta)}} \cdot \frac{\log(1+\alpha^{-1})}{n} \right)^{-1} \\[4pt]
&\ge
\frac{\log(2/\alpha)}{\big( Q_{1-\alpha/8}^{1/2} + c_2n^{-1/2} \big)^2}
\end{align*}
for some constants $c_2>0$ depending on $\alpha$ and $\eta$. By \myCref{lem-liminf-to-lim}, we can further construct $\overline{k}_-(C)$ such that $\widehat{k}(C)\ge \min\big\{K, \overline{k}_-(C) \big\}$ and
\[
\lim_{C\to\infty}\frac{\overline{k}_-(C)}{C} = \frac{\log(2/\alpha)}{\big( Q_{1-\alpha/8}^{1/2} + c_2n^{-1/2} \big)^2}.
\]

\paragraph{Upper bound.} Take arbitrary $\epsilon\in(0,Q_{1-\alpha})$. By definition, with probability strictly greater than $\alpha$,
\begin{equation}\label{eqn-thm-effective-sample-size-KL-upper-1}
\Delta(\testfunction) > Q_{1-\alpha} - \epsilon.
\end{equation} 
Applying a union bound to the events in \eqref{eqn-thm-effective-sample-size-KL-1}, \eqref{eqn-thm-effective-sample-size-KL-2}, \eqref{eqn-thm-effective-sample-size-KL-3} and \eqref{eqn-thm-effective-sample-size-KL-upper-1}, we obtain that for every $k\in[K]$, with probability strictly greater than $5\alpha/8$,
\begin{multline}\label{eqn-thm-effective-sample-size-KL-upper-2}
\KLdiv{\simresponse_k}{\responsebar}
\ge 
Q_{1-\alpha} - \epsilon
-
C_0 \bigg\{
\frac{\sqrt{Q_{1-\alpha/8}}}{\eta(1-\eta)} \cdot \sqrt{\frac{\log(1+\alpha^{-1})}{k}}
+
\sqrt{\frac{Q_{1-\alpha/8}}{\eta(1-\eta)}} \cdot \sqrt{\frac{\log(1+\alpha^{-1})}{n}} \\[4pt]
+
\left(1 + \frac{\sqrt{Q_{1-\alpha/8}}}{\eta(1-\eta)} \right) \left[ \frac{\log(1+\alpha^{-1})}{k} + \frac{\log(1+\alpha^{-1})}{n} \right]
+
\sqrt{\frac{1}{\eta(1-\eta)}}\frac{\log(1+\alpha^{-1})}{kn} \bigg\},
\end{multline}
where $C_0>0$ is a universal constant. Denote the right hand side of \eqref{eqn-thm-effective-sample-size-KL-upper-2} by $b'(n,k)-\epsilon$, and let
\[
k_+(C,\epsilon) = \inf\left\{ k\in\ZZ_+ : b'(n,k) - \epsilon \ge \frac{C\log(2/\alpha)}{k} \right\}.
\]
It is easy to see that $k_+(C,\epsilon)$ is increasing in $C$, and as $C\to\infty$, $k_+(C,\epsilon)\to \infty$. 

If $k_+(C,\epsilon)<\infty$ for all $C\ge 1$ sufficiently large, then $\PP\big(  \responsebar \not\in \simCIKL(k_+(C,\epsilon);C) \big) > 5\alpha/8$, so $\PP\big(  \responsebar \in \simCIKL(k_+(C,\epsilon);C) \big) < 1- 5\alpha/8$. Thus, by \eqref{eqn-thm-effective-sample-size-KL-upper-0},
\[
\widehat{k}(C) \le \min\big\{ K, \, k_+(C,\epsilon) \big\}.
\] 
By definition,
\[
b'\left( n,k_+(C,\epsilon)-1 \right) - \varepsilon < \frac{C\log(2/\alpha)}{k_+(C,\epsilon)-1}.
\]
Taking $\liminf$ as $C\to\infty$ on both sides yields
\begin{multline*}
Q_{1-\alpha} - \epsilon - C_0 \left\{ \sqrt{\frac{Q_{1-\alpha/8}}{\eta(1-\eta)}} \cdot \sqrt{\frac{\log(1+\alpha^{-1})}{n}} + \left(1 + \frac{\sqrt{Q_{1-\alpha/8}}}{\eta(1-\eta)} \right)\frac{\log(1 + \alpha^{-1})}{n} \right\} \\
\le 
\log(2/\alpha)\cdot \liminf_{C\to\infty} \frac{C}{k_+(C,\epsilon)}.
\end{multline*}
Since this holds for arbitrarily small $\epsilon>0$, then
\begin{align*}
\limsup_{C\to\infty}\frac{k_+(C,\epsilon)}{C} 
&\le
\log(2/\alpha) \cdot \left( Q_{1-\alpha} - \epsilon - c_3'\sqrt{\frac{Q_{1-\alpha/8}}{\eta(1-\eta)}} \cdot \frac{\log(1+\alpha^{-1})}{n} \right)^{-1} \\[4pt]
&\le
\frac{\log(2/\alpha)}{\big( (Q_{1-\alpha}-\epsilon)^{1/2} - c_3 n^{-1/2} \big)_+^2}
\end{align*}
for some constant $c_3>0$ depending on $\eta$ and $\alpha$. By \myCref{lem-liminf-to-lim}, we can further construct $\overline{k}_+(C,\epsilon)$ such that $\widehat{k}(C)\le \min\big\{K, \overline{k}_+(C,\epsilon) \big\}$ and
\[
\lim_{C\to\infty}\frac{\overline{k}_+(C,\epsilon)}{C} = \frac{\log(2/\alpha)}{\big( (Q_{1-\alpha}-\epsilon)^{1/2} - c_3 n^{-1/2} \big)_+^2}.
\]
When $k_+(C,\epsilon)=\infty$ for some $C$, the right hand side of the above must be $\infty$ as well, and the result still holds. 

Following a similar argument as the last part of the proof of \myCref{thm-effective-sample-size-chi}, we can construct $\widetilde{k}_+(C)$ such that $\widehat{k}(C)\le \min\big\{K, \widetilde{k}_+(C) \big\}$ and
\[
\lim_{C\to\infty}\frac{\widetilde{k}_+(C)}{C} = \frac{\log(2/\alpha)}{\big( Q_{1-\alpha}^{1/2} - c_3 n^{-1/2} \big)_+^2}.
\]
This finishes the proof.

\section{Technical Lemmas}

\begin{lemma}[Reversed KL-based Bernoulli concentration bound]\label{lem-Ber-concentration-reverse-KL}
Fix $p\in(0,1)$. Let $\{y_i\}_{i=1}^n$ be i.i.d.~samples of $\Bernoulli(p)$. Define $\bar{y}_n = \frac{1}{n}\sum_{i=1}^n y_i$. Fix $\delta\in(0,1)$. There exist universal constants $c_1,c_2>0$ such that
\[
\PP\left( \KLdiv{p}{\bar{y}_n} \le \frac{c_2\log(2/\delta)}{n} \right) \ge 1-\delta,
\quad
\forall n \ge \frac{c_1\log(2/\delta)}{p(1-p)}.
\]
\end{lemma}

\begin{proof}[Proof of \myCref{lem-Ber-concentration-reverse-KL}]
By Bernstein's inequality (e.g., (2.10) in \cite{BLM13}),
\begin{equation}\label{eqn-lem-Ber-concentration-reverse-KL-proof-1}
\PP\left( \left|\bar{y}_n - p\right| \le \sqrt{\frac{2p(1-p)\log(2/\delta)}{n}} + \frac{2\log(2/\delta)}{3n}  \right) \ge 1 - \delta.
\end{equation}
From now on suppose that the event in \eqref{eqn-lem-Ber-concentration-reverse-KL-proof-1} happens. There exists a universal constant $c>0$ such that when $n\ge c\log(2/\delta)\cdot p^{-1}(1-p)^{-1}$,
\[
\left|\bar{y}_n - p\right| \le \sqrt{\frac{2p(1-p)\log(2/\delta)}{n}} + \frac{2\log(2/\delta)}{3n}
\le 
\frac{p(1-p)}{2}.
\]
In this case, $\bar{y}_n\in(0,1)$. Let $g(r)=\KLdiv{p}{r}$, $r\in(0,1)$, then $g'(r) = \frac{r-p}{r(1-r)}$. By the mean value theorem, there exists $\xi$ between $\bar{y}_n$ and $p$ such that
\[
\KLdiv{p}{\bar{y}_n} = \frac{\xi-p}{\xi(1-\xi)}\left(\bar{y}_n - p \right).
\]
Since $(\xi-p)(\bar{y}_n - p) \le \left( \bar{y}_n - p \right)^2$ and $\xi(1-\xi) \ge p(1-p) - |\xi-p| \ge p(1-p)/2$, then
\begin{align*}
\KLdiv{p}{\bar{y}_n}
&\le 
\frac{2}{p(1-p)} \cdot \left( \bar{y}_n - p \right)^2 \\[4pt]
&\lesssim \frac{1}{p(1-p)} \left( \frac{p(1-p)\log(2/\delta)}{n} + \frac{\log(2/\delta)^2}{n^2} \right)
\lesssim
\frac{\log(2/\delta)}{n}.
\end{align*}
In the last inequality, we have used the fact that $n\gtrsim \log(2/\delta)\cdot p^{-1}(1-p)^{-1}$.
\end{proof}

\begin{lemma}[KL-based perturbed Bernoulli concentration bound]\label{lem-Ber-concentration-KL-perturbed}
Fix $q,p\in(0,1)$. Let $\{x_i\}_{i=1}^k$ be i.i.d.~samples of $\Bernoulli(q)$. Define $\bar{x}_k = \frac{1}{k}\sum_{i=1}^k x_i$. For every $\delta\in(0,1)$, with probability at least $1-\delta$,
\[
\left| \KLdiv{\bar{x}_k}{ p} - \KLdiv{q }{ p} \right|
\le 
\frac{\sqrt{q(1-q)\cdot \KLdiv{q}{p}}}{m(p,q)}\sqrt{\frac{\log(2/\delta)}{k}} + \left( 1 + \frac{ \sqrt{2\KLdiv{q}{p}} }{m(p,q)}\right)\frac{\log(2/\delta)}{k},
\]
where $m(p,q) = \min \left\{ p(1-p), q(1-q) \right\}$.
\end{lemma}

\begin{proof}[Proof of \myCref{lem-Ber-concentration-KL-perturbed}]
By \myCref{lem-KL-three-point-identity} with $a=\bar{x}_k$, $A=q$ and $b=p$, we have
\begin{equation}\label{eqn-lem-Ber-concentration-KL-perturbed-proof-0}
\KLdiv{\bar{x}_k}{ p} - \KLdiv{q }{ p}
=
\KLdiv{\bar{x}_k}{ q} + \left(\bar{x}_k - q\right) \cdot \left( \log\frac{q}{1-q} - \log\frac{p}{1-p}  \right). 
\end{equation}
By \myCref{lem-Ber-concentration-KL} and since Bernstein's inequality can be derived from \myCref{lem-Ber-concentration-KL}, the following two inequalities hold with probability at least $1-\delta$:
\begin{align}
&\KLdiv{\bar{x}_k }{ q} \le \frac{\log(2/\delta)}{k},\label{eqn-lem-Ber-concentration-KL-perturbed-proof-1} \\[4pt]
&\left|\bar{x}_k - q\right| \le \sqrt{\frac{2q(1-q)\log(2/\delta)}{k}} + \frac{2\log(2/\delta)}{3k}. \label{eqn-lem-Ber-concentration-KL-perturbed-proof-2}
\end{align}
By \myCref{lem-local-Lipschitz-log-x/(1-x)},
\begin{equation}\label{eqn-lem-Ber-concentration-KL-perturbed-proof-3}
\left| \log\frac{q}{1-q} - \log\frac{p}{1-p}  \right|
\le 
\frac{|p-q|}{m(p,q)}.
\end{equation}
Substituting \eqref{eqn-lem-Ber-concentration-KL-perturbed-proof-1}, \eqref{eqn-lem-Ber-concentration-KL-perturbed-proof-2} and \eqref{eqn-lem-Ber-concentration-KL-perturbed-proof-3} into \eqref{eqn-lem-Ber-concentration-KL-perturbed-proof-0} yields that with probability at least $1-\delta$,
\[
\left| \KLdiv{\bar{x}_k}{ p } - \KLdiv{q}{ p } \right|
\le 
\frac{\log(4/\delta)}{k} + \left( \sqrt{\frac{2q(1-q)\log(2/\delta)}{k}} + \frac{2\log(2/\delta)}{3k} \right) \cdot \frac{|p-q|}{m(p,q)}.
\]
We finish the proof by invoking Pinsker's inequality $|p-q| \le \sqrt{\KLdiv{q}{p}/2}$ (e.g., Theorem 4.19 in \cite{BLM13}).
\end{proof}

\begin{lemma}[KL-based perturbed two-sample Bernoulli concentration bound]\label{lem-Ber-concentration-KL-perturbed-double}
Fix $q,p\in(0,1)$. Let $\{x_i\}_{i=1}^k$ be i.i.d.~samples of $\Bernoulli(q)$, and $\{y_i\}_{i=1}^n$ i.i.d.~samples of $\Bernoulli(p)$. Suppose $\{x_i\}_{i=1}^k$ and $\{y_i\}_{i=1}^n$ are independent. Define $\bar{x}_k = \frac{1}{k}\sum_{i=1}^k x_i$ and $\bar{y}_n = \frac{1}{n} \sum_{i=1}^n y_i$. Fix $\delta\in(0,1)$. There exists a universal constant $c>0$ such that for all $k\in\ZZ_+$ and $n\ge c\log(4/\delta)\cdot p^{-1}(1-p)^{-1}$, the following inequality holds with probability $1-\delta$:
\begin{multline*}
\left| \KLdiv{\bar{x}_k}{\bar{y}_n} - \KLdiv{q }{ p} \right|
\lesssim 
\frac{\sqrt{q(1-q)\KLdiv{q}{p}}}{m(p,q)} \cdot \sqrt{\frac{\log(4/\delta)}{k}}
+
\sqrt{\frac{\KLdiv{q}{p}}{p(1-p)}} \cdot \sqrt{\frac{\log(4/\delta)}{n}} \\[4pt]
+
\left(1 + \frac{\sqrt{\KLdiv{q}{p}}}{m(p,q)} \right) \left[ \frac{\log(4/\delta)}{k} + \frac{\log(4/\delta)}{n} \right]
+
\sqrt{\frac{q(1-q)}{p(1-p)}}\frac{\log(4/\delta)}{kn},
\end{multline*}
where $m(p,q) = \min \left\{ p(1-p), q(1-q) \right\}$, and $\lesssim$ hides a universal constant.
\end{lemma}

\begin{proof}[Proof of \myCref{lem-Ber-concentration-KL-perturbed-double}]
By \myCref{lem-KL-three-point-identity},
\begin{align*}
& \KLdiv{\bar{x}_k}{ \bar{y}_n} 
= 
\KLdiv{q }{ \bar{y}_n} + \KLdiv{\bar{x}_k}{ q} + \left(\bar{x}_k - q\right) \cdot \left( \log\frac{q}{1-q} - \log\frac{\bar{y}_n}{1-\bar{y}_n}  \right), \\[4pt]
& \KLdiv{q}{ \bar{y}_n} = 
\KLdiv{q }{ p} + \KLdiv{p}{\bar{y}_n} + \left(q - p\right) \cdot \left( \log\frac{p}{1-p} - \log\frac{\bar{y}_n}{1-\bar{y}_n}  \right),
\end{align*}
which implies
\begin{align}
&\KLdiv{\bar{x}_k}{ \bar{y}_n} - \KLdiv{q }{ p} \notag \\[4pt]
&=
\KLdiv{\bar{x}_k}{q} + \KLdiv{p}{\bar{y}_n} \notag \\[4pt]
&\qquad + \left(\bar{x}_k - q\right) \cdot \left( \log\frac{q}{1-q} - \log\frac{\bar{y}_n}{1-\bar{y}_n}  \right)
+ \left(q - p\right) \cdot \left( \log\frac{p}{1-p} - \log\frac{\bar{y}_n}{1-\bar{y}_n}  \right) \notag \\[4pt]
&=
\KLdiv{\bar{x}_k}{q} + \KLdiv{p}{\bar{y}_n} \notag \\
&\qquad+
(\bar{x}_k-q)\cdot \left(\log\frac{q}{1-q} - \log\frac{p}{1-p} \right)
+
\left[ (\bar{x}_k-q) + (q-p) \right] \cdot \left( \log\frac{p}{1-p} - \log\frac{\bar{y}_n}{1-\bar{y}_n}  \right). \label{eqn-lem-Ber-concentration-KL-perturbed-double-proof-0}
\end{align}

By \myCref{lem-Ber-concentration-KL}, with probability at least $1-\delta/2$,
\begin{align*}
& \KLdiv{\bar{x}_k}{q} \le \frac{\log(4/\delta)}{k}, \\[4pt]
& \left|\bar{x}_k - q\right| \le \sqrt{\frac{2q(1-q)\log(4/\delta)}{k}} + \frac{2\log(4/\delta)}{3k}.
\end{align*}
By \myCref{lem-Ber-concentration-reverse-KL}, there exists a universal constant $c_1>0$ such that when $n\ge c_1\log(4/\delta)\cdot p^{-1}(1-p)^{-1}$, with probability at least $1-\delta/2$,
\begin{align*}
& \KLdiv{ p}{\bar{y}_n} 
\lesssim
\frac{\log(4/\delta)}{n}, \\[4pt]
& \left|\bar{y}_n - p\right| \le \sqrt{\frac{2p(1-p)\log(4/\delta)}{n}} + \frac{2\log(4/\delta)}{3n} \le \frac{p(1-p)}{2}.
\end{align*}
By \myCref{lem-local-Lipschitz-log-x/(1-x)},
\begin{align*}
& \left| \log\frac{p}{1-p} - \log\frac{\bar{y}_n}{1-\bar{y}_n}  \right|
\le 
\frac{2}{p(1-p)} \cdot \left[ \sqrt{\frac{2q(1-p)\log(4/\delta)}{n}} + \frac{2\log(4/\delta)}{3n} \right], \\[4pt]
& \left| \log\frac{q}{1-q} - \log\frac{p}{1-p} \right| \le \frac{|p-q|}{m(p,q)}.
\end{align*}

Substituting the inequalities above into \eqref{eqn-lem-Ber-concentration-KL-perturbed-double-proof-0} and applying a union bound, we obtain that with probability at least $1-\delta$,
\begin{align*}
& \left| \KLdiv{\bar{x}_k}{ \bar{y}_n} - \KLdiv{q }{ p} \right| \\[6pt]
&\lesssim
\frac{\log(4/\delta)}{k} + \frac{\log(4/\delta)}{n}
+
\left[ \sqrt{\frac{q(1-q)\log(4/\delta)}{k}} + \frac{\log(4/\delta)}{k} \right] \cdot \frac{|p-q|}{m(p,q)} \\[4pt]
&\qquad +
\left[ \sqrt{\frac{q(1-q)\log(4/\delta)}{k}} + \frac{\log(4/\delta)}{k} \right]\cdot \frac{1}{p(1-p)} \left[ \sqrt{\frac{p(1-p)\log(4/\delta)}{n}} + \frac{\log(4/\delta)}{n}  \right] \\[4pt]
&\qquad +
|q-p|\cdot \frac{1}{p(1-p)}\left[ \sqrt{\frac{p(1-p)\log(4/\delta)}{n}} + \frac{\log(4/\delta)}{n}  \right] \\[6pt]
&\lesssim
\frac{\sqrt{q(1-q)\KLdiv{q}{p}}}{m(p,q)} \cdot \sqrt{\frac{\log(4/\delta)}{k}}
+
\sqrt{\frac{\KLdiv{q}{p}}{p(1-p)}} \cdot \sqrt{\frac{\log(4/\delta)}{n}} \\[4pt]
&\qquad +
\left(1 + \frac{\sqrt{\KLdiv{q}{p}}}{m(p,q)} \right) \left[ \frac{\log(4/\delta)}{k} + \frac{\log(4/\delta)}{n} \right]
+
\sqrt{1 + \frac{q(1-q)}{p(1-p)}}\frac{\log(4/\delta)}{kn}.
\end{align*}
Here we have used Pinsker's inequality $|p-q| \le \sqrt{\KLdiv{q}{p}/2}$ (e.g., Theorem 4.19 in \cite{BLM13}).
\end{proof}

\begin{lemma}[Three-point identity]\label{lem-KL-three-point-identity}
For all $a,A,b\in(0,1)$,
\[
\KLdiv{a}{b} = \KLdiv{a}{A} + \KLdiv{A}{b} + (a-A) \left( \log\frac{A}{1-A} - \log\frac{b}{1-b} \right).
\]
\end{lemma}

\begin{proof}[Proof of \myCref{lem-KL-three-point-identity}]
This is verified by direct calculation.
\end{proof}

\begin{lemma}\label{lem-local-Lipschitz-log-x/(1-x)}
For all $p,q\in(0,1)$,
\[
\left| \log\frac{p}{1-p} - \log\frac{q}{1-q} \right| \le \frac{|p-q|}{m(p,q)},
\]
where $m(p,q) = \min \left\{ p(1-p), q(1-q) \right\}$.
\end{lemma}

\begin{proof}[Proof of \myCref{lem-local-Lipschitz-log-x/(1-x)}]
Without loss of generality assume $p<q$. Let $f(r) = \log\frac{r}{1-r} - \log\frac{p}{1-p}$, then $f'(r) = \frac{1}{r(1-r)}$, so by the mean value theorem, there exists $\xi\in(p,q)$ such that
\[
\left| \log\frac{q}{1-q} - \log\frac{p}{1-p}  \right|
=
\frac{|p-q|}{\xi(1-\xi)}
\le 
\max_{r\in [p,q]}\frac{|p-q|}{r(1-r)}
=
\frac{|p-q|}{m(p,q)}.
\]
This finishes the proof.
\end{proof}

\begin{lemma}\label{lem-liminf-to-lim}
If $f:\RR\to\RR$ satisfies $\liminf\limits_{x\to\infty} f(x) / x\ge L$, then there exists $f_-:\RR\to\RR$ such that $f_-(x)\le f(x)$ for all $x\in\RR$ and $\lim\limits_{x\to\infty} f_-(x) / x = L$.

Similarly, if $f:\RR\to\RR$ satisfies $\limsup\limits_{x\to\infty} f(x) / x\le U$, then there exists $f_+:\RR\to\RR$ such that $f(x)\le f_+(x)$ for all $x\in\RR$ and $\lim\limits_{x\to\infty} f_+(x) / x = U$.
\end{lemma}

\begin{proof}[Proof of \myCref{lem-liminf-to-lim}]
It suffices to prove the first part; the second part can be proved by replacing $f$ with $-f$.

Since $\liminf_{x\to\infty} f(x) / x\ge L$, then there exists a strictly increasing sequence $\{x_k\}_{k=1}^{\infty}$ such that for each $k\in\ZZ_+$,
\[
\frac{f(x)}{x} \ge L - \frac{1}{k}, \qquad \forall x \ge x_k.
\]
Define $f_-:\RR\to\RR$ by
\[
f_-(x) =
\begin{cases}
f(x), &\quad\text{if } x < x_1 \\
(L-1/k)x, &\quad\text{if } x \in [x_k,x_{k+1})
\end{cases}.
\]
Clearly $f(x) \ge f_-(x)$ for all $x\in\RR$. To see that $\lim_{x\to\infty} f_-(x)/x=L$, take an arbitrary $k\in\ZZ_+$. Then for all $x\ge x_k$, there exists $k'\ge k$ such that $x\in[x_{k'},x_{k'+1}]$, which implies
\[
\left| \frac{f_-(x)}{x} - L \right| = \frac{1}{k'} \le \frac{1}{k},
\]
so
\[
\limsup_{x\to\infty} \left| \frac{f_-(x)}{x} - L \right| \le \frac{1}{k}.
\]
Since this holds for all $k\in\ZZ_+$, then we conclude $f_-(x)/x\to L$ as $x\to\infty$.
\end{proof}

\begin{lemma}[Empirical Bernstein bound]\label{lem-emp-Bernstein}
Suppose $\{y_i\}_{i=1}^k$ are i.i.d.~random variables taking values in $[a,b]$, and $k\ge 2$. Define the range $M=b-a$, mean $\mu=\EE[y_1]$, sample mean $\bar{y}=\frac{1}{k}\sum_{i=1}^k y_i$, and sample variance $\samplesd_k^2 = \frac{1}{k-1}\sum_{i=1}^n(y_i-\bar{y})^2$. For every $\delta\in(0,1)$, with probability at least $1-\delta$,
\[
\left| \bar{y} - \mu \right| \le \samplesd_k \sqrt{\frac{2\log(4/\delta)}{k}} + \frac{7M\log(4/\delta)}{3(k-1)}.
\]
\end{lemma}

\begin{proof}[Proof of \myCref{lem-emp-Bernstein}]
This is a direct consequence of Theorem 11 in \cite{MPo09}.
\end{proof}

\section{Details of Numerical Experiments}\label{sec-appendix-experiments}

\subsection{The OpinionQA Dataset}\label{sec-opinion}

\paragraph{Selection of survey questions.} The original dataset is categorized into topics such as health, crime/security, and political issues. Ideally, we would want to consider questions from the same category to ensure that they are similar enough. However, the category with the most questions has fewer than 200 questions. We thus consider pooling all questions. The dataset has 1,442 survey questions in total, which is too large for our computational resources. We selected a subset of questions as follows. First, while the number of choices ranges from 2 to 19, most questions have 5 choices. To give a fair comparison and for simplicity, we only consider questions with 5 choices. Second, not all questions have choices that can be clearly ordered in sentiments, such as the following one:
\begin{quote}
\textit{Who do you think has the most responsibility to reduce the amount of made-up news and information? 1. The government, 2. Technology companies, 3. The public, 4. The news media, 5. None of these, 6. Refused.}
\end{quote}
We asked GPT-4o to determine if a question's choices can be ordered in sentiments and we keep those that have GPT-4o's affirmative answer. This leaves us with 546 questions. To compensate for the loss of similarity by pooling questions across various topics and to further reduce our computational cost, we selected 400 questions that are ``most similar'' to each other by embedding the question statements using OpenAI's \texttt{text-embedding-3-small}, calculating the mean, and selecting the 400 questions with the smallest Euclidean distance to the mean. Out of these 400 questions, 15 questions have various issues with their choices by manual inspection, so we exclude them. This leaves us with 385 questions. All these questions happen to have at least 400 responses.

\paragraph{Example questions.} The questions in the OpinionQA dataset span a wide range of topics, including health, crime/security, and political issues. Some example questions are as follows:
\begin{itemize}
\item \textit{How much, if at all, do you think wages and incomes are contributing to your opinion about how the economy is doing? 
\begin{center}
1.~A great deal \quad 2.~A fair amount \quad 3.~Not too much \quad 4.~Not at all \quad 5.~Refused
\end{center}}
\item \textit{Regardless of whether you would want to move, how likely is it that you will move to a different community at some point in the future?
\begin{center}
1.~Very likely \quad 2.~Somewhat likely \quad 3.~Not too likely \quad 4.~Not at all likely \quad 5.~Refused
\end{center}}
\item \textit{How much, if anything, would you be willing to change about how you live and work to help reduce the effects of global climate change? Would you be willing to make:
\begin{center}
1.~A lot of changes \,  2.~Some changes \,  3.~Only a few changes \, 4.~No changes at all \,  5.~Refused
\end{center}}
\end{itemize}

\paragraph{Profiles.} Excluding surveyees with missing information, each of the 385 questions we consider has at least 400 responses. Since there was no information on the surveyees' identification, by dropping repeated profiles we can only say that there are at least 32,864 surveyees. Each surveyee is described by 12 features. Their corresponding categories are listed in \myCref{tab:opinionprofilefeatures}.

\begin{table}[h]
\TABLE
{Categories of Surveyees' Features in the OpinionQA Dataset. \label{tab:opinionprofilefeatures}}
{\begin{tabular}{|c|c|}
\hline
Feature & Options \\
\hline
US Citizenship &  `Yes', `No' \\
\hline
Region & `Northeast', `Midwest', `South', `West' \\
\hline
Sex & `Male', `Female' \\
\hline
Age & `18-29', `30-49', `50-64', `65+' \\
\hline
Marital Status & `Married', `Divorced', `Separated', `Widowed', `Never been married' \\
\hline
Race & `White', `Black', `Asian', `Hispanic', `Other' \\
\hline
Educational Background & \makecell{`Less than high school', `High school graduate', \\ `Some college, no degree', `Associate's degree', \\ `College graduate/some postgrad', `Postgraduate'}\\
\hline
Income & \makecell{`Less than \$30,000', `\$30,000-\$50,000', `\$50,000-\$75,000', \\ `\$75,000-\$100,000', `\$100,000 or more'} \\
\hline
Religious Affiliation & \makecell{`Protestant', `Roman Catholic', `Mormon', `Orthodox', \\ `Jewish', `Muslim', `Buddhist', `Hindu', \\ `Atheist', `Agnostic', `Other', `Nothing in particular'} \\
\hline
Religious Attendance & \makecell{`More than once a week', `Once a week', `Once or twice a month', \\ `A few times a year', `Seldom', `Never'} \\
\hline
Political Party & `Republican', `Democrat', `Independent', `Other' \\
\hline
Political ideology & \makecell{`Very conservative', `Conservative',\\ `Moderate', `Liberal', `Very liberal'} \\
\hline
\end{tabular}}
{}
\end{table}

\paragraph{Synthetic response generation.} We generate synthetic profiles by bootstrapping the 32,864 unique real profiles. We then generate synthetic answers by prompting LLMs to pretend that they are a surveyee with the synthetic profile and answer the question. An example prompt is as follows:
\begin{quote}
\textit{Pretend that you reside in the US and you are a US citizen from the West region of the country. You are female, your age is between 18 and 29, and you are single. In terms of race, you are white. In terms of education, you attended college but did not graduate. Your annual income is less than \$30,000. Religion-wise, you do not belong to any particular religion, and you never attend religious services. Politically, you are affiliated with a political party that is not Democratic or Republican, and you consider your political ideology to be liberal. Please answer the following question}: 

\textit{How much, if at all, do you think what happens to black people in the country overall affects what happens in your own life? [`1. A lot', `2. Some', `3. Not much', `4. Not at all', `5. Refused'].}

\textit{Please provide your answer choice (a single number from 1 to 5) in double square brackets.}
\end{quote}
In our experiments, LLMs usually directly gave answers in the required format, e.g., `[[2]]'.

\subsection{The EEDI Dataset}\label{sec-eedi}

\paragraph{Example questions.} Some example questions from the EEDI dataset are as follows:
\begin{itemize}
\item \textit{What number belongs in the box? $\square + 7 = 2$
\[
\text{A) 9} \quad \text{B) -5} \quad \text{C) -6} \quad \text{D) 5}
\]
}
\item \textit{If you multiply a square number by $9$, you get a square number. Is this statement:
\[
\text{A) always true} \quad \text{B) sometimes true} \quad \text{C) never true} \quad \text{D) impossible to say}
\]}
\item \textit{Which calculation is equal to $-20$?
\[
\text{A) } 2 \times (-2) - (-4) \times 4 \quad \text{B) } -28 - (-4) \times 2 \quad \text{C) } (-5)\top 2 + 5 \quad \text{D) } (-42) \div (-2) + 1
\]}
\end{itemize}

\paragraph{Profile distribution.} Excluding students with missing information which take up less than 10\% of the total population, there are 2,111 students who answered at least one of the 412 questions. Each student is described by three features: gender, age, and whether or not they are eligible for free school meals or premium pupil. Gender is represented by 1 or 2, where 1 corresponds to female and 2 corresponds to male. The students' ages are rounded to integers from 11 and 14. Whether or not a student is eligible for free school meals is represented by 0 or 1, where 0 corresponds to not eligible and 1 corresponds to eligible. The distribution of these students' features is presented in \myCref{tab:eediprofiledist}.

\begin{table}[h]
\TABLE
{Summary Statistics of Student Features in the EEDI Dataset.\label{tab:eediprofiledist}}
{\begin{tabular}{|c|c|c|c|c|c|}
\hline
 & Min & Max & Mean & Median & Standard Deviation \\
\hline
Gender & $1$ & $2$ & $1.4988$ & $1$ & $0.5001$ \\
\hline
Age & $11$ & $14$ & $11.2776$ & $11$ & $0.4696$ \\
\hline
Premium Pupil & $0$ & $1$ & $0.2842$ & $0$ & $0.4512$ \\
\hline
\end{tabular}}
{}
\end{table}

\paragraph{Synthetic response generation.} For each question, we generate synthetic profiles by sampling with replacement from the real profiles. We then generate synthetic answers by prompting LLMs to pretend that they are a student with the synthetic profile and answer the question. We adapted the prompt from \cite{HMG24} with slight modifications to reduce computational cost. An example prompt featuring an 11-year-old boy who is not eligible for free school meals is as follows:
\begin{quote}
\textit{Pretend that you are an 11-year-old student. Your gender is male. You are not eligible for free school meals or pupil premium due to being relatively financially advantaged. Given your characteristics, is it likely that you would be able to solve the following problem?}

\textit{Problem: [Insert question here]}

\textit{If yes, put the final answer choice (a single letter) in double square brackets. If you are likely to struggle with this problem, put a plausible incorrect answer choice (a single letter) in double square brackets.}
\end{quote}
An example answer from GPT-4o when given the second example question above is as follows:
\begin{quote}
\textit{As an 11-year-old student, I might have learned about square numbers and multiplication in school. However, the problem may be a bit tricky if I haven't thought about how multiplying square numbers by other numbers can also result in square numbers. I might not immediately realize that 9 is actually a square number itself (3 squared), which makes this property more evident.}

\textit{Considering this, I could find the reasoning challenging and decide based on a misconception. I might go with a plausible incorrect answer choice like [[B]] because I might think that it's only sometimes possible without realizing the full mathematical principle involved.}
\end{quote}

\section{Additional Experiment Results}\label{sec-appendix-experiments-results}

In this section, we provide additional experiment results for the simple method and the general method.

\subsection{Sharpness of the Simple Method}\label{sec-appendix-sharpness-simple}

In \myCref{tab-sharpness-simple-OpinionQA} and \myCref{tab-sharpness-simple-EEDI}, we report the $95$th percentile of the relative error $\big| \widehat{k} - k^*_{\test} \big| / k^*_{\test}$ over the $100$ train-test splits for the simple method on the OpinionQA and EEDI datasets, respectively. We observe that the selected sample size $\widehat{k}$ is consistently close to $k^*_{\test}$, verifying the sharpness of $\widehat{k}$.

\begin{table}[h]
\TABLE
{95th Percentile of $\big| \widehat{k} - k^*_{\test} \big| / k^*_{\test}$ for the Simple Method on the OpinionQA Dataset.\label{tab-sharpness-simple-OpinionQA}}
{\begin{tabular}{lcccc}
		\hline
$\alpha$ & $0.05$ & $0.1$ & $0.15$ & $0.2$ \\
        \hline
Claude 3.5 Haiku & $0.26$ & $0.23$ & $0.27$ & $0.25$ \\
DeepSeek-V3 & $0.45$ & $0.38$ & $0.26$ & $0.26$ \\
GPT-3.5 Turbo & $0.51$ & $0.23$ & $0.14$ & $0.13$ \\
GPT-4o mini & $0.31$ & $0.30$ & $0.17$ & $0.20$ \\
GPT-4o & $0.25$ & $0.23$ & $0.20$ & $0.20$ \\
GPT-5 mini & $0.29$ & $0.27$ & $0.23$ & $0.31$ \\
Llama 3.3 70B & $0.44$ & $0.15$ & $0.23$ & $0.27$ \\
Mistral 7B & $0.24$ & $0.31$ & $0.38$ & $0.25$ \\
Random & $0.20$ & $0.15$ & $0.13$ & $0.13$ \\
    \hline
    \end{tabular}}
{The $95$th percentiles are taken over $100$ random train-test splits of the survey questions.}
\end{table}

\begin{table}[h]
\TABLE
{95th Percentile of $\big| \widehat{k} - k^*_{\test} \big| / k^*_{\test}$ for the Simple Method on the EEDI Dataset.\label{tab-sharpness-simple-EEDI}}
{\begin{tabular}{lcccc}
        \hline
$\alpha$ & $0.05$ & $0.10$ & $0.15$ & $0.2$ \\
        \hline
Claude 3.5 Haiku & $0.16$ & $0.13$ & $0.13$ & $0.16$ \\
DeepSeek-V3 & $0.35$ & $0.30$ & $0.30$ & $0.30$ \\
GPT-3.5 Turbo & $0.11$ & $0.08$ & $0.08$ & $0.09$ \\
GPT-4o mini & $0.18$ & $0.13$ & $0.12$ & $0.08$ \\
GPT-4o & $0.28$ & $0.53$ & $0.36$ & $0.33$ \\
GPT-5 mini & $0.10$ & $0.08$ & $0.09$ & $0.08$ \\
Llama 3.3 70B & $0.17$ & $0.11$ & $0.12$ & $0.12$ \\
Mistral 7B & $0.08$ & $0.13$ & $0.09$ & $0.09$ \\
Random & $0.11$ & $0.17$ & $0.19$ & $0.19$ \\
    \hline
    \end{tabular}}
{The $95$th percentiles are taken over $100$ random train-test splits of the survey questions.}
\end{table}

\subsection{Experiment Results for the General Method}

In this section, we provide experiment results for the general method. Similar to \myCref{sec-experiments-results}, we will evaluate the following metrics:
\begin{enumerate}
\item miscoverage probability proxy $\coveragealt_{\test}(\widehat{k})$,
\item sharpness of $\widehat{k}$, through comparison with an oracle sample size $k_{\test}^* = \max\{k : \coveragealt_{\test}(k) \le \alpha \}$,
\item width of the selected confidence interval $\simCIBern(\widehat{k})$,
\item estimated hidden population size $\widehat{\kappa} = \widehat{k}/C$.
\end{enumerate}

\paragraph{Coverage validity.} In \myCref{fig-miscoverage-general}, we plot the miscoverage probability proxy $\coveragealt_{\test}(\widehat{k})$ for different LLMs on the two datasets. Similar to the case of the simple method, $\coveragealt_{\test}(\widehat{k})$ is consistently close to or below the target miscoverage level $\alpha$, which shows that the selected confidence intervals $\simCIBern(\widehat{k})$ have approximately $(1-\alpha)$ coverage probabilities. This verifies the coverage guarantee of our method in \myCref{thm-coverage}.

\begin{figure}[h]
	\FIGURE{
    \subcaptionbox{OpinionQA Dataset}
    {\includegraphics[width=0.48\linewidth]{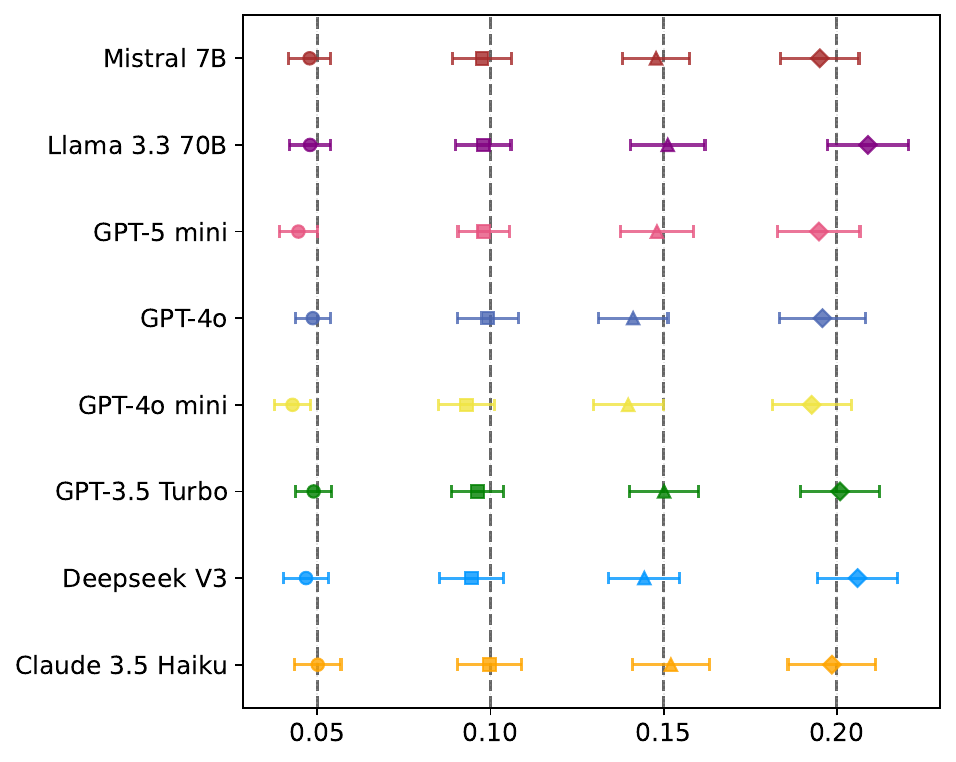}}
    \hfill\subcaptionbox{EEDI Dataset}
    {\includegraphics[width=0.48\linewidth]{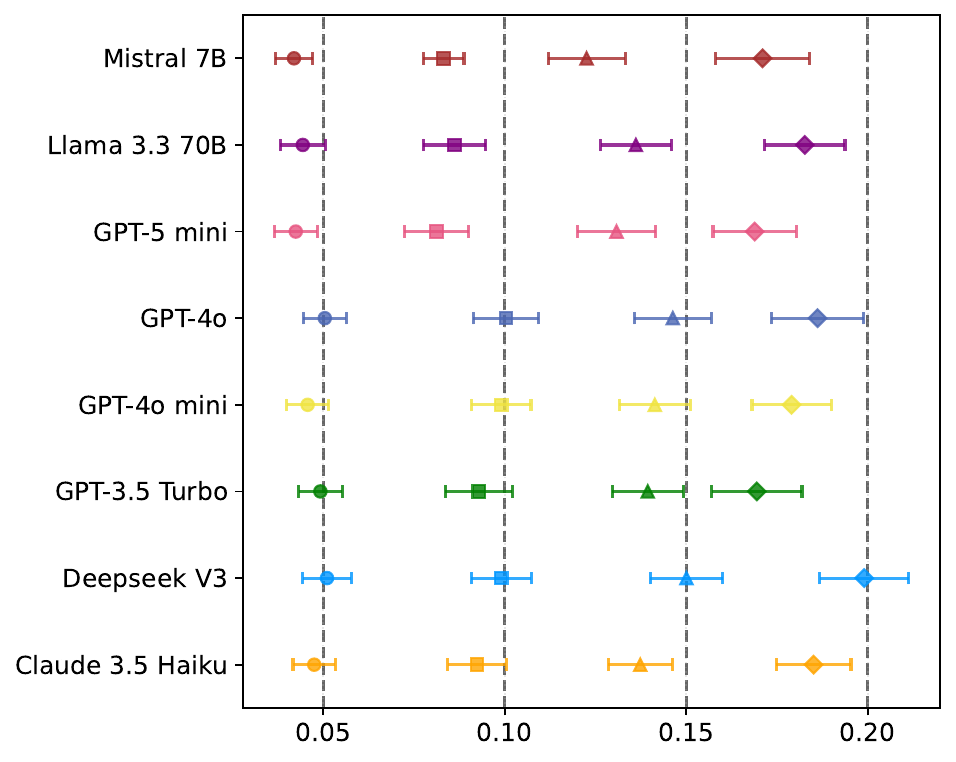}}
	}
	{Miscoverage Probability Proxy $\coveragealt_{\test}(\widehat{k})$ for the General Method. \label{fig-miscoverage-general}}
	{Horizontal axis: target miscoverage level $\alpha$. Vertical axis: LLM. Circles, squares, triangles and diamonds represent $\coveragealt_{\test}(\widehat{k})$ for $\alpha=0.05,0.1,0.15,0.2$, respectively. The results are averaged over $100$ random train-test splits of the questions. The half-width of each error bar is $1.96$ times the standard error.}
\end{figure}

\paragraph{Sharpness of selected sample size.} To verify the sharpness of the sample size $\widehat{k}$ selected by the general method, we will compare it with an oracle sample size $k_{\test}^* = \max\{k \in \ZZ_+ : \coveragealt_{\test}(k) \le \alpha \}$, which is the maximum sample size that guarantees $(1-\alpha)$ proxy coverage over the testing set of survey questions. \myCref{tab-sharpness-general-OpinionQA} and \myCref{tab-sharpness-general-EEDI} report the $95$th percentile of the relative error $\big| \widehat{k} - k^*_{\test} \big| / k^*_{\test}$ over the $100$ train-test splits on the OpinionQA and EEDI datasets, respectively. We observe that $\widehat{k}$ is consistently close to $k^*_{\test}$, verifying its sharpness.

\begin{table}[h]
\TABLE
{95th Percentile of $\big| \widehat{k} - k^*_{\test} \big| / k^*_{\test}$ for the General Method on the OpinionQA Dataset.\label{tab-sharpness-general-OpinionQA}}
{\begin{tabular}{lcccc}
		\hline
$\alpha$ & $0.05$ & $0.1$ & $0.15$ & $0.2$ \\
        \hline
Claude 3.5 Haiku & $0.25$ & $0.22$ & $0.27$ & $0.24$ \\
DeepSeek-V3 & $0.44$ & $0.36$ & $0.26$ & $0.25$ \\
GPT-3.5 Turbo & $0.47$ & $0.24$ & $0.14$ & $0.12$ \\
GPT-4o mini & $0.31$ & $0.28$ & $0.17$ & $0.22$ \\
GPT-4o & $0.23$ & $0.23$ & $0.19$ & $0.21$ \\
GPT-5 mini & $0.26$ & $0.25$ & $0.23$ & $0.31$ \\
Llama 3.3 70B & $0.42$ & $0.15$ & $0.23$ & $0.27$ \\
Mistral 7B & $0.25$ & $0.30$ & $0.37$ & $0.25$ \\
Random & $0.19$ & $0.15$ & $0.13$ & $0.12$ \\
    \hline
    \end{tabular}}
{The $95$th percentiles are taken over $100$ random train-test splits of the survey questions.}
\end{table}

\begin{table}[h]
\TABLE
{95th Percentile of $\big| \widehat{k} - k^*_{\test} \big| / k^*_{\test}$ for the General Method on the EEDI Dataset.\label{tab-sharpness-general-EEDI}}
{\begin{tabular}{lcccc}
        \hline
$\alpha$ & $0.05$ & $0.1$ & $0.15$ & $0.2$ \\
        \hline
Claude 3.5 Haiku & $0.15$ & $0.13$ & $0.11$ & $0.16$ \\
DeepSeek-V3 & $0.35$ & $0.26$ & $0.25$ & $0.30$ \\
GPT-3.5 Turbo & $0.11$ & $0.08$ & $0.09$ & $0.09$ \\
GPT-4o mini & $0.18$ & $0.16$ & $0.12$ & $0.09$ \\
GPT-4o & $0.29$ & $0.50$ & $0.38$ & $0.32$ \\
GPT-5 mini & $0.07$ & $0.08$ & $0.09$ & $0.09$ \\
Llama 3.3 70B & $0.14$ & $0.12$ & $0.12$ & $0.12$ \\
Mistral 7B & $0.12$ & $0.09$ & $0.09$ & $0.10$ \\
Random & $0.12$ & $0.15$ & $0.17$ & $0.18$ \\
    \hline
    \end{tabular}}
{The $95$th percentiles are taken over $100$ random train-test splits of the survey questions.}
\end{table}

\paragraph{Interval width.} In \myCref{fig-width-general}, we plot the widths of the selected confidence intervals $\simCIBern(\widehat{k})$ for different LLMs and target miscoverage levels $\alpha$. The results are similar to those of the simple method. In simulating social opinions (\myCref{fig-width-simple-OpinionQA}), GPT-4o has the smallest misalignment gap; in simulating middle-school student answers to mathematics questions, all LLMs yield wide confidence intervals, reflecting large misalignment gaps.

\begin{figure}[H]
	\FIGURE{
    \subcaptionbox{OpinionQA Dataset \label{fig-width-general-OpinionQA}}
    {\includegraphics[width=0.46\linewidth]{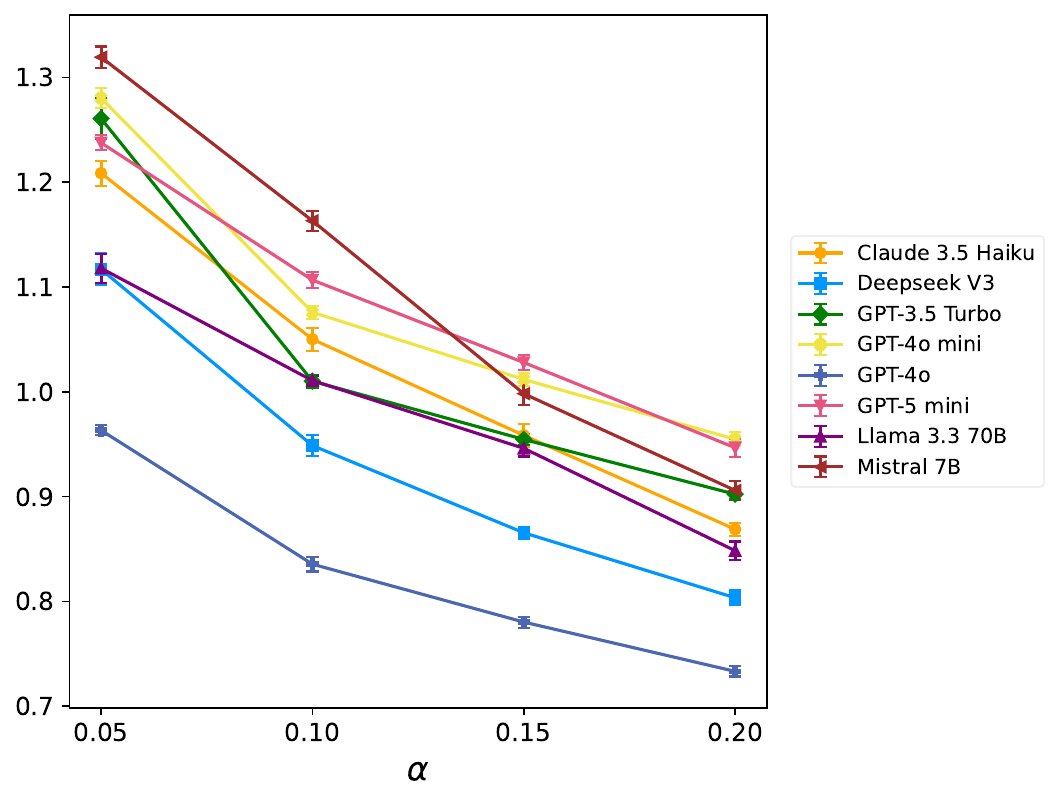}}
    \hfill\subcaptionbox{EEDI Dataset \label{fig-width-general-EEDI}}
    {\includegraphics[width=0.46\linewidth]{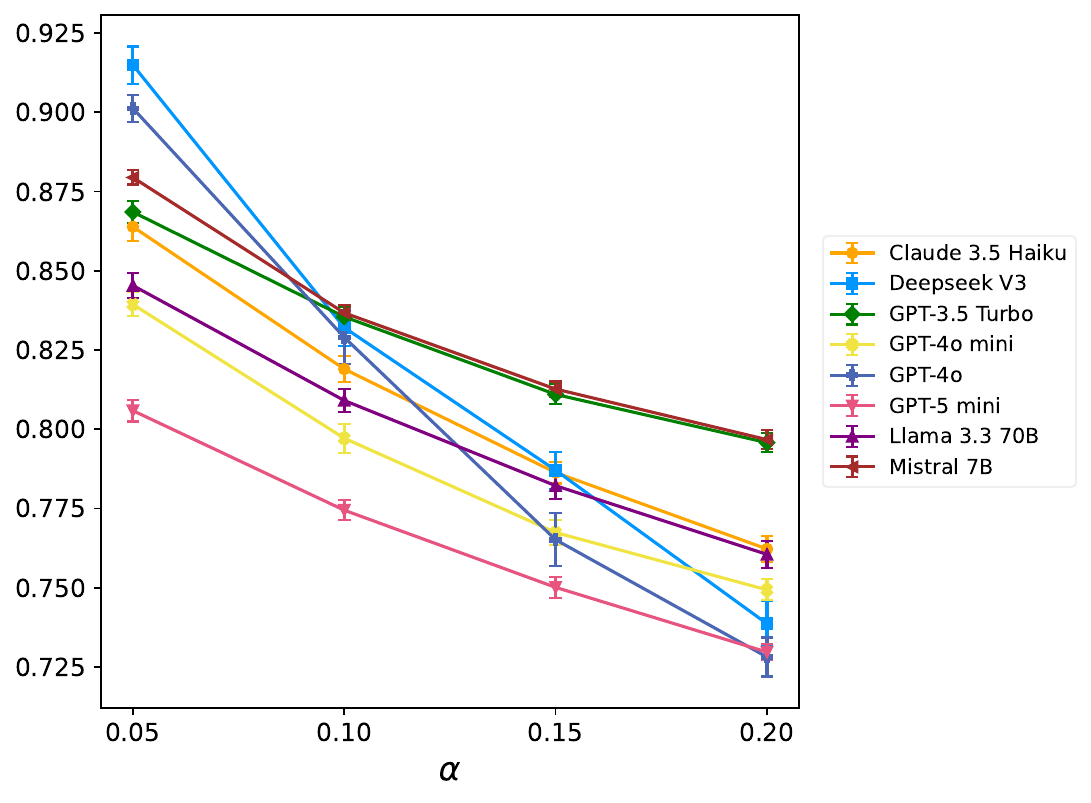}}
	}
	{Widths of Confidence Intervals $\simCIBern(\widehat{k})$ for the General Method. \label{fig-width-general}}
	{Horizontal axis: target miscoverage level $\alpha$. Vertical axis: width of $\simCIBern(\widehat{k})$. The results are averaged over $100$ random train-test splits. The half-width of each error bar is $1.96$ times the standard error.}
\end{figure}

\paragraph{Estimated hidden population size $\widehat{\kappa} = \widehat{k}/C$.} In \myCref{fig-kappa-general}, we report the estimated hidden population size $\widehat{\kappa} = \widehat{k}/C$ for different LLMs on the OpinionQA and EEDI datasets, under $\alpha=0.05$. The results are similar to those for the simple method.

\begin{figure}[H]
	\FIGURE{
    \subcaptionbox{OpinionQA Dataset}
    {\includegraphics[width=0.4\linewidth]{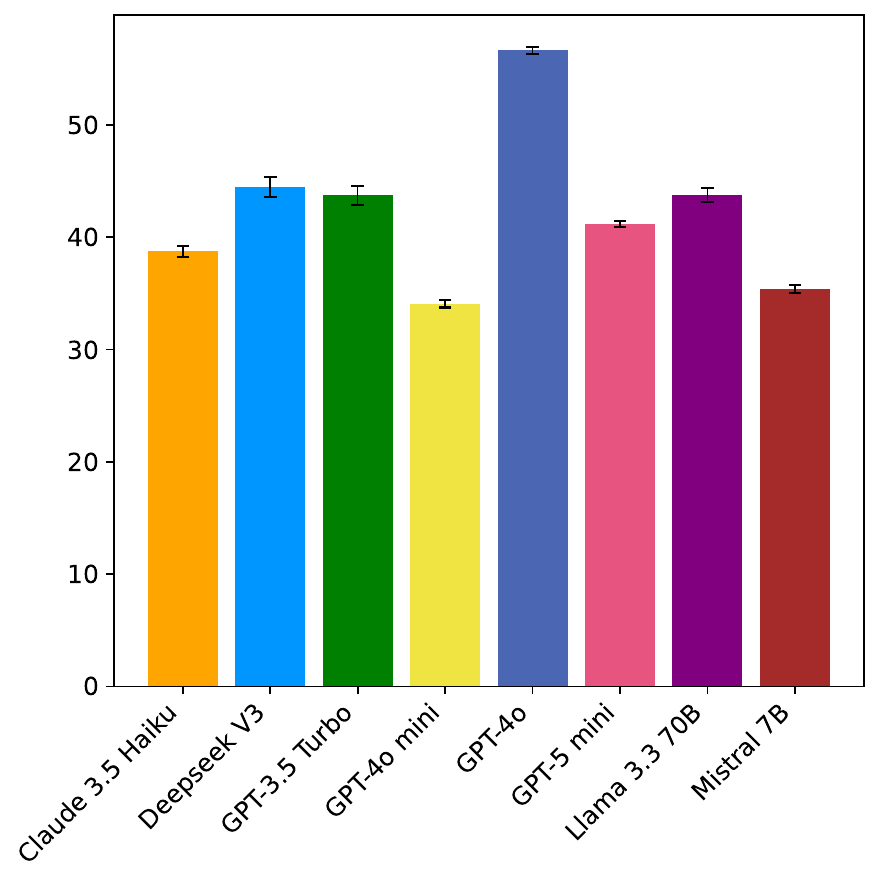}}
    \hfill\subcaptionbox{EEDI Dataset}
    {\includegraphics[width=0.4\linewidth]{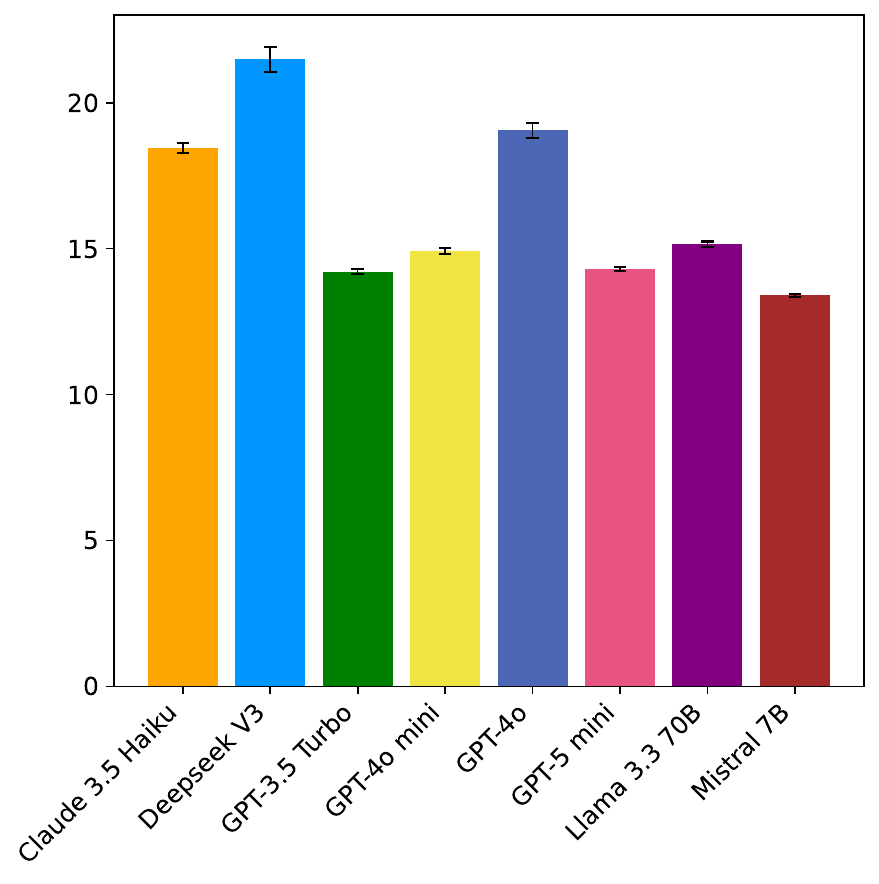}}
	}
	{Estimated Hidden Population Sizes $\widehat{\kappa}=\widehat{k}/C$ of Different LLMs from the General Method. \label{fig-kappa-general}}
	{The results are averaged over $100$ random train-test splits. The half-width of each error bar is $1.96$ times the standard error.}
\end{figure}

\subsection{Ablation Studies for the Dilation Factor}\label{sec-ablation-C}

In the main experiments, we set the dilation factor to $C=2$. We now examine the sensitivity of the estimated hidden population size $\widehat{\kappa}=\widehat{k}/C$ to this choice. We conduct ablation studies with GPT-5 mini on the OpinionQA dataset, varying the dilation factor $C$ and recomputing $\widehat{k}/C$ over $100$ random train-test splits.

\myCref{fig-ablation-C} shows that the estimated hidden population size is stable across the values of $C$ considered, for both the simple and general methods. In particular, the average value of $\widehat{k}/C$ varies by less than $2$ as $C$ changes. This stability is consistent with the theoretical interpretation in \myCref{sec-effective-sample-size}, where the normalized sample size $\widehat{k}(C)/C$ estimates the LLM's effective human sample size as $C\to\infty$.

\begin{figure}[H]
	\FIGURE{
    \subcaptionbox{Simple Method \label{fig-ablation-simple}}
    {\includegraphics[width=0.4\linewidth]{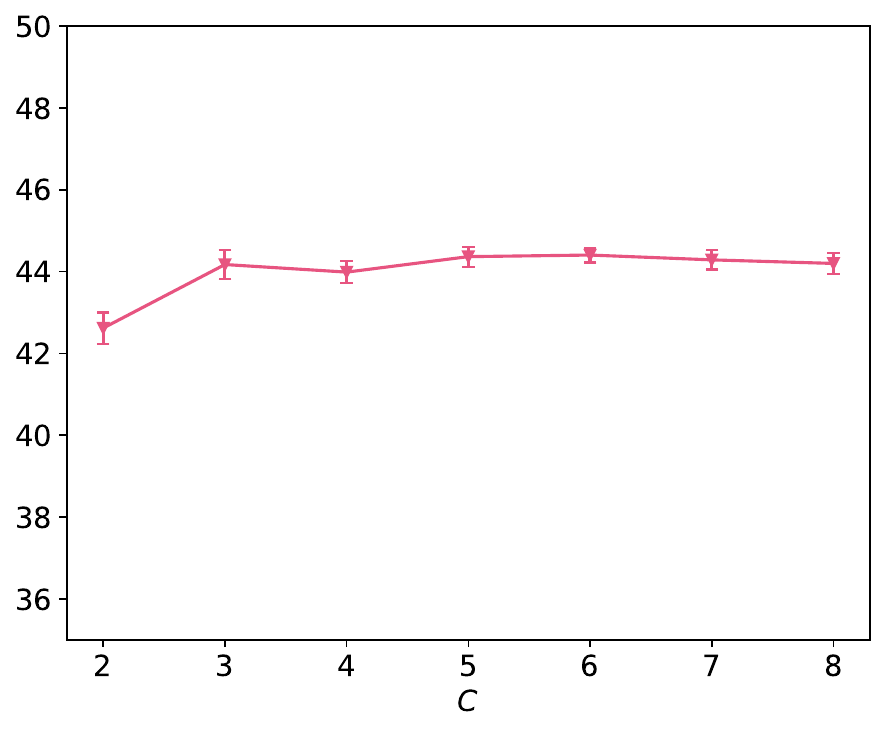}}
	\hfill\subcaptionbox{General Method \label{fig-ablation-general}}
    {\includegraphics[width=0.4\linewidth]{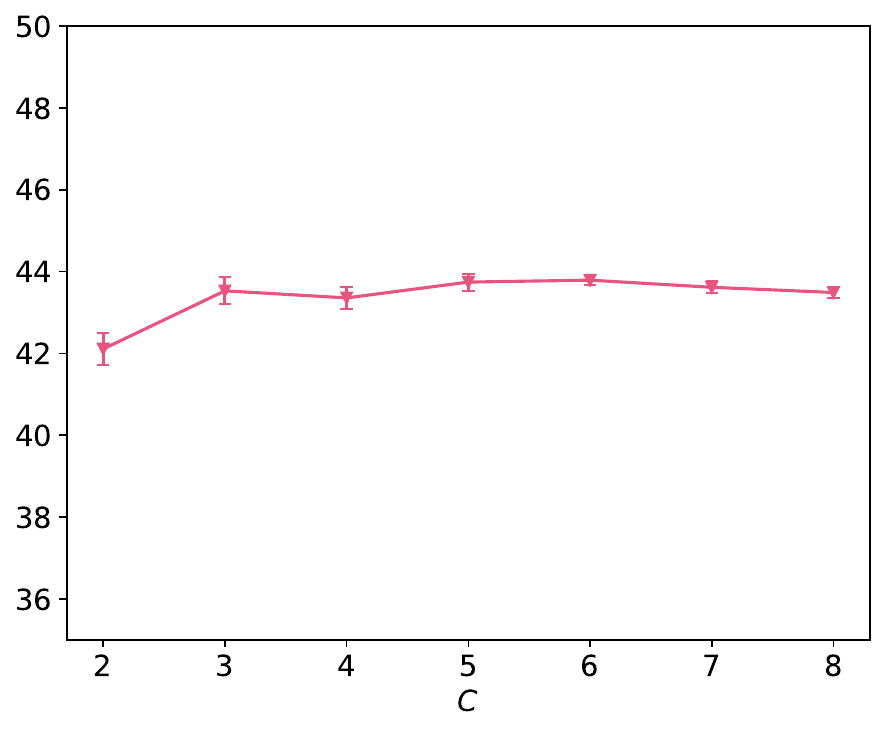}}
	}
	{Ablation Studies for the Dilation Factor $C$. \label{fig-ablation-C}}
	{Horizontal axis: dilation factor $C$. Vertical axis: estimated hidden population size $\widehat{k}/C$ of GPT-5 mini. The results are averaged over $100$ random train-test splits. The half-width of each error bar is 1.96 times the standard error.}
\end{figure}

\newpage

{
\bibliographystyle{ims}
\bibliography{bib}
}

\end{document}